\documentclass[twoside,11pt]{article}
\usepackage{jmlr2e}
\usepackage{amsmath}
\usepackage{float}

\usepackage{algorithm}
\usepackage{algorithmicx}
\allowdisplaybreaks[4]
\usepackage{algpseudocode}
\renewcommand{\algorithmicrequire}{\textbf{Input:}}  
\renewcommand{\algorithmicensure}{\textbf{Output:}} 

\ShortHeadings{Decentralized Policy Gradient  Nash Equilibria Learning of Stochastic Games}{Chen and Li}
\firstpageno{1}

\begin{document}

\title{Decentralized Policy Gradient for Nash Equilibria Learning  of General-sum Stochastic Games}

\author{\name  Yan Chen \email YanChen@stu.ecnu.edu.cn \\
       \addr School of Mathematical Sciences\\
       East China Normal University\\
       No. 5005, South Lianhua Road,  Shanghai 200241, China
       \AND
       \name  Tao Li \email tli@math.ecnu.edu.cn\\
       \addr School of Mathematical Sciences\\
       East China Normal University\\
       No. 5005, South Lianhua Road,  Shanghai 200241, China}
\maketitle

\begin{abstract}
We study Nash equilibria learning of a general-sum stochastic game with  an unknown transition probability density
function. Agents take actions at  the current environment state  and their joint action influences the transition of the environment state and their
immediate rewards. Each agent only observes the environment state
and its own immediate reward and is  unknown about the actions or immediate rewards of others. We introduce the concepts of weighted asymptotic Nash equilibrium with probability $1$ and in probability. For the case with exact pseudo gradients, we design a two-loop algorithm  by the equivalence of Nash equilibrium  and  variational inequality problems.
 In the outer loop, we sequentially update a constructed strongly monotone variational inequality by updating a proximal parameter while employing a single-call extra-gradient algorithm in the inner loop for solving the constructed variational inequality. We show that if the associated Minty variational inequality has a solution, then the designed algorithm  converges to the $k^{\frac{1}{2}}$-weighted asymptotic  Nash equilibrium.  Further, for the case with  unknown pseudo gradients, we propose a decentralized algorithm, where the G(PO)MDP gradient estimator of the pseudo gradient is provided by  Monte-Carlo  simulations. The convergence to the  $k^{\frac{1}{4}}$-weighted asymptotic Nash equilibrium in probability is achieved.
\end{abstract}


\section{Introduction}

In a Markov decision process,  an agent aims at finding a policy to maximize its own expectation of cumulative discounted immediate rewards. At any given environment state, an agent chooses a policy and takes an action. Then, the agent gains an immediate reward  and the action causes the environment a transition to the next state. In a multi-agent scenario, the decision-making of an agent is affected  not only by the environment state, but also by the behaviors of other agents. Game theory studies how multiple agents make decisions when they interact directly and dynamically, and how these decisions reach an equilibrium. Motivated by which, Shapley came up with the framework of stochastic games ~\citep{Shapley:53}, also known as Markov games ~\citep{Littman:94}, which allows people to extend the Markov decision process to the case with multiple agents. In a stochastic game, each agent take its action independently and both state transitions of the environment and immediate rewards  depend on agents' actions. Quantitative classical algorithms, such as iterated algorithms ~\citep{Shapley:53,Hoffman:66} and mathematical programming ~\citep{FilalJA91covering,KVrieze}, have been proposed for finding Nash equilibria of general-sum stochastic games when agents know the transition probability density function of the environment state. For the case with the unknown state transition probability density function, \cite{Littman:94} introduced multi-agent reinforcement learning as a framework and applied Q-Learning to a simple two-player zero-sum  stochastic game. Since then,  multi-agent reinforcement learning algorithms, like Q-Learning ~\citep{Christopher} and Actor-Critic ~\citep{Konda}, have been widely used. These algorithms allow agents to act while learning  if  the transition probability density of the environment state  is unknown.  \cite{Hu:98} proposed a Q-Learning algorithm of two-player general-sum stochastic games, where they constructed a bimatrix game based on the current estimation of action-value function and computed the Nash equilibrium for the bimatrix game in each iteration. It is shown that the Q-Learning algorithm is convergent  if every bimatrix game arising from the learning process has a global optimum point or a saddle point and agents
update the estimations of action-value functions according to values at this point.
 \cite{Hu2} extended the result to a multi-agent context and designed an improved algorithm: Nash Q-Learning (NashQ).
 \cite{Littman2} presented Friend or Foe Q-Learning (FFQ), where agents assume its opponent as either a friend or foe, which can be considered as an extension of NashQ. They showed that FFQ learns the action-value function under a Nash equilibrium if there exists an adversarial equilibrium or coordination equilibrium of the stochastic game. \cite{Greenwald03covering} generalized NashQ and FFQ
and introduced Correlated Q-Learning. They demonstrated the convergence by simulations but didn't give the theoretical proof. \cite{Prasad15covering} designed an actor-critic algorithm for Nash equilibria of general-sum stochastic games from the view of dynamic programming. The critic updates the state-value functions and the actor performs gradient descent on policies. They established that the algorithm converges to a  Nash equilibrium asymptotically.
 \cite{Perolat} built a stochastic approximation for a
fictitious play process using an Actor-Critic algorithm. They proved the convergence of the method towards a Nash
equilibrium for both cases with  two-player zero-sum
 and cooperative (that is, when players receive the same immediate rewards) stochastic games.

All algorithms mentioned above generated by Q-Learning and Actor-Critic are based on the estimations of action-value functions. They learn the estimations of action-value functions and then choose their actions according to the estimations. It's impossible to design  policies without the estimations of action-value functions  for these algorithms ~\citep{Richard}. In particular, these algorithms  generated by Q-Learning encounter the computation difficulty of Nash equilibria of stage games which makes the algorithms implementation more difficult. At the same time, these algorithms need huge tables to store the estimations of action-value functions. So they can't deal with large action spaces or a continuum of action spaces. To this end,  \cite{Richard} introduced policy gradient methods. Firstly, policies of each agent are parameterized and then policy parameters are updated according to the gradient of  state-value functions with respect to  parameters. Policy gradient methods learn parameterized policies directly and action selections no longer depend directly on action-value functions. At this point,  action-value functions can still be used to learn the parameters of policies but are unnecessary for action selections. By virtue of this advantage, policy gradient methods provide a practical way to handle the stochastic games where the action space is massive or even a continuum.
At the same time, policy gradient methods search directly in the parameter space. As a result, they enjoy better theoretical convergence guarantees
~\citep{Yang,Zhang,Agarwal}. Commonly used classes of  parameterized policies include direct parameterization, $\alpha$-greedy direct parameterization, Gaussian parameterization, softmax parameterization,  log-linear parameterization and so on. In recent years, with the great success of the theoretical research on policy gradient methods for a Markov decision process, policy gradient methods with different classes of  parameterized policies are also applied to zero-sum and general-sum stochastic games.
For a two-agent zero-sum stochastic game,  \cite{Daskalakis} focused on  $\alpha$-greedy  direct parameterization and used the REINFORCE gradient estimator of the policy gradient. They showed that if
 agents descend and ascend along the gradients of total reward functions with respect to their parameters respectively, their policies converge to a min-max equilibrium of the game, as long as their learning rates follow a two-timescale rule.
\cite{Zhao} considered softmax parameterization. They found a minimax equilibrium for
the matrix game constructed by the estimations of state-value functions and then
performed natural policy gradient descent to update the estimations  in each iteration.
They proved the algorithm can find a near-optimal policy. \cite{Wei} used direct parameterization and updated the policies of two agents by running an optimistic gradient descent or ascent algorithm  with a critic that slowly learns the state-value function of each state. They showed that the algorithm converges to the set of Nash equilibria if the induced discounted Markov chain under any stationary policies is irreducible. For general-sum stochastic games,  \cite{ Stefanos} defined the notion of Markov potential games by transplanting potential games into the setting of Markov games. They took $\alpha$-greedy direct parameterization and replaced
 the actual gradients of total reward  functions with respect to parameters with  REINFORCE gradient estimators. The convergence to the $\epsilon$-Nash equilibrium of the gradient ascent algorithm of parameters is given.   \cite{Runyu} considered direct parameterization and showed that Nash equilibria and first-order  Nash equilibria of general-sum stochastic games are equivalent. They gave the  rate of converging  to strict Nash equilibria if agents perform gradient descent for the case with  known gradients of state-value functions with respect to policy parameters. Also, They  gave
global convergence rates for both exact gradients and gradients estimated by samples of Markov potential games.  \cite{Mao} proposed a decentralized algorithm in which each agent independently uses an
optimistic V-learning and performs a  mirror descent for policy updating with direct  parameterization. Their algorithm converges to  a coarse correlated
equilibrium, a solution concept that generalizes Nash equilibrium by allowing possible correlations among the agents' policies.

It is worth noting that the results of ~\cite{Daskalakis}, ~\cite{Stefanos} and ~\cite{Runyu} depend on the fact that one can find  Nash equilibria by the first-order necessary optimality conditions of  total reward functions with respect to the parameter of each player if  the total reward functions satisfy the gradient domination theorem. Inspired by the above research, ~\cite{Daskalakis}, ~\cite{Stefanos} and ~\cite{Runyu} find  Nash equilibria of zero-sum stochastic  games and  Markov potential games with unknown transition probability density functions. While for general-sum stochastic games, ~\cite{Runyu} is restricted to the case where the state and action spaces are finite, the gradients of total reward   functions with respect to parameters are known to each agent, and policies are with direct parameterization. In this paper, we study learning Nash equilibria of a general-sum stochastic game with an unknown transition probability density function. The joint actions of agents influence the state transition of the environment and their
immediate rewards. Each agent only observes  states
and their immediate rewards and is unaware of the actions or rewards of  other agents. Compared with ~\cite{Daskalakis}, ~\cite{Stefanos} and ~\cite{Runyu}, we study the general-sum stochastic game in which the state and action spaces are compact and convex and the transition probability density function is unknown
to all agents. Focusing on the equivalence between  Nash equilibrium and  variational inequality problems, we propose algorithms for learning Nash equilibria. We introduce the concepts of weighted asymptotic Nash equilibrium with probability $1$ and in probability and illuminate the connection between these concepts. It is shown that the algorithms  converge to the  weighted asymptotic Nash equilibrium for the case of exact gradients and the  weighted asymptotic Nash equilibrium in probability  for the case with unknown gradients respectively. ~\cite{Stefanos} and ~\cite{Runyu}  updated policy parameters by gradient ascent and they showed the convergence of the algorithm using non-convex optimization by the existing of potential functions in Markov potential games. While for the general-sum stochastic games, there are no longer potential functions, thus, one  can not ensure the convergence  if  updating policy parameters by gradient ascent. Different from ~\cite{Stefanos} and ~\cite{Runyu}, we consider  the variational inequality problem which is equivalent to the Nash equilibrium problem and design a two-loop algorithm. For the case with exact pseudo gradients, we design a two-loop algorithm in which  we sequentially update a constructed strongly monotone variational inequality in the outer loop by updating a proximal parameter and  employ a single-call extra-gradient algorithm  in the inner loop for solving the constructed variational inequality. As a consequence, it is possible to employ the  variational inequality to establish the convergence to the
 $k^{\frac{1}{2}}$-weighted asymptotic  Nash equilibrium of our algorithm if  the related Minty variational inequality has  a solution.
 For the case with unknown pseudo gradients, \cite{Daskalakis} assumes that agents negotiate learning rates at the beginning of a  zero-sum  stochastic game, which leads to an incompletely decentralized algorithm. While it's unnecessary to negotiate learning rates in advance, which cuts down the communication cost and therefore our algorithm is completely decentralized. Agents estimate  pseudo gradients by interacting with the environment in the scenario of unknown pseudo gradients. The G(PO)MDP gradient estimator with a single trajectory has a high variance. Therefore, different from ~\cite{Daskalakis},
 we make use of the average of multiple trajectories. When agents interact with the environment, although the unbiased G(PO)MDP estimator of the pseudo gradient samples from an infinite time horizon, the Monte Carlo simulations are not feasible to sample in an infinite time horizon, so we adopt the G(PO)MDP estimator of a finite time horizon ~\citep{Tianyi,SLu}. The errors between the estimated pseudo gradient and the real pseudo gradient are analysed and we establish the convergence to the
 $k^{\frac{1}{4}}$-weighted asymptotic  Nash equilibrium in probability.

 The remainder of this paper is organized as follows. In Section II, the stochastic game problem is formulated. In Section III, we present the equivalence between  Nash equilibrium  and variational inequality problems
 and the existence of Nash equilibrium. In Section IV, we propose the algorithm for learning Nash equilibria and analyse the convergence of the algorithm for the case with  exact pseudo gradients. In Section V, for the case with  unknown pseudo gradients,
 the algorithm for learning Nash equilibria is given by introducing the G(PO)MDP gradient estimator and the convergence of the algorithm is showed. In Section VI, numerical examples are given to illustrate our algorithms. In Section VII, conclusion and future research topics are given.

The following notations will be used throughout this paper: For a given vector $x=(x_{i})_{i}$, $\|x\|=\sqrt{\sum_{i} x_{i}^{2}}$ denotes its Euclidean norm and $|x|=\sum_{i} |x_{i}|$ denotes its $L^{1}$-norm. $\mathbb{E}\left[X\right]$ denotes the expectation of stochastic variable $X$. $\left\langle \cdot, \cdot \right\rangle$ denotes the inner product of vectors in Euclidian space. $\mathbf{0}$ denotes vector $0$. $\delta(\cdot)$ denotes Dirac delta function.

\section{Problem Formulation}
A stochastic game is an extension of a  Markov decision process of a single agent. A stochastic game is denoted by a tuple $\Gamma=(\mathcal{N},\mathcal{S},${\Large$\times$}$ _{i=1}^{N}\mathcal{A}_{i},\rho,(r_{i})_{i=1}^{N},\gamma)$, where
\begin{itemize}
\item[(i)]$\mathcal{N}=\{1,2, \ldots, N\}$ is the set of agents and its cardinality is $N$;
\item[(ii)]$\mathcal{S}$ is the environment state space which is observed by all agents and  is compact in $ \mathbb{R}^{q}$, where $q$ is a positive integer;
\item[(iii)]{\Large$\times$}$ _{i=1}^{N}\mathcal{A}_{i}$ is the joint action space, where $\mathcal{A}_{i}$ is the action space of agent $i$, which is compact in $\mathbb{R}^{m_{i}}$, where $m_{i}$ is a positive integer;
\item[(iv)]$\rho_{t,t+1}(s^{\prime}|s,a)$ is the conditional transition probability density function of the environment state from $s(t)=s$ to $s(t+1)=s^{\prime}$ if agents take a joint action $a(t)=a=(a_{i})_{i=1}^{N}$ at time $t$, that is,
\begin{equation}\nonumber
\begin{aligned}
&\rho_{t,t+1}\left(s^{\prime} \mid s, a\right) \geqslant 0,\forall\  s, s^{\prime} \in \mathcal{S}, \forall\  a \in \mathcal{A},\int_{\mathcal{S}} \rho_{t,t+1}\left(s^{\prime} \mid s, a\right) \mathrm{d} s^{\prime}=1,
\end{aligned}
\end{equation}
and, especially, the conditional transition probability density function degenerates to a Dirac delta function if the state space is discrete, and so without loss of generality, we only focus on a continuum of state space;
\item[(v)]$r_{i}(s,a)$ is the immediate reward function of agent $i$ if the environment state is $s$ and the joint action $a$ is chosen by the agents;
\item[(vi)]$\gamma \in (0,1)$ is the
discount factor that describes the influence of the rewards obtained
in the future on the agents' policy \citep{Prasad15covering}.
\end{itemize}

For any given state $s$, a policy of agent $i$ is defined as a conditional probability density function over $\mathcal{A}_{i}$. The admissible policy set of agent $i$ is defined as
$$
\Pi_{i} \triangleq \left\{\pi_{i}(\cdot \mid s), \forall\ s \in \mathcal{S} \biggm| \int_{\mathcal{A}_{i}} \pi_{i}(a_{i}^{\prime} \mid s)\
 \mathrm{d} a_{i}^{\prime}=1,\ \pi_{i}(a_{i} \mid s) \geq 0, \forall  \ a_{i} \in \mathcal{A}_{i}\right\}.
$$
We assume that the policies of all agents are stationary, that is, the policies are independent of time. The set of joint policies of all agents is denoted by $\Pi=${\Large$\times$}$ _{i=1}^{N}\Pi_{i}$. A joint policy is denoted by $\pi\left(a \mid s\right)=\prod_{i=1}^{N} \pi_{i}\left( a_{i}\mid s\right)$, $\ a=(a_{i})_{i=1}^{N} \in \mathcal{A}$, $s \in \mathcal{S}$.

Let $\rho^{\pi}$ denote the Markov kernel of the Markov chain induced by the policy $\pi$, that is, for any state $s(t)=s$ at  time $t$ and any state $s(t+1)=s^{\prime}$ at time $t+1$, it follows that
$\rho^{\pi}_{t,t+1}\left(s^{\prime} \mid s\right)=\int_{\mathcal{A}} \rho_{t,t+1}\left(s^{\prime} \mid s, a\right) \pi(a|s) \mathrm{d}a$.
We denote the probability density function of the initial state $s(t)$ by $\rho_{t}$.

We consider more generalized  policy parameterization with stochastic parameters. Let stochastic parameter $\theta_{i}$ be a stochastic variable on some probability space $(\Omega,{\cal F},P)$ taking values in $\Theta_{i} \subseteq \mathbb{R}^{d_{i}}$, where $d_{i}$ is a positive integer. Let $\theta=\left(\theta_{i}\right)_{i=1}^{N} \in \Theta=${\Large$\times$}$ _{i=1}^{N}\Theta_{i}$$\subseteq \mathbb{R}^{\sum_{i=1}^{N}
d_{i}}$. Let $\pi_{\theta}(a \mid s)=\prod_{i=1}^{N}
\pi_{\theta_{i}}(a_{i} \mid s)$ denote the joint policy.

The state-value function $V_{i}^{\pi_{\theta}}(s)$ of agent $i$ is defined as the conditional expectation of the discounted sum of immediate rewards starting from the initial state $s(t)=s$ at time $t$ by choosing actions according to the policy $\pi_{\theta}$, that is,
\begin{align}\label{valuefunction}
V_{i}^{\pi_{\theta}}(s,t)=\mathbb{E}\left[\sum_{l=t}^{\infty} \gamma^{l-t} r_{i}\left(s\left(l\right),a\left(l\right) \right)\biggm|s(t)=s, \theta\right] .
\end{align}

The total reward function $J_{i}^{\pi_{\theta}}(t)$ of agent $i$ is defined as the conditional expectation of the  discounted sum of immediate rewards  starting from the initial state $s(t)$  at time $t$ with the probability density function $\rho_{t}$ by choosing actions according to the policy $\pi_{\theta}$, that is,
\begin{align}\label{totalreward}
J_{i}^{\pi_{\theta}}(t)=\mathbb{E}\left[
\sum_{l=t}^{\infty} \gamma^{l-t} r_{i}\left(s\left(l\right),a\left(l\right) \right)\biggm|\theta\right].
\end{align}

It is known that $J_{i}^{\pi_{\theta}}(t)$ is a Borel measurable function of $\theta$. For notational convenience, we denote $V_{i}^{\pi_{\theta}}(s,t)$ by $V_{i}^{\theta}(s,t)$, $J_{i}^{\pi_{\theta}}(t)$ by $J_{i}^{\theta}(t)$ and $\rho^{\pi_{\theta}}$ by $\rho^{\theta}$.

At each time $l=t, t+1,\ldots$, agents take an action $a(l) \in \mathcal{A}$ according to the policy $\pi_{\theta}$ given the current state $s(l) \in S$ observed by all agents. Then, agent $i$ gains an immediate reward $r_{i}(s(l),a(l))$, $i=1,2,\ldots, N,$ and the state transitions to the next state
$s(l+1) \in S$. Each agent aims at maximizing its total reward function.

In our model, we assume that agent $i$ observes its own immediate reward and is completely unknown about the rewards and actions of other agents. The transition probability density function is unknown to all agents.

We study  learning Nash equilibria of the general-sum stochastic game with parameterized policies. The Nash equilibrium of the game is defined as follows.

\begin{definition}
\citep{NashJr} (Nash equilibrium)\label{Nashdf}
For the game $\Gamma$, if there exists $\theta^{*}=(\theta_{i}^{*})_{i=1}^{N}\\ \in \Theta$ such that
$$
\sup _{i \in \mathcal{N}}\Big(\sup _{\theta_{i} \in \Theta_{i}}J_{i}^{(\theta_{i}, \theta_{-i}^{*})}(t)-J_{i}^{(\theta_{i}^{*}, \theta_{-i}^{*})}(t) \Big) \leqslant 0,
$$
then $(\pi_{\theta_{i}^{*}})_{i=1}^{N}$ is called a Nash equilibrium.

\end{definition}

\begin{definition}
\citep{CDaskalakis} ($\epsilon$-Nash equilibrium)
For the game $\Gamma$ and for any given $\epsilon >0$, if there exists $\theta^{*}=(\theta_{i}^{*})_{i=1}^{N} \in \Theta$ such that
$$
\sup _{i \in \mathcal{N}}\Big(\sup _{\theta_{i} \in \Theta_{i}}J_{i}^{(\theta_{i}, \theta_{-i}^{*})}(t)-J_{i}^{(\theta_{i}^{*}, \theta_{-i}^{*})}(t) \Big) \leqslant \epsilon,
$$
then  $(\pi_{\theta_{i}^{*}})_{i=1}^{N}$ is called an $\epsilon$-Nash equilibrium.
\end{definition}

The iterative output of the algorithm to learn a Nash equilibrium is often a sequence of random variables. Then we introduce the following concepts of weighted asymptotic Nash equilibrium with probability $1$ and  in probability .

\begin{definition}\label{nashweightapp}
For the game $\Gamma$ and for a given random parameter sequence $\{\theta_{k}=(\theta_{i,k})_{i=1}^{N}$, $k \geqslant 1\}$,
if there exists a positive sequence $\{\gamma_{k}, k \geqslant 1\}$ and a nonnegative random sequence $\{\epsilon_{k}, k \geqslant 1\}$, satisfying $\sum_{k=1}^{\infty}\gamma_k=\infty$ and
\begin{equation}\nonumber
\begin{aligned}
\sup_{i \in \mathcal{N}}\frac{\sum_{k=1}^{K}\gamma_{k}\Big(\sup \limits _{\theta_{i} \in \widetilde{\Theta_{i}}}J_{i}^{(\theta_{i}, \theta_{-i,k})}(t)-J_{i}^{(\theta_{i,k}, \theta_{-i,k})}(t)\Big)}{\sum_{k=1}^{K}\gamma_{k}}\leqslant
\epsilon_{K},  \ K=1,2,...,\ a.s., \\
 for \ any \ given \ \text{countable set}\ \widetilde{\Theta_{i}} \subseteq \Theta_{i},
\end{aligned}
\end{equation}
 then the sequence of policies $\big((\pi_{\theta_{i,k}})_{i=1}^{N}\big)
_{k=1}^{\infty}$ is called a $\gamma_{k}$-weighted $\epsilon_{k}$-Nash equilibrium  with probability $1$.
\end{definition}
\begin{definition}\label{nashalmost2}
For the game $\Gamma$ and for a given random parameter sequence $\{\theta_{k}=(\theta_{i,k})_{i=1}^{N}$, $k \geqslant 1\}$, if there exists a positive sequence $\{\gamma_{k}, k \geqslant 1\}$, satisfying $\sum_{k=1}^{\infty}\gamma_k=\infty$ and
\begin{equation}\nonumber
\begin{aligned}
\lim_{K\to\infty}\sup_{i \in \mathcal{N}}\frac{\sum_{k=1}^{K}\gamma_{k}\Big(\sup \limits _{\theta_{i} \in \widetilde{\Theta_{i}}}J_{i}^{(\theta_{i}, \theta_{-i,k})}(t)-J_{i}^{(\theta_{i,k}, \theta_{-i,k})}(t)\Big)}{\sum_{k=1}^{K}\gamma_{k}}=0 \ a.s., \\
  for \ any \ given \ \text{countable set} \ \widetilde{\Theta_{i}} \subseteq \Theta_{i},
\end{aligned}
\end{equation}
then the sequence of policies $\big((\pi_{\theta_{i,k}})_{i=1}^{N}\big)
_{k=1}^{\infty}$ is called a $\gamma_{k}$-weighted asymptotic Nash equilibrium  with probability $1$.
\end{definition}
\begin{definition}\label{nashprobability}
For the game $\Gamma$ and for a given random parameter sequence  $\{\theta_{k}=(\theta_{i,k})_{i=1}^{N}$, $k \geqslant 1\}$, if there exists a positive sequence $\{\gamma_{k}, k \geqslant 1\}$, satisfying $\sum_{k=1}^{\infty}\gamma_k=\infty$, and for any $\delta\in (0,1]$, $\epsilon>0$, there exists $K_0>0$ such that
\begin{equation}\nonumber
\begin{aligned}
P\left\{\sup_{i \in \mathcal{N}}\frac{\sum_{k=1}^{K}\gamma_{k}\Big(\sup \limits _{\theta_{i} \in \widetilde{\Theta_{i}}}J_{i}^{(\theta_{i}, \theta_{-i,k})}(t)-J_{i}^{(\theta_{i,k}, \theta_{-i,k})}(t)\Big)}{\sum_{k=1}^{K}\gamma_{k}} \leqslant \epsilon\right\}\geqslant 1-\delta,  \ K\geq K_0,\\
  for \ any \ given \ \text{countable set} \ \widetilde{\Theta_{i}} \subseteq \Theta_{i},
\end{aligned}
\end{equation}
then the sequence of  policies $\big((\pi_{\theta_{i,k}})_{i=1}^{N}\big)
_{k=1}^{\infty}$ is called a  $\gamma_{k}$-weighted asymptotic Nash equilibrium  in probability.
\end{definition}

\begin{remark}
Since $\Theta_{i}$ is an uncountable set, $\sup \limits_{\theta_{i} \in \Theta_{i}}J_{i}^{(\theta_{i}, \theta_{-i,k})}(t)$ may be not a random variable anymore. To avoid making complex separability assumptions on $J_{i}^{(\theta_{i}, \theta_{-i,k})}(t)$, we only focus on the supremum of $J_{i}^{(\theta_{i}, \theta_{-i,k})}(t)$
over any countable subset of $\Theta_{i}$ in Definitions $\ref{nashweightapp}$-$\ref{nashprobability}$.
Especially, if the random parameter sequence $\{\theta_{i,k}, k \geqslant 1\}$ is deterministic, then $\widetilde{\Theta_{i}}$ can be replaced by $\Theta_{i}$ in the above definitions.
\end{remark}

It's easy to prove the following theorem which implies the connection of the above three definitions.

\begin{theorem}\label{weightuninoweight}
For the game $\Gamma$, if there exists a random sequence $\{\epsilon_{k}, k \geqslant 1\}$
such that  $\big((\pi_{\theta_{i,k}})_{i=1}^{N}\big)
_{k=1}^{\infty}$ is a $\gamma_{k}$-weighted $\epsilon_{k}$-Nash equilibrium with probability $1$ and $\{\epsilon_{k}, k\geqslant 1\}$ tends to zero with probability  $1$ (in probability), then $\big((\pi_{\theta_{i,k}})_{i=1}^{N}\big)
_{k=1}^{\infty}$ is a  $\gamma_{k}$-weighted asymptotic Nash equilibrium with probability $1$ (in probability).
\end{theorem}

The following theorem illustrates that if a sequence of policies is a $\gamma_{k}$-weighted  $\epsilon_{k}$-Nash equilibrium of the game $\Gamma$ with probability $1$, then a random subsequence of the original sequence is an $\epsilon_{k}$-Nash equilibrium in expectation.

\begin{theorem}\label{connectionofne}
For the game $\Gamma$, if $\big((\pi_{\theta_{i,k}})_{i=1}^{N}\big)
_{k=1}^{\infty}$ is a $\gamma_{k}$-weighted $\epsilon_{k}$-Nash equilibrium with probability $1$ and for any $K=1,2,...$, there exist random variables $\tau_{K} \in \{1, \ldots, K\}$ such that $P(\tau_{K}=k)=\frac{\gamma_{k}}{\sum_{k=1}^{K}\gamma_{k}}$ and $\{\tau_{k}, k\geqslant 1\}$ is independent of $\{\theta_{k}=(\theta_{i,k})_{i=1}^{N}, k \geqslant 1\}$, then
$\sup\limits _{i \in \mathcal{N}}
\mathbb{E}\Big[\sup \limits_{\theta_{i}\in \widetilde{\Theta_{i}}}J_{i}^{(\theta_{i}, \theta_{-i,\tau_K})}(t)-J_{i}^{(\theta_{i,\tau_K}, \theta_{-i,\tau_K})}(t)\Big|\theta_{k},k=1,2,...,K\Big] \leqslant \epsilon_{K}$, where $\widetilde{\Theta_{i}}$ is any countable subset of $\Theta_{i}$.
\end{theorem}
\begin{proof}
See Appendix B.
\end{proof}

\section{ Existence of Nash Equilibrium}
In this section, we will show that the policy parameterization satisfies a gradient
domination theorem and will prove the equivalence between Nash equilibrium  and variational inequality problems. Then the existence of  Nash equilibrium is established.

We start with the following assumptions on policy parameters and immediate rewards.

\begin{assumption}\label{assumption0}
The conditional probability density function of the initial state for any given $\theta \in \Theta$ is independent of $\theta$, that is, $\rho_{t}(s\mid\theta)=\rho_{t}(s)$.
\end{assumption}

\begin{assumption}\label{assumption1}
There exists $U_{R} >0$ such that $\sup\limits _{i \in \mathcal{N},\ (s,a) \in \mathcal{S} \times \mathcal{A}}|r_{i}(s, a)| \leqslant U_{R}$.  $\Theta_{i} \subseteq \mathbb{R}^{d_{i}}$ is a nonempty compact convex set. By the compactness of $\Theta_{i}$, we may assume there exists $D_{i}>0$ such that $ \left\|\theta_{i}\right\| \leqslant \sqrt{2} D_{i}$. $\pi_{\theta_{i}} \left(a_{i} \mid s\right)$ is concave and continuously differentiable with respect to $\theta_{i} \in  \Theta_{i}$.
\end{assumption}

In particular, direct parameterization, $\alpha$-greedy direct parameterization and Gaussian parameterization under some conditions satisfy Assumption \ref{assumption1}.

\begin{assumption}\label{assumption2}
For any $i \in \mathcal{N}$, the policy  $\pi_{\theta_{i}}$ satisfies
the following conditions:\\
$\nabla_{\theta_{i}} \log \pi_{\theta_{i}}(a_{i} \mid s)$ exists, $\forall \ s \in \mathcal{S}$, $\forall \ a_{i} \in \mathcal{A}_{i}$, and there exist $L_{\Theta} > 0$ and $B_{\Theta}>0$ such that
\begin{align}
&\left\|\nabla_{\theta_{i}} \log \pi_{\theta^{1}_{i}}(a_{i} \mid s )-\nabla_{\theta_{i}} \log \pi_{\theta^{2}_{i}}(a_{i} \mid s)\right\| \leqslant L_{\Theta} \left\|\theta^{1}_{i}-\theta^{2}_{i}\right\|,\forall\ \theta^{1}_{i}, \ \theta^{2}_{i} \in \Theta_{i}, \forall \ s \in \mathcal{S}, \forall \ a_{i} \in \mathcal{A}_{i};\notag\\
& \sup _{\theta_{i} \in \Theta_{i}, \ s \in \mathcal{S},\ a_{i} \in \mathcal{A}_{i}}\left\|\nabla_{\theta_{i}} \log \pi_{\theta_{i}}(a_{i} \mid s)\right\| \leqslant B_{\Theta}\notag.
\end{align}
\end{assumption}

\begin{remark}
Some commonly used parameterized policies such as Gaussian policy under some conditions satisfy Assumptions \ref{assumption1}-\ref{assumption2}. For Gaussian policy, $$\pi_{\theta_{i}}(a_{i} \mid s)=\frac{1}{\sigma \sqrt{2\pi}} \exp\left(-\frac{(a_{i}-\phi^{\top}_{i}(s) \theta_{i})^{2}}{2 \sigma^{2}}\right),$$
where $\phi_{i}^{\top}(s) \theta_{i}$ is the mean of actions, $\sigma^{2}$ is  the variance, and $\phi_{i}(s) \in \mathbb{R}^{d_{i}}$ is the feature vector to approximate the mean action at the state $s$, if the following conditions are satisfied:
(i) $\mathcal{A}$ is bounded; (ii) $\sup \limits_{s \in \mathcal{S}}\left\|\phi_{i}(s)\right\| < \infty$; (iii) $\sup \limits _{i \in \mathcal{N},\ s \in \mathcal{S},\ a_{i} \in \mathcal{A}_{i},\ \theta_{i} \in \Theta_{i}}\left\|\phi^{\top}_{i}(s) \theta_{i}-a_{i}\right\| \leqslant \sigma$; (iv) $\Theta_{i}$ is compact and convex in $\mathbb{R}^{d_{i}}$,  then Assumptions \ref{assumption1}-\ref{assumption2} hold.
\end{remark}

The action-value function $Q_{i}^{\pi_{\theta}}(s,a,t)$ of agent $i$ is defined as the conditional expectation of the  discounted sum of immediate rewards  starting from the initial state $s(t)=s$ and the joint action $a(t)=a$ at time $t$  by choosing actions according to the policy $\pi_{\theta}$, that is,
\begin{equation}\label{qvaluefunction}
Q_{i}^{\pi_{\theta}}(s,a,t)=r_{i}(s,a,t)+\mathbb{E}\left[\sum_{l=t+1}^{\infty} \gamma^{l-t} r_{i}\left(s\left(l\right),a\left(l\right) \right) \bigg| s\left(t\right)=s,a\left(t\right)=a, \theta \right].
\end{equation}

We denote $Q_{i}^{\pi_{\theta}}(s,a,t)$ by $Q_{i}^{\theta}(s,a,t)$.

The relationship between the state-value function and the action-value function is given as
\begin{align}
&V_{i}^{\theta}(s,t)=\int_{\mathcal{A}} Q^{\theta}_{i}(s,a,t)\pi_{\theta}(a\mid s) \mathrm{d} a,\label{valuerelation}\\
&Q_{i}^{\theta}(s,a,t)=r_{i}(s,a,t)+\gamma \int_{\mathcal{S}}  V_{i}^{\theta}(s^{\prime},t+1)\rho_{t,t+1}(s^{\prime} \mid s,a) \mathrm{d}s^{\prime}.\label{valuerelation1}
\end{align}
(\ref{valuefunction}) and (\ref{qvaluefunction}) demonstrate that although $V_{i}^{\theta}(s,t)$ and $Q_{i}^{\theta}(s,a,t)$ are marked out the initial time $t$, $V_{i}^{\theta}(s,t)$ and
$Q_{i}^{\theta}(s,a,t)$ are dependent on the initial state rather than the initial time.

By Assumptions \ref{assumption0}-\ref{assumption2}, (\ref{valuefunction}), (\ref{totalreward}), (\ref{valuerelation}) and (\ref{valuerelation1}), we have the following proposition.

\begin{proposition}\label{policygradient}
If Assumptions \ref{assumption0}-\ref{assumption2} hold, then
\begin{align}\nonumber
\nabla_{\theta_{i}}  J_{i}^{\theta}(t)=\frac{1}{1-\gamma} \int_{\mathcal{S} \times \mathcal{A}} d_{\rho_{t}}^{\theta}(s^{\prime})  \nabla_{\theta_{i}} \pi_{\theta_{i}}(a_{i} \mid s^{\prime})  \pi_{\theta_{-i}}(a_{-i} \mid s^{\prime}) Q_{i}^{\theta}(s^{\prime}, a,t+1) \mathrm{d} s^{\prime} \mathrm{d} a,
\end{align}
where
\begin{equation}\nonumber
\rho_{t,l}^{\theta}\big(s^{\prime} \mid s\big)
= \begin{cases}\delta(s^{\prime}-s), & l=t, \\ \int_{\mathcal{S}} \rho_{l-1,l}^{\theta}\big(s^{\prime} \mid s^{\prime \prime}\big) \rho_{t,l-1}^{\theta}\big(s^{\prime \prime} \mid s\big) \mathrm{d} s^{\prime \prime}, & l \geqslant t+1,
 \end{cases}
\end{equation}
$\rho_{t,l}^{\theta}\big(s^{\prime} \mid s\big)=\rho_{t,l}^{\theta}\big(s(l)=s^{\prime} \mid s(t)=s\big)$
and
$
d_{\rho_{t}}^{\theta}(s^{\prime})=(1-\gamma) \int_{ \mathcal{S}}\sum_{l=t}^{\infty} \gamma^{l-t} \rho_{t,l}^{\theta}\big(s^{\prime} \mid s\big) \rho_{t}(s)\mathrm{d}s
$ is called the probability density function of the discounted state distribution induced by $\pi_{\theta}$.
\end{proposition}

\begin{proof}
We will prove  by induction that
\begin{align}
\nabla_{\theta_{i}}  V_{i}^{\theta}(s,t)
=&\int_{\mathcal{S} \times \mathcal{A}}\sum_{l=t}^{t+k}\gamma^{l-t}\rho_{t,l}^{\theta}\big(s^{\prime} \mid s\big) \nabla_{\theta_{i}} \log\pi_{\theta_{i}}(a_{i} \mid s^{\prime})  \pi_{\theta}(a \mid s^{\prime})  Q_{i}^{\theta}(s^{\prime}, a,l)  \mathrm{d} s^{\prime} \mathrm{d} a
\notag\\ &+ \gamma^{k+1}\left( \int_{\mathcal{S}} \rho_{t,t+k+1}^{\theta}(s^{\prime} \mid s) \nabla_{\theta_{i}} V_{i}^{\theta}(s^{\prime},t+k+1) \mathrm{d}s^{\prime} \right) ,\ \forall \ k \in \mathbb{N}\label{policyinduction}.
\end{align}
Let $k=0$, by Assumptions \ref{assumption1}-\ref{assumption2}, taking  the derivative with respect to $\theta_{i}$ on  both sides of (\ref{valuerelation}) and combining with (\ref{valuerelation1}) give
\begin{align}
\nabla_{\theta_{i}}  V_{i}^{\theta}(s,t)=&
\nabla_{\theta_{i}} \left(\int_{\mathcal{A}} Q^{\theta}_{i}(s,a,t)\pi_{\theta}(a\mid s) \mathrm{d} a\right) \notag\\
=&\int_{\mathcal{A}} Q^{\theta}_{i}(s,a,t)\nabla_{\theta_{i}}\pi_{\theta}(a\mid s)
\mathrm{d} a+
\int_{\mathcal{A}} \nabla_{\theta_{i}}Q^{\theta}_{i}(s,a,t)\pi_{\theta}(a\mid s) \mathrm{d} a\notag\\
=&\int_{\mathcal{A}} Q^{\theta}_{i}(s,a,t)\nabla_{\theta_{i}}
\pi_{\theta_{i}}(a_{i}\mid s)\pi_{\theta_{-i}}(a_{-i}\mid s)\mathrm{d}a\notag\\
&+\int_{\mathcal{A}} \nabla_{\theta_{i}}\big(r_{i}(s,a,t) +\gamma \int_{\mathcal{S}} \rho_{t,t+1}(s^{\prime} \mid s,a)  V_{i}^{\theta}(s^{\prime},t+1) \mathrm{d}s^{\prime}\big) \pi_{\theta}(a\mid s) \mathrm{d} a\notag\\
=&\int_{\mathcal{A}} Q^{\theta}_{i}(s,a,t)
 \nabla_{\theta_{i}}\log\pi_{\theta_{i}}(a_{i}\mid s)\pi_{\theta}(a \mid s)
\mathrm{d}a\notag\\
&+\gamma \int_{\mathcal{S}} \rho_{t,t+1}^{\theta}(s^{\prime} \mid s) \nabla_{\theta_{i}} V_{i}^{\theta}(s^{\prime},t+1) \mathrm{d}s^{\prime}. \label{policyinductionk0}
\end{align}
Thus, (\ref{policyinduction}) is true for $k=0$.
Let $m $ be any positive integer and suppose (\ref{policyinduction}) is true for $k=m$, that is,
\begin{align}
\nabla_{\theta_{i}}  V_{i}^{\theta}(s,t)
=&\int_{\mathcal{S} \times \mathcal{A}}\sum_{l=t}^{t+m}\gamma^{l-t}\rho_{t,l}^{\theta}\big(s^{\prime} \mid s\big) \nabla_{\theta_{i}} \log \pi_{\theta_{i}}(a_{i} \mid s^{\prime})  \pi_{\theta}(a \mid s^{\prime}) Q_{i}^{\theta}(s^{\prime}, a,l)  \mathrm{d} s^{\prime} \mathrm{d} a \notag\\
&+ \gamma^{m+1} \int_{\mathcal{S}} \rho_{t,t+m+1}^{\theta}(s^{\prime} \mid s)  \nabla_{\theta_{i}} V_{i}^{\theta}(s^{\prime},t+m+1) \mathrm{d}s^{\prime} .\label{inductionm}
\end{align}
For the term $\nabla_{\theta_{i}} V_{i}^{\theta}(s^{\prime},t+m+1)$, by Assumptions \ref{assumption1}-\ref{assumption2}, similar to the proof of (\ref{policyinductionk0}), we have
\begin{align*}
\nabla_{\theta_{i}}  V_{i}^{\theta}(s^{\prime},t+m+1)=
&\int_{\mathcal{A}} Q^{\theta}_{i}(s^{\prime},a,t+m+1)
\nabla_{\theta_{i}}\log\pi_{\theta_{i}}(a_{i}\mid s^{\prime})\pi_{\theta}(a\mid s^{\prime})
\mathrm{d}a\notag\\
&+\gamma \int_{\mathcal{S}} \rho_{t+m+1,t+m+2}^{\theta}(s^{\prime \prime} \mid s^{\prime}) \nabla_{\theta_{i}} V_{i}^{\theta}(s^{\prime \prime},t+m+2) \mathrm{d}s^{\prime\prime}.
\end{align*}
By (\ref{inductionm}) and the above inequality, we have
\begin{align*}
&\nabla_{\theta_{i}}  V_{i}^{\theta}(s,t)\\
=&\int_{\mathcal{S} \times \mathcal{A}}\sum_{l=t}^{t+m}\gamma^{l-t}\rho_{t,l}^{\theta}\big(s^{\prime} \mid s\big) \nabla_{\theta_{i}} \log \pi_{\theta_{i}}(a_{i} \mid s^{\prime})  \pi_{\theta}(a \mid s^{\prime}) Q_{i}^{\theta}(s^{\prime}, a,l)\mathrm{d} s^{\prime} \mathrm{d} a   \notag\\
& + \gamma^{m+1}\left( \int_{\mathcal{S}\times \mathcal{A}} \rho_{t,t+m+1}^{\theta}(s^{\prime} \mid s)  Q^{\theta}_{i}(s^{\prime},a,t+m+1)\nabla_{\theta_{i}} \log
\pi_{\theta_{i}}(a_{i}\mid s^{\prime})\pi_{\theta}(a \mid s^{\prime})
\mathrm{d}a \mathrm{d}s^{\prime} \right) \notag\\
&+ \gamma^{m+2}\int_{\mathcal{S}} \rho_{t,t+m+2}^{\theta}(s^{ \prime} \mid s)
 \nabla_{\theta_{i}}V_{i}^{\theta}(s^{\prime},t+m+2) \mathrm{d}s^{\prime}\notag\\
=&\int_{\mathcal{S} \times \mathcal{A}}\sum_{l=t}^{t+m+1}\gamma^{l-t}\rho_{t,l}^{\theta}\big(s^{\prime} \mid s\big) \nabla_{\theta_{i}} \log\pi_{\theta_{i}}(a_{i} \mid s^{\prime})  \pi_{\theta}(a \mid s^{\prime}) Q_{i}^{\theta}(s^{\prime}, a,l)  \mathrm{d} s^{\prime} \mathrm{d} a \notag\\
&+\gamma^{m+2} \left( \int_{\mathcal{S}} \rho_{t,t+m+2}^{\theta}(s^{ \prime} \mid s) \nabla_{\theta_{i}} V_{i}^{\theta}(s^{\prime},t+m+2) \mathrm{d}s^{\prime} \right).\notag
\end{align*}
Thus, (\ref{policyinduction}) holds for $k=m+1$. By the principle of induction,  (\ref{policyinduction}) is true for all $k \in \mathbb{N}$.

From (\ref{valuefunction}), (\ref{qvaluefunction}) and Assumption \ref{assumption1}, it follows that
$\sup \limits_{s \in \mathcal{S},a \in \mathcal{A}}|Q_{i}^{\theta}(s^{\prime}, a,l)| \leqslant \frac{U_{R}}{1-\gamma}$, which together with Assumption \ref{assumption2} and (\ref{policyinductionk0}) implies
\begin{align}
\sup \limits_{s \in \mathcal{S}}\left|\nabla_{\theta_{i}}V_{i}^{\theta}(s,t)\right| \leqslant & \sup \limits_{s \in \mathcal{S}}\left|\int_{\mathcal{A}} Q^{\theta}_{i}(s,a,t)
\nabla_{\theta_{i}}\log \pi_{\theta_{i}}(a_{i}\mid s)\pi_{\theta}(a \mid s)
\mathrm{d}a\right|\notag\\
&+\gamma\sup \limits_{s \in \mathcal{S}}\left| \int_{\mathcal{S}} \rho_{t,t+1}^{\theta}(s^{\prime} \mid s) \nabla_{\theta_{i}} V_{i}^{\theta}(s^{\prime},t+1)  \mathrm{d}s^{\prime}\right|\notag\\
\leqslant &\frac{U_{R}B_{\Theta}}{1-\gamma}+\gamma  \sup \limits_{s\in \mathcal{S}}\left|\nabla_{\theta_{i}} V_{i}^{\theta}(s,t+1) \right|\notag.
\end{align}
This together with the fact that $V_{i}^{\theta}(s,t)$ are dependent on the initial state rather than the initial time gives
\begin{align}
\sup \limits_{s \in \mathcal{S}}\left|\nabla_{\theta_{i}}V_{i}^{\theta}(s,t)\right| \leqslant
\frac{U_{R}B_{\Theta}}{(1-\gamma)^{2}}.\label{boundedpseudogradient}
\end{align}
By Assumptions \ref{assumption1}-\ref{assumption2}, we have
\begin{align}
&\sum_{l=t}^{\infty}\int_{\mathcal{S} \times \mathcal{A}}\gamma^{l-t}
\rho_{t,l}^{\theta}\big(s^{\prime} \mid s\big) \left|\nabla_{\theta_{i}}\log \pi_{\theta_{i}}(a_{i} \mid s^{\prime})\right|  \pi_{\theta}(a \mid s^{\prime})  \left|Q_{i}^{\theta}(s^{\prime}, a,l)\right|  \mathrm{d} s^{\prime} \mathrm{d} a \notag\\
\leqslant & \frac{B_{\Theta}U_{R}}{1-\gamma} \sum_{l=t}^{\infty}\gamma^{l-t}\int_{\mathcal{S} \times \mathcal{A}}
\rho_{t,l}^{\theta}\big(s^{\prime} \mid s\big) \pi_{\theta}(a \mid s^{\prime})  \mathrm{d} s^{\prime}
\mathrm{d} a
\leqslant \frac{B_{\Theta}U_{R}}{(1-\gamma)^{2}}.\notag
\end{align}
Then by the above inequality, (\ref{boundedpseudogradient}) and the Dominated Convergence Theorem, letting $k$ tends to infty on  both sides of (\ref{policyinduction}) gives
\begin{align}
\nabla_{\theta_{i}}  V_{i}^{\theta}(s,t)
=&\int_{\mathcal{S} \times \mathcal{A}}\sum_{l=t}^{\infty}\gamma^{l-t}\rho_{t,l}^{\theta}\big(s^{\prime} \mid s\big) \nabla_{\theta_{i}} \log \pi_{\theta_{i}}(a_{i} \mid s^{\prime})  \pi_{\theta}(a \mid s^{\prime}) Q_{i}^{\theta}(s^{\prime}, a,l)  \mathrm{d} s^{\prime} \mathrm{d} a. \notag
\end{align}
By Assumption \ref{assumption0}, it is known that the conditional probability density function of $s(t)$  for a given $\theta$ satisfies $\rho_{t}(s\mid \theta)=\rho_{t}(s)$, which together with (\ref{valuefunction}), (\ref{totalreward}),  the above equality and the property of conditional expectation gives
\begin{align}\nonumber
\nabla_{\theta_{i}}  J_{i}^{\theta}(t)
=&\nabla_{\theta_{i}}\mathbb{E}\left[
\sum_{l=t}^{\infty} \gamma^{l-t} r_{i}\big(s\big(l\big),a\big(l\big) \big)\bigg | \theta\right] \notag\\
=&\nabla_{\theta_{i}}\mathbb{E}\left[
\mathbb{E}\left[\sum_{l=t}^{\infty} \gamma^{l-t} r_{i}\big(s\big(l\big),a\big(l\big) \big)\bigg | s(t)=s,\theta\right]\bigg |\theta\right]\notag\\
=&\nabla_{\theta_{i}}\int_{\mathcal{S}}\mathbb{E}\left[\sum_{l=t}^{\infty} \gamma^{l-t} r_{i}\big(s\big(l\big),a\big(l\big) \big)\bigg | s(t)=s, \theta\right]\rho_{t}(s\mid \theta) \mathrm{d}s\notag\\
=& \nabla_{\theta_{i}}\int_{\mathcal{S}}V_{i}^{\theta}(s,t) \rho_{t}(s) \mathrm{d}s \notag\\
=& \int_{\mathcal{S}} \nabla_{\theta_{i}}  V_{i}^{\theta}(s,t)
\rho_{t}(s) \mathrm{d}s \notag\\
=&\int_{\mathcal{S}}\int_{ \mathcal{S} \times \mathcal{A}} \sum_{l=t}^{\infty} \gamma^{l-t} \rho_{t,l}^{\theta}\big(s^{\prime} \mid s\big) \nabla_{\theta_{i}} \log\pi_{\theta_{i}}(a_{i} \mid s^{\prime})  \pi_{\theta}(a \mid s^{\prime}) Q_{i}^{\theta}(s^{\prime}, a,l)  \mathrm{d} s^{\prime} \mathrm{d} a \rho_{t}(s) \mathrm{d}s \notag \\
=&\frac{1}{1-\gamma} \int_{\mathcal{S} \times \mathcal{A}} (1-\gamma) \int_{\mathcal{S}}\sum_{l=t}^{\infty} \gamma^{l-t} \rho_{t,l}^{\theta}\big(s^{\prime} \mid s\big)  \nabla_{\theta_{i}} \log \pi_{\theta_{i}}(a_{i} \mid s^{\prime}) \pi_{\theta}(a \mid s^{\prime})\notag\\
&~~~\times Q_{i}^{\theta}(s^{\prime}, a,l) \rho_{t}(s) \mathrm{d} s^{\prime} \mathrm{d} a  \mathrm{d} s  \notag\\
=&\frac{1}{1-\gamma} \int_{\mathcal{S} \times \mathcal{A}} d_{\rho_{t}}^{\theta}(s^{\prime})  \nabla_{\theta_{i}} \pi_{\theta_{i}}(a_{i} \mid s^{\prime})  \pi_{\theta_{-i}}(a_{-i} \mid s^{\prime}) Q_{i}^{\theta}(s^{\prime}, a,t+1) \mathrm{d} s^{\prime} \mathrm{d} a.\notag
\end{align}
\end{proof}

To give the gradient domination theorem, we need the following assumption  widely used in literatures  \cite{Zhang}, \cite{Daskalakis} and \cite{Runyu}.

\begin{assumption}\label{assumption3}
The induced discounted Markov chain by $N$ agents satisfy
$d_{\rho_{t}}^{\theta}(s^{\prime}) > 0$, $\forall\ s^{\prime} \in S$, $\forall\ \theta \in \Theta$.
\end{assumption}

We define the pseudo gradient mapping  $F:\Theta \rightarrow \mathbb{R}^{\sum_{i=1}^{N} d_{i}}$ as $F(\theta)=\left(F_{i}\left(\theta\right)\right)_{i=1}^{N}=\left(-\nabla_{\theta_{i}} J_{i}^{\theta}(t)\right)_{i=1}^{N}$, $\forall \ \theta \in \Theta$. Then we have  the following gradient domination theorem by Assumptions \ref{assumption0}-\ref{assumption3}.

\begin{lemma}\label{gradientd}(Gradient domination theorem)
If Assumptions \ref{assumption0}-\ref{assumption3} hold, then for any
$\theta=\big(\theta_{i}, \theta_{-i}\big)
 \in \Theta$, we have
\begin{equation}\label{domination}
 \sup _{\theta_{i}^{\prime} \in \Theta_{i}} J_{i}^{(\theta_{i}^{\prime}, \theta_{-i})}(t) -J_{i}^{(\theta_{i}, \theta_{-i})}(t)
 \leqslant M_{1} \sup _{\bar{\theta}_{i} \in \Theta_{i}}
 \langle F_{i}(\theta),\bar{\theta}_{i}-\theta_{i}\rangle,
 \\\ \forall \ i=1,2, \ldots, N  ,
\end{equation}
where $M_{1}=\sup \limits_{\theta,\hat{\theta}^{i}\in \Theta}\left\|\frac{d_{\rho_{t}}^{\hat{\theta}^{i}}}{d_{\rho_{t}}^{\theta}}
\right\|_{\infty}$, $
\left\|\frac{d_{\rho_{t}}^{\hat{\theta}^{i}}}{d_{\rho_{t}}^{\theta}}
\right\|_{\infty}=\sup\limits _{s^{\prime} \in \mathcal{S}} \frac{d_{\rho_{t}}^{\hat{\theta}^{i}}(s^{\prime})}{d_{\rho_{t}}^{\theta}(s^{\prime})}
$ and $\hat{\theta}^{i}=(\theta_{i}^{\prime},\theta_{-i}) \in \Theta$.
\end{lemma}

\begin{proof}
By Assumption $\ref{assumption0}$, Assumption $\ref{assumption3}$ and Lemma \ref{performancedifference}, we have
\begin{align}\label{difference}
&J_{i}^{(\theta_{i}^{\prime}, \theta_{-i})}(t) -J_{i}^{(\theta_{i}, \theta_{-i})}(t)\notag\\
=&\frac{1}{1-\gamma} \int_{\mathcal{S}\times \mathcal{A}_{i}} d_{\rho_{t}}^{\hat{\theta}^{i}}(s^{\prime}) \pi_{\theta_{i}^{\prime}}\big(a_{i} \mid s^{\prime}\big) \left(\int_{\mathcal{A}_{-i}} \pi_{\theta_{-i}}\big(a_{-i} \mid s^{\prime}\big) A_{i}^{\theta}\big(s^{\prime}, a,t+1\big)\mathrm{d} a_{-i}\right)\mathrm{d} a_{i} \mathrm{d} s^{\prime}\notag\\
\leqslant &\frac{1}{1-\gamma} \sup _{\bar{\theta}_{i} \in \Theta_{i}}  \int_{\mathcal{S}\times \mathcal{A}_{i}} \frac{d_{\rho_{t}}^{\hat{\theta}^{i}}(s^{\prime})}{d_{\rho_{t}}^{\theta}(s^{\prime})} d_{\rho_{t}}^{\theta}(s^{\prime}) \pi_{\bar{\theta_{i}}}\big(a_{i} \mid s^{\prime}\big) \bigg(\int_{\mathcal{A}_{-i}} \pi_{\theta_{-i}}\big(a_{-i} \mid s^{\prime}\big)\notag\\
 & \times A_{i}^{\theta}\big(s^{\prime}, a,t+1\big)\mathrm{d} a_{-i}\bigg)\mathrm{d} a_{i} \mathrm{d} s^{\prime}\notag\\
\leqslant & \left\|\frac{d_{\rho_{t}}^{\hat{\theta}^{i}}}{d_{\rho_{t}}^{\theta}}\right\|_{\infty}
\frac{1}{1-\gamma} \sup _{\bar{\theta}_{i} \in \Theta_{i}}  \int_{\mathcal{S}\times \mathcal{A}_{i}}  d_{\rho_{t}}^{\theta}(s^{\prime}) \pi_{\bar{\theta_{i}}}\big(a_{i} \mid s^{\prime}\big) \Bigg(\int_{\mathcal{A}_{-i}} \pi_{\theta_{-i}}\big(a_{-i} \mid s^{\prime}\big) \notag \\
&\times A_{i}^{\theta}\big(s^{\prime}, a,t+1\big)\mathrm{d} a_{-i}\Bigg)\mathrm{d}a_{i} \mathrm{d} s^{\prime}\notag\\
=&\left\|\frac{d_{\rho_{t}}^{\hat{\theta}^{i}}}{d_{\rho_{t}}^{\theta}}\right\|_{\infty} \sup _{\bar{\theta}_{i} \in \Theta_{i}}  \int_{\mathcal{S}\times \mathcal{A}_{i}}  \big(\pi_{\bar{\theta_{i}}}\big(a_{i} \mid s^{\prime}\big)-\pi_{\theta_{i}}\big(a_{i} \mid s^{\prime}\big)\big)\frac{1}{1-\gamma} d_{\rho_{t}}^{\theta}(s^{\prime}) \notag\\
& \times \left(\int_{\mathcal{A}_{-i}} \pi_{\theta_{-i}}\big(a_{-i} \mid s^{\prime}\big) A_{i}^{\theta}\big(s^{\prime}, a_{i}, a_{-i},t+1\big)\mathrm{d} a_{-i}\right)\mathrm{d} a_{i} \mathrm{d} s^{\prime}\notag\\
=&\left\|\frac{d_{\rho_{t}}^{\hat{\theta}^{i}}}{d_{\rho_{t}}^{\theta}}\right\|_{\infty} \sup _{\bar{\theta}_{i} \in \Theta_{i}}  \int_{\mathcal{S}\times \mathcal{A}_{i}}  \big(\pi_{\bar{\theta_{i}}}\big(a_{i} \mid s^{\prime}\big)-\pi_{\theta_{i}}\big(a_{i} \mid s^{\prime}\big)\big)\frac{1}{1-\gamma} d_{\rho_{t}}^{\theta}(s^{\prime}) \notag\\
& \times \left(\int_{\mathcal{A}_{-i}} \pi_{\theta_{-i}}\big(a_{-i} \mid s^{\prime}\big) Q_{i}^{\theta}\big(s^{\prime}, a_{i}, a_{-i},t+1\big)\mathrm{d} a_{-i}\right)\mathrm{d} a_{i} \mathrm{d}s^{\prime}\ ,
\end{align}
where $A_{i}^{\theta}(s^{\prime}, a,t+1)=Q_{i}^{\theta}(s^{\prime}, a_{i},a_{-i},t+1)-V_{i}^{\theta}(s^{\prime},t+1)$,
$\int_{\mathcal{A}_{i}}  \pi_{\theta_{i}}
\big(a_{i} \mid s^{\prime}\big) \int_{\mathcal{A}_{-i}}  \pi_{\theta_{-i}}\big(a_{-i} \mid s^{\prime}\big)\\
 A_{i}^{\theta}\big(s^{\prime}, a,t+1\big)\mathrm{d} a_{-i} \mathrm{d} a_{i} =0$ is used in the third equality and the last equality is by $\int_{\mathcal{A}_{i}}  \pi_{\theta_{i}}\big(a_{i} \mid s^{\prime}\big) \int_{\mathcal{A}_{-i}}  \pi_{\theta_{-i}}\big(a_{-i} \mid s^{\prime}\big) V_{i}^{\theta}\big(s^{\prime},t+1\big)\mathrm{d} a_{-i}
\mathrm{d} a_{i}
=\int_{\mathcal{A}_{i}}  \pi_{\bar{\theta_{i}}}\big(a_{i} \mid s^{\prime}\big)  \int_{\mathcal{A}_{-i}} \pi_{\theta_{-i}}\big(a_{-i} \mid s^{\prime}\big)  V_{i}^{\theta}\big(s^{\prime},t+1\big)\mathrm{d} a_{-i} \mathrm{d} a_{i}$.

Noting that $\pi_{\theta_{i}}(a_{i} \mid s^{\prime})$ is concave with respect to $
 \theta_{i}$, we know that $\pi_{\bar{\theta_{i}}}\big(a_{i} \mid s^{\prime}\big)-\pi_{\theta_{i}}\big(a_{i} \mid s^{\prime}\big) \leqslant (\bar{\theta_{i}}-\theta_{i})^{T} \nabla_{\theta_{i}} \pi_{\theta_{i}}(a_{i}\mid s^{\prime})$. This together with Assumption \ref{assumption3}, Proposition \ref{policygradient} and (\ref{difference}) yields
 \begin{align*}
&\sup _{\theta_{i}^{\prime} \in \Theta_{i}}J_{i}^{(\theta_{i}^{\prime}, \theta_{-i})}(t) -J_{i}^{(\theta_{i}, \theta_{-i})}(t)\\
 \leqslant&  \sup _{\theta_{i}^{\prime} \in \Theta_{i}}\left\|\frac{d_{\rho_{t}}^{\hat{\theta}^{i}}}{d_{\rho_{t}}^{\theta}}\right\|_{\infty} \sup _{\bar{\theta}_{i} \in \Theta_{i}}  \int_{\mathcal{S}\times \mathcal{A}_{i}} (\bar{\theta_{i}}-\theta_{i})^{T}\nabla_{\theta_{i}} \pi_{\theta_{i}}(a_{i}\mid s^{\prime}) \frac{d_{\rho_{t}}^{\theta}(s^{\prime})}{1-\gamma}  \\
 &\times \left(\int_{\mathcal{A}_{-i}} \pi_{\theta_{-i}}\big(a_{-i} \mid s^{\prime}\big) Q_{i}^{\theta}\big(s^{\prime}, a_{i}, a_{-i},t+1\big)\mathrm{d}  a_{-i}\right) \mathrm{d}  a_{i} \mathrm{d}  s^{\prime}\\
 = & \sup _{\theta_{i}^{\prime} \in \Theta_{i}} \left\|\frac{d_{\rho_{t}}^{\hat{\theta}^{i}}}{d_{\rho_{t}}^{\theta}}\right\|_{\infty} \sup _{\bar{\theta}_{i} \in \Theta_{i}}(\bar{\theta_{i}}-\theta_{i})^{T} \nabla_{\theta_{i}} J_{i}^{\theta}(t)\\
 \leqslant &\sup_{\theta,\hat{\theta}^{i}\in \Theta}\left\|\frac{d_{\rho_{t}}^{\hat{\theta}^{i}}}{d_{\rho_{t}}^{\theta}}\right\|_{\infty} \sup _{\bar{\theta}_{i} \in \Theta_{i}}(\bar{\theta_{i}}-\theta_{i})^{T} \nabla_{\theta_{i}} J_{i}^{\theta}(t)
 =  M_{1}\sup _{\bar{\theta}_{i} \in \Theta_{i}}
 \langle F_{i}(\theta),\theta_{i}-\bar{\theta}_{i}\rangle,
\end{align*}
that is, (\ref{domination}) follows.
\end{proof}

 We  assume the parameterized policies are concave with respect to parameters, so the gradient domination theorems for direct parameterization ~\citep{Runyu} and  $\alpha$-greedy direct parameterization ~\citep{Zhang,Daskalakis} are special cases of our result.

Lemma \ref{gradientd} enables agents to approximate their best response to other agents' policies if they update policies in the algorithm  by controlling the upper bound of the right side of (\ref{domination}) or using the first-order necessary optimality
conditions of the total reward functions with respect to the parameter of each player, which makes it possible to learn a Nash equilibrium ~\citep{Daskalakis,Runyu}.

Inspired by Lemma \ref{gradientd}, we can find a Nash equilibrium by means of the first-order necessary optimality conditions of the total reward functions with respect to the parameter of each player. Hence we give the equivalence between Nash equilibrium  and variational inequality problems.

\begin{definition}
~\citep{KINDERLEHRER1980}
For a given subset $K$ in $\mathbb{R}^{n}$ and a mapping $G: K \rightarrow \mathbb{R}^{n}$, a variational inequality problem is to find a vector $x^{*} \in K$ such that
\begin{equation*}
\left\langle G(x^{*}), x^{*}-x\right\rangle \leqslant 0 ,\\\ \forall\ x \in K,
\end{equation*}
or  $x^{*}$ is called a solution of  SVI$(G,K)$.
\end{definition}

\begin{definition}
~\citep{G.J.MINTY}
For a given subset $K$ in $\mathbb{R}^{n}$ and a mapping $G: K \rightarrow \mathbb{R}^{n}$, a Minty variational inequality problem is to find a vector $x^{*} \in K$ such that
\begin{equation*}
\left\langle G(x), x^{*}-x\right\rangle \geqslant 0 ,\\\ \forall\ x \in K,
\end{equation*}
or $x^{*}$ is called a solution of MVI$(G,K)$.
\end{definition}

\begin{definition}
~\citep{Runyu} (First-order Nash equlibrium)\label{deffirstorderne}
For the game $\Gamma$, if there exists $\theta^{*}=\big(\theta_{i}^{*}\big)_{i=1}^{N}
 \in \Theta$ which is a solution of SVI$(F,\Theta)$, that is,
\begin{equation}\nonumber
\sup_{i \in \mathcal{N}}\left(\sup _{\theta_{i} \in \Theta_{i}}  \left\langle F_{i}(\theta^{*}), \theta_{i}^{*}-\theta_{i} \right\rangle \right) \leqslant 0,
\end{equation}
then  $\big(\pi_{\theta_{i}^{*}}\big)_{i=1}^{N}$ is called a first-order Nash equilibrium.
\end{definition}

\begin{definition}($\epsilon$-first-order Nash equlibrium)
For the game $\Gamma$ and for any given $\epsilon >0$, if there
exists $\theta^{*}=\big(\theta_{i}^{*}\big)_{i=1}^{N}  \in \Theta$ such that
\begin{equation}\nonumber
\sup_{i \in \mathcal{N}}\left(\sup _{\theta_{i} \in \Theta_{i}}  \left\langle F_{i}(\theta^{*}), \theta_{i}^{*}-\theta_{i} \right\rangle \right) \leqslant \epsilon,
\end{equation}
then $\big(\pi_{\theta_{i}^{*}}\big)_{i=1}^{N}$ is called an $\epsilon$-first-order Nash equilibrium.
\end{definition}

\begin{lemma}\label{equivalenceofnefne}
For the game $\Gamma$, if Assumptions
\ref{assumption0}-\ref{assumption3} hold, then $\big(\pi_{\theta_{i}^{*}}\big)_{i=1}^{N}$ is a first-order Nash equlibrium if and only if it is a Nash equilibrium; $\big(\pi_{\theta_{i}^{*}}\big)_{i=1}^{N}$ is an $\epsilon$-first-order Nash equilibrium if and only if it is a $M_{1}\epsilon$-Nash equilibrium, where $M_{1}$ is given by Lemma \ref{gradientd}.
\end{lemma}

By  Definition \ref{deffirstorderne} and Lemma
\ref{equivalenceofnefne}, we can characterize the Nash equilibrium problem in terms of SVI$(F,\Theta)$.

\begin{theorem}\label{existencenash}
If Assumptions \ref{assumption0}-\ref{assumption3} hold, then there exists a solution of
SVI$(F,\Theta)$ and the  game $\Gamma$ has a Nash equilibrium.
\end{theorem}
\begin{proof}
See Appendix B.
\end{proof}

\section{Learning Nash Equilibria with  Exact Pseudo Gradients}
In this section, we assume all agents access to the exact pseudo gradient $F(\theta)$. Before we design the algorithm, we will prove that $F(\theta)$ is Lipschitz  continuous with respect to $\theta$ at first.

\begin{lemma}\label{FLIPSCHITZ}
If Assumptions \ref{assumption0}-\ref{assumption2} hold, then $F(\theta)$ is $L$-Lipschitz continuous with respect to $\theta \in \Theta$, where the Lipschitz constant
$$
L=\sqrt{\frac{2(U_{R})^{2}(L_{\Theta})^{2}}{(1-\gamma)^6}+\frac{2(1+\gamma)^{2}(U_{R})^{2}N(B_{\Theta})^{4}}{(1-\gamma)^6}}.
$$
\end{lemma}
\begin{proof}
See Appendix B.
\end{proof}

Then, we  propose  Algorithm \ref{alg::exactGradient} for learning Nash equilibrium with exact pseudo gradients.

\renewcommand{\algorithmicrequire}{\textbf{Input:}}  
\renewcommand{\algorithmicensure}{\textbf{Output:}} 

\begin{algorithm}[H]

         \begin{algorithmic}[1] 
         \caption{Algorithm for Exact Pseudo gradients}
          \label{alg::exactGradient}
          \State Input: Lipschitz constant $L$, integer $K \geqslant 1$, weight $\gamma_{k}$, initial values $\theta_{1}=(\theta_{i,1})_{i=1}^{N} \in \Theta$,
         \State \ \ \ \ \ \ \ \ \ $\beta \in \left(0, \frac{1}{L}\right)$, $\widetilde{\eta}=\frac{1}{2\sqrt{L^{2}+\frac{1}
         {\beta^{2}}}}$.
          \For{$k = 1$, \ldots, $K$}
             \State Input: initial values $\theta^{1}=(\theta^{1}_{i})_{i=1}^{N}
             =z^{1} \in \Theta$, integer $H_{k}$, stepsize $\eta$ satisfying\\
            \ \ \ \ $\eta \in \left(0, \min\bigg\{\frac{1}{\left(\frac{1}
            {\beta}-L\right)}
, \frac{-\left(\frac{1}{\beta}-L\right)+
\sqrt{\left(\frac{1}{\beta}-L\right)^{2}+
32\left(L^{2}+\frac{1}{\beta^{2}}\right)}
}{16}, \frac{-\left(\frac{1}{\beta}-L\right)+
\sqrt{\left(\frac{1}{\beta}-L\right)^{2}
+8\left(L^{2}+\frac{1}{\beta^{2}}\right)}
}{32}
\bigg\}\right)$.
             \State  Let $F_{k}(\theta) = (F_{i,k}(\theta))_{i=1}^{N}=
             F(\theta)+\frac{1}{\beta}(\theta-\theta_{k})
             =\left(F_{i}(\theta)+\frac{1}{\beta}(\theta_{i}-
             \theta_{i,k})\right)_{i=1}^{N}$.
             \For{$h = 1$, \ldots, $H_{k}$}
              \For{$i=1$, \ldots, $N$}
              \State $\theta_{i}^{h+1}=\mathop{\arg\min}\limits_{
              \theta_{i} \in \Theta_{i}} \left\{\left\langle 2\eta F_{i,k}(\theta^{h}), \theta_{i} \right\rangle+\|\theta_{i}-z_{i}^{h}\|^{2}\right\}.$
              \State $z_{i}^{h+1}=\mathop{\arg\min}\limits_{\theta_{i} \in \Theta_{i}} \left\{\left\langle 2\eta F_{i,k}(\theta^{h+1}), \theta_{i} \right\rangle+\|\theta_{i}-z_{i}^{h}\|^{2}\right\}.$

            \EndFor
           \EndFor
          \State $z^{H_{k}+1}=(z_{i}^{H_{k}+1})_{i=1}^{N}$.
          \State $\theta_{k+1}=\mathop{\arg\min}\limits_{\theta \in \Theta} \left\{\left\langle 2\widetilde{\eta} F_{k}(z^{H_{k}+1}), \theta \right\rangle+\|\theta-z^{H_{k}+1}\|^{2}
          \right\}$.
         \EndFor
        \State Randomly choose $\tau_{K}$ satisfying  $P(\tau_{K}=k)=\frac{\gamma_{k}}{\sum_{k=1}^{K}\gamma_{k}}$ , $k=1,\ldots, K $.
         \State Output: $\theta_{\tau_{K}}$.
        \end{algorithmic}
 \end{algorithm}

In Algorithm \ref{alg::exactGradient}, similar to~\cite{Lin} and  ~\cite{JKOSHAL2010,JKOSHAL}, by adding a strongly monotone term $\frac{1}{\beta}(\theta-\theta_{k})$ to $F$, we construct SVI($F_{k}, \Theta)$ in the outer loop and  provide SVI($F_{k}, \Theta)$ is $\left(\frac{1}{\beta}-L\right)$-strongly monotone if $\frac{1}{\beta} > L$ by Lemma \ref{stronglymonotone}. We update SVI($F_{k}, \Theta)$ by updating the proximal parameter $\theta_{k}$.
In the inner loop, we employ a single-call extra-gradient algorithm for solving the constructed strongly monotone variational inequality while ~\cite{Lin} adopted an extra-gradient algorithm with two calls of pseudo gradients.
 We aim at alleviating the cost of pseudo gradients per iteration. ~\cite{Y} also used a single-call extra-gradient algorithm to approximate the solution of a strongly monotone variational inequality.  They measured the performance of the average of the inner loop output by the dual gap function (Definition \ref{definitionofgap}) of SVI$(F_{k},\Theta)$ while we measure the performance of the last-iterate output of the inner loop  by the prime gap function
  (Definition \ref{definitionofgap}) of SVI$(F_{k},\Theta)$. From Definition \ref{definitionofgap}, we know our result is not a corollary of ~\cite{Y}. For the convergence result of the outer-loop iteration, we measure the performance of the output of the outer loop by the prime gap function of SVI$(F,\Theta)$ while ~\cite{Lin} measured the difference between the output of the outer loop and  the real solution of SVI($F_{k}, \Theta)$.

For the convergence of our algorithm, we make the following assumption which  has been adopted in recent
works  on variational inequalities ~\citep{Lin,Song.C}.
\begin{assumption}\label{assumption4}
MVI$(F,\Theta)$ has a solution.
\end{assumption}

Below we will give the lemma and the theorem in this section.
\begin{lemma}\label{exactlemma}
If Assumptions \ref{assumption0}-\ref{assumption4} hold and we choose $\gamma_{k}=k^{\frac{1}{2}}$, $\gamma_{0}=0$ and  $H_{k}=k$ in Algorithm \ref{alg::exactGradient}, then  $\big((\pi_{\theta_{i,k}})_{i=1}^{N}\big)
_{k=1}^{\infty}$ given by Algorithm \ref{alg::exactGradient}  satisfies
\begin{equation}\label{lem52}
\begin{aligned}
\sup _{i \in \mathcal{N}}\frac{\sum_{k=1}^{K} k^{\frac{1}{2}}\left(\sup \limits_{\theta_{i} \in \Theta_{i}}J_{i}^{(\theta_{i}, \theta_{-i,k})}(t)-J_{i}^{(\theta_{i,k}, \theta_{-i,k})}(t)\right)}{\sum_{k=1}^{K}k^{\frac{1}{2}}}
\leqslant& L(K),
\end{aligned}
\end{equation}
and  is  a $k^{\frac{1}{2}}$-weighted  $L(k)$-Nash equilibrium of the game $\Gamma$, where $L(K)=2\sqrt{3}DM_{1}
\Big(4LD+\frac{2D}{\beta}+\frac{\sqrt{N}B_{\Theta}U_{R}}{(1-\gamma)^{2}}
\Big)\frac{1}{K^{\frac{1}{2}}}+2\sqrt{6}\left(1+\sqrt{2}\right)^{\frac{1}{2}} D\left(\frac{1}{\frac{1}{\beta}-L}+2 \beta\right)^{\frac{1}{2}} \left(1+2 L^{2}+\frac{2}{\beta^{2}}\right)^{\frac{1}{4}} \left(\frac{L^{2}+\frac{1}{\beta^{2}}}{e \eta\left(\frac{1}{\beta}-L\right)}\right)^{\frac{1}{4}} \Big(4LD\\+\frac{2D}{\beta}+\frac{\sqrt{N}B_{\Theta}U_{R}}{(1-\gamma)^{2}}
\Big) \frac{M_{1}}{K^{\frac{1}{4}}}$, $M_{1}$ is given by Lemma \ref{gradientd}, $e$ is Euler's number and $D=\sqrt{\sum_{i=1}^{N} D_{i}^{2}}$.
\end{lemma}

\begin{proof}
By  Lemma  \ref{propertyofprox} (ii), we have
\begin{align*}
\frac{\left\|z_{i}^{h+1}-\theta_{i}\right\|^{2}}{2} \leqslant &
\frac{\left\|z_{i}^{h}-\theta_{i}\right\|^{2}}{2}-
\left\langle \eta F_{i,k}(\theta^{h+1}), \theta_{i}^{h+1}-\theta_{i}\right\rangle+\frac{\left\|\eta F_{i,k}(\theta^{h+1})-\eta F_{i,k}(\theta^{h})\right\|^{2}}{2}\\
&-\frac{\left\| \theta_{i}^{h+1}-z_{i}^{h}\right\|^{2}}{2}, \forall \ i \in \mathcal{N}.
\end{align*}
Taking summation for both sides of the above inequality from $i=1$ to $N$ and rearranging the above inequality lead to
\begin{equation}\label{prox12}
\begin{aligned}
\left\langle \eta F_{k}(\theta^{h+1}), \theta^{h+1}-\theta\right\rangle \leqslant &\frac{\left\|z^{h}-\theta\right\|^{2}}{2}-
\frac{\left\|z^{h+1}-\theta \right\|^{2}}{2} +\frac{\left\|\eta F_{k}(\theta^{h+1})-\eta F_{k}(\theta^{h})\right\|^{2}}{2}\\
&-\frac{\left\| \theta^{h+1}-z^{h}\right\|^{2}}{2}.
\end{aligned}
\end{equation}
For the third term on the right side of the above inequality, by Assumptions \ref{assumption0}-\ref{assumption2} and Lemma
\ref{stronglylipschitz}, we have
\begin{equation}\nonumber
\begin{aligned}
\frac{\left\|\eta F_{k}(\theta^{h+1})-\eta F_{k}(\theta^{h})\right\|^{2}}{2} \leqslant \eta^{2} \left(L^{2}+\frac{1}{\beta^{2}}\right)\left\| \theta^{h+1}-\theta^{h}\right\|^{2}.
\end{aligned}
\end{equation}
This together with (\ref{prox12}) gives
\begin{equation}\label{prox2}
\begin{aligned}
\left\langle \eta F_{k}(\theta^{h+1}), \theta^{h+1}-\theta\right\rangle \leqslant &\frac{\left\|z^{h}-\theta\right\|^{2}}{2}-
\frac{\left\|z^{h+1}-\theta \right\|^{2}}{2} +\eta^{2} \left(L^{2}+\frac{1}{\beta^{2}}\right)\left\| \theta^{h+1}-\theta^{h}\right\|^{2}\\
&-\frac{\left\| \theta^{h+1}-z^{h}\right\|^{2}}{2}.
\end{aligned}
\end{equation}
For the term $\left\| \theta^{h+1}-\theta^{h}\right\|^{2}$ in the above inequality,  by $C_{2}$ inequality, we have
\begin{equation}\label{Young}
\begin{aligned}
\left\| \theta^{h+1}-\theta^{h}\right\|^{2}\leqslant 2\left\| \theta^{h+1}-z^{h}\right\|^{2}+ 2\left\| \theta^{h}-z^{h}\right\|^{2}.
\end{aligned}
\end{equation}
For the first term on the right side of (\ref{Young}), by the non-expansion property of the proximal mapping in Lemma \ref{propertyofprox} (iii), Assumptions \ref{assumption0}-\ref{assumption2} and Lemma \ref{stronglylipschitz}, we have
\begin{align*}
 \left\| \theta^{h}-z^{h}\right\|^{2}=\sum_{i=1}^{N} \left\| \theta^{h}_{i}-z^{h}_{i}\right\|^{2}& \leqslant \sum_{i=1}^{N}\left\|\eta F_{i,k}(\theta^{h-1})-\eta F_{i,k}(\theta^{h})\right\|^{2}\\
 &=\eta^{2}\left\| F_{k}(\theta^{h-1})- F_{k}(\theta^{h})\right\|^{2}\leqslant \eta^{2}\left(2L^{2}+\frac{2}{\beta^{2}}\right)\left\| \theta^{h-1}-\theta^{h}\right\|^{2}.
\end{align*}
This together with $(\ref{Young})$ gives
\begin{equation}\nonumber
\begin{aligned}
\left\| \theta^{h+1}-\theta^{h}\right\|^{2}\leqslant 2\left\| \theta^{h+1}-z^{h}\right\|^{2}+ 2\eta^{2}\left(2L^{2}+\frac{2}{\beta^{2}}\right)\left\| \theta^{h-1}-\theta^{h}\right\|^{2}.
\end{aligned}
\end{equation}
By $\left\| \theta^{h+1}-\theta^{h}\right\|^{2}=
2\left\| \theta^{h+1}-\theta^{h}\right\|^{2}-\left\| \theta^{h+1}-\theta^{h}\right\|^{2}$ and the above inequality, we have
\begin{equation}\nonumber
\begin{aligned}
\left\| \theta^{h+1}-\theta^{h}\right\|^{2}
\leqslant 4
\left\| \theta^{h+1}-z^{h}\right\|^{2}
+4\eta^{2}\left(2L^{2}+\frac{2}{\beta^{2}}\right)\left\| \theta^{h-1}-\theta^{h}\right\|^{2}-\left\| \theta^{h+1}-\theta^{h}\right\|^{2}.
\end{aligned}
\end{equation}
Combining (\ref{prox2})  with the above inequality gives
\begin{equation}\nonumber
\begin{aligned}
\left\langle \eta F_{k}(\theta^{h+1}), \theta^{h+1}-\theta\right\rangle \leqslant & \frac{\left\|z^{h}-\theta\right\|^{2}}{2}-
\frac{\left\|z^{h+1}-\theta \right\|^{2}}{2} +\left(4\eta^{2} \left(L^{2}+\frac{1}{\beta^{2}}\right)-\frac{1}{2}\right)\left\| \theta^{h+1}-z^{h}\right\|^{2}\\
&+8\eta^{4} \left(L^{2}+\frac{1}{\beta^{2}}\right)^{2}\left\| \theta^{h-1}-\theta^{h}\right\|^{2}-\eta^{2} \left(L^{2}+\frac{1}{\beta^{2}}\right)\left\| \theta^{h+1}-\theta^{h}\right\|^{2}.
\end{aligned}
\end{equation}
According to Assumptions
\ref{assumption0}-\ref{assumption2} and Lemma
\ref{stronglymonotone}, we know that $F_{k}(\theta)$ is $\left(\frac{1}{\beta}-L\right)$-strongly monotone and then SVI($F_{k},\Theta)$ has a unique solution ~\citep{KINDERLEHRER1980}, which is  denoted by $\overline{\theta_{k}}$. Substituting $\overline{\theta_{k}}$ for $\theta$ in the above inequality leads to
\begin{align}
&\left\langle \eta F_{k}(\theta^{h+1}), \theta^{h+1}-\overline{\theta_{k}}\right\rangle \notag\\
\leqslant & \frac{\left\|z^{h}-\overline{\theta_{k}}\right\|^{2}}{2}-
\frac{\left\|z^{h+1}-\overline{\theta_{k}} \right\|^{2}}{2} +\left(4\eta^{2} \left(L^{2}+\frac{1}{\beta^{2}}\right)-\frac{1}{2}\right)\left\| \theta^{h+1}-z^{h}\right\|^{2}
+8\eta^{4} \left(L^{2}+\frac{1}{\beta^{2}}\right)^{2}\notag\\
&\times \left\| \theta^{h-1}-\theta^{h}\right\|^{2}-\eta^{2} \left(L^{2}+\frac{1}{\beta^{2}}\right)\left\| \theta^{h+1}-\theta^{h}\right\|^{2}.\label{prox223}
\end{align}
Noting that $\overline{\theta_{k}}$ solves SVI($F_{k},\Theta)$, by $C_{2}$ inequality and Lemma \ref{stronglymonotone}, we have
\begin{align}\label{prox224}
\eta\left(\frac{1}{\beta}-L\right)\frac{\left\|z^{h}-
\overline{\theta_{k}}\right\|^{2}}{2}
-\eta\left(\frac{1}{\beta}-L\right)\left\|z^{h}-
\theta^{h+1}\right\|^{2}
&\leqslant \eta\left(\frac{1}{\beta}-L\right)\left\|\theta^{h+1}-
\overline{\theta_{k}}\right\|^{2}\notag\\
& \leqslant \eta \left\langle F_{k}(\theta^{h+1})-F_{k}(\overline{\theta_{k}}),
\theta^{h+1}-\overline{\theta_{k}}\right\rangle \notag\\
& \leqslant \left\langle \eta F_{k}(\theta^{h+1}),
\theta^{h+1}-\overline{\theta_{k}}\right\rangle.
\end{align}
Taking summation for both sides of  (\ref{prox223}) and the above inequality gives
\begin{equation}\nonumber
\begin{aligned}
\frac{\left\|z^{h+1}-\overline{\theta_{k}} \right\|^{2}}{2} \leqslant & \left(1-\eta\left(\frac{1}{\beta}-L\right)\right)
\frac{\left\|z^{h}-\overline{\theta_{k}}\right\|^{2}}{2}+8\eta^{4} \left(L^{2}+\frac{1}{\beta^{2}}\right)^{2}\left\| \theta^{h-1}-\theta^{h}\right\|^{2}-\eta^{2} \bigg(L^{2}
\\& +\frac{1}{\beta^{2}}\bigg)
\left\| \theta^{h+1}-\theta^{h}\right\|^{2}+\left[4\eta^{2} \left(L^{2}+\frac{1}{\beta^{2}}\right)+
\eta\left(\frac{1}{\beta}-L\right)-\frac{1}{2}\right] \left\| \theta^{h+1}-z^{h}\right\|^{2}.
\end{aligned}
\end{equation}
From $\eta < \frac{-\left(\frac{1}{\beta}-L\right)+
\left(\left(\frac{1}{\beta}-L\right)^{2}+8\left(L^{2}+\frac{1}{\beta^{2}}\right)\right)^{\frac{1}{2}}
}{8}$, we know that
$4\eta^{2} \left(L^{2}+\frac{1}{\beta^{2}}\right)+\eta\left(\frac{1}{\beta}-L\right)-\frac{1}{2} \leqslant 0$, which together with the above inequality gives
\begin{equation}\nonumber
\begin{aligned}
\frac{\left\|z^{h+1}-\overline{\theta_{k}} \right\|^{2}}{2} \leqslant & \left(1-\eta\left(\frac{1}{\beta}-L\right)\right)
\frac{\left\|z^{h}-\overline{\theta_{k}}\right\|^{2}}{2}+
8\eta^{4} \left(L^{2}+\frac{1}{\beta^{2}}\right)^{2}\left\| \theta^{h-1}-\theta^{h}\right\|^{2}\\
& -\eta^{2} \left(L^{2}+\frac{1}{\beta^{2}}\right)\left\| \theta^{h+1}-\theta^{h}\right\|^{2}.
\end{aligned}
\end{equation}
Then, from $\eta < \frac{-\left(\frac{1}{\beta}-L\right)+
\left(\left(\frac{1}{\beta}-L\right)^{2}+32\left(L^{2}+\frac{1}{\beta^{2}}\right)\right)^{\frac{1}{2}}
}{16}$, we know that $8\eta^{4} \left(L^{2}+\frac{1}{\beta^{2}}\right)^{2} \leqslant \eta^{2} \Big(L^{2}\\+\frac{1}{\beta^{2}}\Big) \left(1-\eta\left(\frac{1}{\beta}-L\right)\right)
$. This together with the above inequality gives
\begin{equation}\nonumber
\begin{aligned}
\frac{\left\|z^{h+1}-\overline{\theta_{k}} \right\|^{2}}{2} +\eta^{2} \Big(L^{2}+\frac{1}{\beta^{2}}\Big)\left\| \theta^{h+1}-\theta^{h}\right\|^{2} \leqslant & l_{1}
\Bigg(\frac{\left\|z^{h}-\overline{\theta_{k}}\right\|^{2}}{2}
+\eta^{2} \left(L^{2}+\frac{1}{\beta^{2}}\right)\\
& \times\left\| \theta^{h}-\theta^{h-1}\right\|^{2}\Bigg),
\end{aligned}
\end{equation}
where $l_{1}=1-\eta\big(\frac{1}{\beta}-L\big)$. From $0<\eta < \frac{1}{\frac{1}{\beta}-L}$, it follows that $l_{1} \in (0,1)$. By Assumption \ref{assumption1} and the above inequality, we have
\begin{equation}\nonumber
\begin{aligned}
&\frac{\left\|z^{H_{k}+1}-\overline{\theta_{k}} \right\|^{2}}{2} +\eta^{2} \left(L^{2}+\frac{1}{\beta^{2}}\right)\left\| \theta^{H_{k}+1}-\theta^{H_{k}}\right\|^{2}\\
\leqslant  & l_{1}^{H_{k}}\left(\frac{\left\|z^{1}-\overline{\theta_{k}}\right\|^{2}}{2}+\eta^{2} \left(L^{2}+\frac{1}{\beta^{2}}\right)\left\| \theta^{1}-\theta^{0}\right\|^{2}\right)
\leqslant l_{1}^{H_{k}} \left(1+2\left(L^{2}+\frac{1}{\beta^{2}}\right)\right)2D^{2},
\end{aligned}
\end{equation}
and thus
\begin{equation}\label{prox2261}
\begin{aligned}
\frac{\left\|z^{H_{k}+1}-\overline{\theta_{k}} \right\|^{2}}{2} \leqslant l_{1}^{H_{k}} \left(1+2\left(L^{2}+\frac{1}{\beta^{2}}\right)\right)2D^{2}.
\end{aligned}
\end{equation}
By Assumptions \ref{assumption0}-\ref{assumption2}, Lemma \ref{stronglymonotone} and Lemma \ref{contorolstrong}, we have
\begin{equation}\nonumber
\left\langle F_{k}\big(\theta_{k+1}\big), \theta_{k+1}-\theta\right\rangle \leqslant
2D\left(2L^{2}+\frac{2}{\beta^{2}}\right)^{\frac{1}{2}}\big(2+\sqrt{2}\big)\left\|z^{H_{k}+1}-\overline{\theta_{k}} \right\|.
\end{equation}
Substituting (\ref{prox2261}) into the above inequality gives
\begin{equation}\label{prox228}
\begin{aligned}
\left\langle F_{k}\big(\theta_{k+1}\big), \theta_{k+1}-\theta\right\rangle & \leqslant C_{1}l_{1}^{\frac{H_{k}}{2}},\ \forall \ \theta \in \Theta,
\end{aligned}
\end{equation}
where $C_{1}=8\big(1+\sqrt{2}\big)D^{2}
\left(L^{2}+\frac{1}{\beta^{2}}\right)^{\frac{1}{2}}
\left(1+2\left(L^{2}+\frac{1}{\beta^{2}}\right)\right)^{\frac{1}{2}}$. Noting that $\overline{\theta_{k}}$ is the solution of  SVI($F_{k},\Theta)$, by
(\ref{prox228}), we have
\begin{equation}\label{thetawidetheta}
\begin{aligned}
\left\langle  F_{k}\big(\theta_{k+1}\big)-F_{k}\big(
\overline{\theta_{k}}\big), \theta_{k+1}-\overline{\theta_{k}}\right\rangle =&\left\langle F_{k}\big(\theta_{k+1}\big), \theta_{k+1}-\overline{\theta_{k}}\right\rangle+\left\langle -F_{k}\big(\overline{\theta_{k}}\big), \theta_{k+1}-\overline{\theta_{k}}\right\rangle\\
\leqslant & C_{1}l_{1}^{\frac{H_{k}}{2}}.
\end{aligned}
\end{equation}
Combining  Assumptions \ref{assumption0}-\ref{assumption2} and Lemma \ref{stronglymonotone} with the above inequality implies
\begin{equation}\label{answerofstrong}
\left(\frac{1}{\beta}-L\right) \left\|\theta_{k+1}-\overline{\theta_{k}}\right\|^{2} \leqslant \left\langle  F_{k}\big(\theta_{k+1}\big)-F_{k}\big(\overline{\theta}_{k}\big), \theta_{k+1}-\overline{\theta}_{k}\right\rangle \leqslant C_{1}l_{1}^{\frac{H_{k}}{2}}.
\end{equation}
By Assumption \ref{assumption4}, we know that MVI$(F,\Theta)$ has a solution, which is denoted by $\theta_{*}$. This together with $F\big(\theta_{k+1}\big)=F_{k}\big(\theta_{k+1}\big)
-\frac{1}{\beta}$\
$\big(\theta_{k+1}-\theta_{k}\big)$ gives
\begin{equation}\nonumber
 \left\langle \beta F_{k}\big(\theta_{k+1}\big)-\big(\theta_{k+1}-\theta_{k}\big), \theta_{k+1}-\theta_{*}\right\rangle \geqslant 0.
\end{equation}
Then, combining  $(\ref{prox228})$ with the above inequality leads to
\begin{equation}\nonumber
-\beta \left(C_{1}l_{1}^{\frac{H_{k}}{2}}\right) \leqslant
\left\langle \beta F_{k}\big(\theta_{k+1}\big), \theta_{*}-\theta_{k+1}\right\rangle \leqslant
 \left\langle  \theta_{k+1}-\theta_{k}, \theta_{*}-\theta_{k+1}\right\rangle,
\end{equation}
that is,
\begin{equation}\nonumber
-\beta \left(C_{1}l_{1}^{\frac{H_{k}}{2}}\right)  \leqslant
 \left\langle  \theta_{k+1}-\theta_{k}, \theta_{*}-\theta_{k+1}\right\rangle .
\end{equation}
This together with $\frac{1}{2}\left\|\theta_{k}-\overline{\theta_{k}}
\right\|^{2}+2\left\|\overline{\theta_{k}}-\theta_{k+1}\right\|^{2} \geqslant -2\left\langle \theta_{k}-\overline{\theta}_{k}, \overline{\theta}_{k}-\theta_{k+1} \right\rangle$ gives
\begin{align*}
\left\|\theta_{k}-\theta_{*}\right\|^{2}
=&\left\|\theta_{k}-\theta_{k+1}+
\theta_{k+1}-\theta_{*}\right\|^{2}\\
=&\left\|\theta_{k}-\theta_{k+1}\right\|^{2}+
\left\|\theta_{k+1}-\theta_{*}\right\|^{2}+
2\left\langle \theta_{k}-\theta_{k+1}, \theta_{k+1}-\theta_{*}\right\rangle\\
\geqslant & \left\|\theta_{k}-\theta_{k+1}\right\|^{2}+
\left\|\theta_{k+1}-\theta_{*}\right\|^{2}-2\beta \left(C_{1}l_{1}^{\frac{H_{k}}{2}}\right) \\
=&\left\|\theta_{k}-\overline{\theta_{k}}
\right\|^{2}+\left\|\overline{\theta_{k}}-\theta_{k+1}\right\|^{2} +2\left\langle \theta_{k}-\overline{\theta}_{k}, \overline{\theta}_{k}-\theta_{k+1} \right\rangle\\
&+\left\|\theta_{k+1}-\theta_{*}\right\|^{2}-2\beta \left(C_{1}l_{1}^{\frac{H_{k}}{2}}\right)\\
\geqslant &\left\|\theta_{k}-\overline{\theta_{k}}
\right\|^{2}+\left\|\overline{\theta_{k}}-
\theta_{k+1}\right\|^{2}-\frac{1}{2} \left\|\theta_{k}-\overline{\theta_{k}}
\right\|^{2}-2\left\|\overline{\theta_{k}}-\theta_{k+1}\right\|^{2}\\
&+\left\|\theta_{k+1}-\theta_{*}\right\|^{2}-2\beta \left(C_{1}l_{1}^{\frac{H_{k}}{2}}\right).
\end{align*}
Then we have
\begin{equation}\nonumber
\begin{aligned}
\frac{1}{2} \left\|\theta_{k}-\overline{\theta_{k}}
\right\|^{2} \leqslant \left\|\theta_{k}-\theta_{*}\right\|^{2} -
\left\|\theta_{k+1}-\theta_{*}\right\|^{2}+
\left\|\overline{\theta_{k}}-\theta_{k+1}\right\|^{2}+2\beta C_{1}l_{1}^{\frac{H_{k}}{2}}.
\end{aligned}
\end{equation}
Multiply both sides of
the above inequality  by $\gamma_{k}$  and
take summation for both sides  from $i=1$ to $N$. Then, by (\ref{answerofstrong}), we have
\begin{align}\nonumber
& \frac{1}{2} \sum_{k=1}^{K} \gamma_{k}\left\|\theta_{k}-\overline{\theta_{k}}
\right\|^{2} \\ \leqslant &\sum_{k=1}^{K} \gamma_{k}\big(\left\|\theta_{k}-\theta_{*}\right\|^{2} -
\left\|\theta_{k+1}-\theta_{*}\right\|^{2}\big)+\left(\frac{1}{\frac{1}{\beta}-L}+2\beta\right)\sum_{k=1}^{K}\gamma_{k}  \left(C_{1}l_{1}^{\frac{H_{k}}{2}}\right)\notag\\
\leqslant & \sum_{k=1}^{K} \big(\gamma_{k-1} \left\|\theta_{k}-\theta_{*}\right\|^{2} -\gamma_{k}
\left\|\theta_{k+1}-\theta_{*}\right\|^{2}\big)+\sum_{k=1}^{K} \big(\gamma_{k}-\gamma_{k-1}\big) \left\|\theta_{k}-\theta_{*}\right\|^{2}\notag\\
&+\left(\frac{1}{\frac{1}{\beta}-L}
+2\beta\right)\sum_{k=1}^{K}\gamma_{k}  \left(C_{1}l_{1}^{\frac{H_{k}}{2}}\right)\notag\\
=&\gamma_{0}\left\|\theta_{1}-\theta_{*}\right\|^{2}
-\gamma_{K}\left\|\theta_{K+1}-\theta_{*}\right\|^{2}+\big(\gamma_{K}-\gamma_{0}\big)4D^{2}+
\left(\frac{1}{\frac{1}{\beta}-L}+2\beta\right)\sum_{k=0}^{K-1}\gamma_{k}  \left(C_{1}l_{1}^{\frac{H_{k}}{2}}\right)\notag\\
\leqslant & 4\gamma_{K}D^{2}+
\left(\frac{1}{\frac{1}{\beta}-L}+2\beta\right)\sum_{k=1}^{K}\gamma_{k}  \left(C_{1}l_{1}^{\frac{H_{k}}{2}}\right)\notag.
\end{align}
Dividing  both sides of the above inequality by $\sum_{k=1}^{K}\gamma_{k}$, we have
\begin{equation}\label{nearfinal1}
\begin{aligned}
\frac{\sum_{k=1}^{K} \gamma_{k}\left\|\theta_{k}-\overline{\theta_{k}}
\right\|^{2}}{\sum_{k=1}^{K}\gamma_{k}} \leqslant 8\frac{\gamma_{K}D^{2}}{\sum_{k=1}^{K} \gamma_{k}}+\frac{2\left(\frac{1}{\frac{1}{\beta}-L}+
2\beta\right)\sum_{k=1}^{K}\gamma_{k}  C_{1}l_{1}^{\frac{H_{k}}{2}}}{\sum_{k=1}^{K} \gamma_{k}}.
\end{aligned}
\end{equation}
From $l_{1}=1-\eta\left(\frac{1}{\beta}-L\right) \in (0,1)$, we know that $l_{1}^{H_{k}} < \frac{1}{eH_{k}\eta\left(\frac{1}{\beta}-L\right)}$, $l_{1}^{\frac{H_{k}}{2}}< \frac{1}{H_{k}^{\frac{1}{2}}\left(e\eta\left(\frac{1}{\beta}-L\right)\right)^{\frac{1}{2}}}$. This together with the above inequality implies
\begin{equation}\nonumber
\begin{aligned}
\frac{\sum_{k=1}^{K} \gamma_{k}\left\|\theta_{k}-\overline{\theta_{k}}
\right\|^{2}}{\sum_{k=1}^{K}\gamma_{k}} \leqslant 8\frac{\gamma_{K}D^{2}}{\sum_{k=1}^{K} \gamma_{k}}+\frac{2\left(\frac{1}{\frac{1}{\beta}-L}
+2\beta\right)\sum_{k=1}^{K}\gamma_{k}  C_{1}\frac{1}{H_{k}^{\frac{1}{2}}\left(e\eta\left(\frac{1}{\gamma}
-L\right)\right)^{\frac{1}{2}}}}{\sum_{k=1}^{K} \gamma_{k}}.
\end{aligned}
\end{equation}
Combining $H_{k}=k$, $\gamma_{k}=k^{\frac{1}{2}}$ and
$\sum_{k=1}^{K} k^{\frac{1}{2}} \geqslant \frac{2}{3}K^
{\frac{3}{2}}$ with the above inequality, we have
\begin{equation}
\begin{aligned}
\frac{\sum_{k=1}^{K} k^{\frac{1}{2}}\left\|\theta_{k}-\overline{\theta_{k}}
\right\|^{2}}{\sum_{k=1}^{K}k^{\frac{1}{2}}}
\leqslant 12D^{2}\frac{1}{K}+3
\left(\frac{1}{\frac{1}{\beta}-L}+2\beta\right)
 \frac{C_{1}}{\left(e\eta\left(\frac{1}{\beta}-L\right)\right)^{\frac{1}{2}}}
 \frac{1}{K^{\frac{1}{2}}}.
\end{aligned}
\end{equation}
Applying Lyapunov inequality into the above inequality leads to
\begin{align}
\frac{\sum_{k=1}^{K} k^{\frac{1}{2}}\left\|\theta_{k}-\overline{\theta_{k}}
\right\|}{\sum_{k=1}^{K}k^{\frac{1}{2}}}
\leqslant &\left(\frac{\sum_{k=1}^{K} k^{\frac{1}{2}}\left\|\theta_{k}-\overline{\theta_{k}}
\right\|^{2}}{\sum_{k=1}^{K}k^{\frac{1}{2}}}\right)^{\frac{1}{2}}\notag\\
\leqslant & 2\sqrt{3}D\frac{1}{K^{\frac{1}{2}}}+\left(3
\left(\frac{1}{\frac{1}{\beta}-L}+2\beta\right)
\frac{C_{1}}{\left(e\eta\left(\frac{1}{\beta}-L\right)\right)^{\frac{1}{2}}}\right)^{\frac{1}{2}}
\frac{1}{K^{\frac{1}{4}}}.\label{realandal}
\end{align}
By Assumptions \ref{assumption0}-\ref{assumption2}, Lemma \ref{FLIPSCHITZ} and Lemma \ref{weaklyandstrongly}, we have
\begin{align}
\langle F(\theta_{k}), \theta_{k}-\theta\rangle=&\langle F(\theta_{k})-F(\overline{\theta_{k}}), \theta_{k}-\overline{\theta_{k}}\rangle+\langle F(\theta_{k})-F(\overline{\theta_{k}}),\overline{\theta_{k}}-\theta\rangle\notag\\
&+\langle F(\overline{\theta_{k}}), \theta_{k}-\overline{\theta_{k}}\rangle+\langle F(\overline{\theta_{k}}),\overline{\theta_{k}}-\theta\rangle\notag\\
\leqslant&  2LD\left\|\theta_{k}-\overline{\theta_{k}}\right\|
+2LD\left\|\theta_{k}-\overline{\theta_{k}}\right\|
+\left\|F(\overline{\theta_{k}})\right\|\left\|\theta_{k}-\overline{\theta_{k}}\right\|
+\frac{2D}{\beta}\left\|\theta_{k}-\overline{\theta_{k}}\right\|.\label{prox230}
\end{align}
By Proposition \ref{policygradient}, we have
\begin{align*}
\left\|F(\theta)\right\|^{2}=\sum_{n=1}^{N}\left\|\nabla_{\theta_{i}} J_{i}^{\theta}(t)\right\|^{2}
\leqslant & \sum_{n=1}^{N} \bigg\|\int_{ \mathcal{S} \times \mathcal{A}} \int_{ \mathcal{S}}\sum_{l=t}^{\infty} \gamma^{l-t} \rho^{\theta}\big(s(l)=s^{\prime} \mid s(t)=s\big) \nabla_{\theta_{i}} \pi_{\theta_{i}}(a_{i} \mid s^{\prime})  \\ &\pi_{\theta_{-i}}(a_{-i} \mid s^{\prime})
 Q_{i}^{\theta}(s^{\prime}, a,l) \rho_{t}(s)\mathrm{d} s \mathrm{d} s^{\prime} \mathrm{d} a \bigg\| ^{2}.
\end{align*}
From (\ref{valuefunction}), (\ref{qvaluefunction}) and Assumption \ref{assumption1}, it follows that
$|Q_{i}^{\theta}(s^{\prime}, a,l)| \leqslant \frac{U_{R}}{1-\gamma}$, which together with Assumption \ref{assumption0}, Assumption \ref{assumption2} and the above inequality implies
\begin{align}\label{boundedoff}
\left\|F(\overline{\theta_{k}})\right\|
&\leqslant
\left(\sum_{n=1}^{N}\left\|\sum_{l=t}^{\infty} \gamma^{l-t}  B_{\Theta} \frac{U_{R}}{1-\gamma} \int_{ \mathcal{S} \times \mathcal{A}} \int_{ \mathcal{S}}\rho^{\theta_{k}}\big(s(l)=s^{\prime} \mid s(t)=s\big) \rho_{t}(s)\mathrm{d} s\mathrm{d} s^{\prime} \mathrm{d} a \right\|^{2}\right)^{\frac{1}{2}}\notag\\
&\leqslant \left(\sum_{n=1}^{N}\left\|\sum_{l=t}^{\infty} \gamma^{l-t}  B_{\Theta} \frac{U_{R}}{1-\gamma}\right\|^{2}\right)^{\frac{1}{2}} \leqslant \frac{\sqrt{N}B_{\Theta}U_{R}}{(1-\gamma)^{2}}.
\end{align}
By (\ref{prox230}) and the above inequality, we have
\begin{equation}\label{prox231}
\begin{aligned}
\langle F(\theta_{k}), \theta_{k}-\theta\rangle\leqslant \left(4LD+\frac{2D}{\beta}+\frac{\sqrt{N}B_{\Theta}U_{R}}{(1-\gamma)^{2}}
\right)\left\|\theta_{k}-\overline{\theta_{k}}\right\|.
\end{aligned}
\end{equation}
Then, combining (\ref{realandal}) with the above inequality gives
\begin{equation}\nonumber
\begin{aligned}
\frac{\sum_{k=1}^{K} k^{\frac{1}{2}}\left(\sup \limits_{\theta \in \Theta} \left\langle F\big(\theta_{k}\big), \theta_{k}-\theta \right\rangle\right)}
{\sum_{k=1}^{K}k^{\frac{1}{2}}}
\leqslant&
\frac{\sum_{k=1}^{K} k^{\frac{1}{2}}\left(4LD+\frac{2D}{\beta}+\frac{\sqrt{N}B_{\Theta}U_{R}}{(1-\gamma)^{2}}
\right)\left\|\theta_{k}-\overline{\theta_{k}}
\right\|}{\sum_{k=1}^{K}k^{\frac{1}{2}}}\\
\leqslant& 2\sqrt{3}D\left(4LD+\frac{2D}{\beta}+\frac{\sqrt{N}B_{\Theta}U_{R}}{(1-\gamma)^{2}}
\right)\frac{1}{K^{\frac{1}{2}}}
+C_{2}\frac{1}{K^{\frac{1}{4}}},
\end{aligned}
\end{equation}
where $C_{2}=\left(4LD+\frac{2D}{\beta}+\frac{\sqrt{N}B_{\Theta}U_{R}}{(1-\gamma)^{2}}
\right)\left(3
\left(\frac{1}{\frac{1}{\beta}-L}+2\beta\right)
\frac{C_{1}}{\left(e\eta\left(\frac{1}{\beta}-L\right)\right)^{\frac{1}{2}}}\right)^{\frac{1}{2}}
$.
Note that \\ $\sup \limits_{\theta_{i} \in \Theta_{i}} \big\langle F_{i}\big(\theta_{k}\big), \theta_{i,k}-\theta_{i}\big\rangle \leqslant  \sup \limits_{\theta \in \Theta} \left\langle F\big(\theta_{k}\big), \theta_{k}-\theta \right\rangle$. This together with the above inequality gives
\begin{align*}
\frac{\sum_{k=1}^{K} k^{\frac{1}{2}}\left( \sup \limits_{\theta_{i} \in \Theta_{i}} \left\langle F_{i}(\theta_{k}), \theta_{i,k}-\theta_{i}\right\rangle\right)}
{\sum_{k=1}^{K}k^{\frac{1}{2}}}
\leqslant& 2\sqrt{3}D\left(4LD+\frac{2D}{\beta}+\frac{\sqrt{N}B_{\Theta}U_{R}}{(1-\gamma)^{2}}
\right)\frac{1}{K^{\frac{1}{2}}}+C_{2}\frac{1}{K^{\frac{1}{4}}}.
\end{align*}
Then, from Assumptions \ref{assumption0}-\ref{assumption3}, Lemma \ref{gradientd} and the above inequality, it follows that
\begin{align*}
&\sup _{i \in \mathcal{N}}\frac{\sum_{k=1}^{K} k^{\frac{1}{2}}\left(\sup \limits_{\theta_{i} \in \Theta_{i}}J_{i}^{(\theta_{i}, \theta_{-i,k})}(t)-J_{i}^{(\theta_{i,k}, \theta_{-i,k})}(t)\right)}{\sum_{k=1}^{K}k^{\frac{1}{2}}}\\
\leqslant & \sup _{i \in \mathcal{N}}
\frac{\sum_{k=1}^{K} k^{\frac{1}{2}}M_{1}\left( \sup \limits _{\theta_{i} \in \Theta_{i}} \left\langle F_{i}(\theta_{k}), \theta_{i,k}-\theta_{i}\right\rangle\right)}{\sum_{k=1}^{K}k^{\frac{1}{2}}}\\
\leqslant& 2\sqrt{3}DM_{1}\left(4LD+\frac{2D}{\beta}
+\frac{\sqrt{N}B_{\Theta}U_{R}}{(1-\gamma)^{2}}
\right)\frac{1}{K^{\frac{1}{2}}}+C_{2}M_{1}\frac{1}{K^{\frac{1}{4}}}
= L(K).
\end{align*}
 By the above inequality,
for any countable set $\widetilde{\Theta_{i}} \subseteq \Theta_{i}$, we have
\begin{align*}
& \sup _{i \in \mathcal{N}}\frac{\sum_{k=1}^{K} k^{\frac{1}{2}}\left(\sup \limits_{\theta_{i} \in \widetilde{\Theta_{i}}}J_{i}^{(\theta_{i}, \theta_{-i,k})}(t)-J_{i}^{(\theta_{i,k}, \theta_{-i,k})}(t)\right)}{\sum_{k=1}^{K}k^{\frac{1}{2}}}\\
\leqslant &
\sup _{i \in \mathcal{N}}\frac{\sum_{k=1}^{K} k^{\frac{1}{2}}\left(\sup \limits_{\theta_{i} \in \Theta_{i}}J_{i}^{(\theta_{i}, \theta_{-i,k})}(t)-J_{i}^{(\theta_{i,k}, \theta_{-i,k})}(t)\right)}{\sum_{k=1}^{K}k^{\frac{1}{2}}}
\leqslant  L(K),
\end{align*}
that is, (\ref{lem52}) holds. Then from Definition \ref{nashweightapp}, we know $\big((\pi_{\theta_{i,k}})_{i=1}^{N}\big)
_{k=1}^{\infty}$ is a $k^{\frac{1}{2}}$ is a $k^{\frac{1}{2}}$-weighted $L(k)$-Nash equilibrium of the game $\Gamma$.
\end{proof}

\begin{theorem}\label{exactgradient}
If Assumptions \ref{assumption0}-\ref{assumption4} hold and we choose $\gamma_{k}=k^{\frac{1}{2}}$, $\gamma_{0}=0$ and $H_{k}=k$, then $\big((\pi_{\theta_{i,k}})_{i=1}^{N}\big)
_{k=1}^{\infty}$ given by Algorithm \ref{alg::exactGradient} is a $k^{\frac{1}{2}}$-weighted asymptotic Nash equilibrium of the game $\Gamma$ and
$\lim \limits_{K\rightarrow \infty}\sup\limits_{i \in \mathcal{N}}
E\bigg[\sup\limits_{\theta_{i}\in \widetilde{\Theta_{i}}}J_{i}^{(\theta_{i}, \theta_{-i,\tau_K})}(t)-J_{i}^{(\theta_{i,\tau_K}, \theta_{-i,\tau_K})}(t)\bigg | \theta_{k},k=1,2,...,K\bigg]=0$, where $\widetilde{\Theta_{i}}$ is any countable subset of $\Theta_{i}$.
\end{theorem}
\begin{proof}
If  Assumptions \ref{assumption0}-\ref{assumption4} hold,  by Lemma \ref{exactlemma}, we have
\begin{equation}\nonumber
\begin{aligned}
\sup _{i \in \mathcal{N}}\frac{\sum_{k=1}^{K} k^{\frac{1}{2}}\left(\sup \limits_{\theta_{i} \in \widetilde{\Theta_{i}}}J_{i}^{(\theta_{i}, \theta_{-i,k})}(t)-J_{i}^{(\theta_{i,k}, \theta_{-i,k})}(t)\right)}{\sum_{k=1}^{K}k^{\frac{1}{2}}}
\leqslant& L(K).
\end{aligned}
\end{equation}
 Noting that $\lim \limits_{k\rightarrow \infty} L(k)=0$, by Theorem
\ref{weightuninoweight}, we know that  $\big((\pi_{\theta_{i,k}})_{i=1}^{N}\big)
_{k=1}^{\infty}$
is a $k^{\frac{1}{2}}$-weighted asymptotic Nash equilibrium of the game $\Gamma$.  Then, by Theorem  $\ref{connectionofne}$, it follows that \\ $\lim \limits_{K\rightarrow \infty}\sup\limits_{i \in \mathcal{N}}
E\bigg[\sup\limits_{\theta_{i} \in \widetilde{\Theta_{i}}}J_{i}^{(\theta_{i}, \theta_{-i,\tau_K})}(t)-J_{i}^{(\theta_{i,\tau_K}, \theta_{-i,\tau_K})}(t)\biggm|\theta_{k},k=1,2,...,K\bigg]=0$.
\end{proof}

\section{Learning Nash Equilibria with Unknown Pseudo Gradients}
In this section, we design Algorithm  \ref{alg::pseudoGradient} for learning  Nash equilibria with  unknown pseudo gradients. We employ Monte-Carlo simulations for the estimations of pseudo gradients. Similar to Algorithm \ref{alg::exactGradient}, Algorithm \ref{alg::pseudoGradient}  is also a two-loop algorithm. In the outer loop, we construct $\hat{F}_{k}(\theta)=\left(\hat{F}_{i}(\theta)+\frac{1}{\beta}(\theta_{i}-
\theta_{i,k})\right)_{i=1}^{N}=\left(-\hat{\nabla}_{\theta_{i}}^{T,K_{1}} J_{i}^{\theta}(t)+\frac{1}{\beta}(\theta_{i}-
\theta_{i,k})\right)_{i=1}^{N}$ and SVI($\hat{F}_{k}, \Theta)$, where $-\hat{\nabla}_{\theta_{i}}^{T,K_{1}} J_{i}^{\theta}(t)$ is the estimation of ${F}_{i}(\theta)$ with $K_{1}$ trajectories and a finite time horizon of length $T$. Then we update the proximal parameter in each iteration. In the inner loop, we provide a single-call extra-gradient algorithm for the constructed variational inequality.
\begin{algorithm}[H]

         \begin{algorithmic}[1]
         \caption{Algorithm for Unknown Pseudo gradients}
          \label{alg::pseudoGradient}
         \State Input: Lipschitz constant $L$, integer $K \geqslant 1$, weight $\gamma_{k}$, initial values $\theta_{1}=(\theta_{i,1})_{i=1}^{N} \in \Theta$,
         \State \ \ \ \ \ \ \ \ \ $\beta \in\left(0, \frac{1}{L}\right)$,  $\widetilde{\eta}=\frac{1}{2\sqrt{L^{2}+\frac{1}
         {\beta^{2}}}}$, number of trajectories $K_{1}$, time horizon $T$.
          \For{$k = 1$, \ldots, $K$}
             \State Input: $l_{2}\in \Big[ \max\Big\{\frac{1}{\beta}-L, 6\left(L^{2}+\frac{1}{\beta^{2}}\right)\Big\}, \infty\Big)$, $l_{1}=\min\Big\{\frac{1}{2\sqrt{l_{2}}}, \frac{1}{4l_{2}}\Big\}$,
             stepsize $\eta_{h}$,
             \State \ \ \ \ \ \ \ \ \ \ initial values $\theta^{1}=(\theta^{1}_{i})_{i=1}^{N}$$=z^{1} \in \Theta$, integer $H_{k}$.
             \State Let $\hat{F}_{k}(\theta) = (\hat{F}_{i,k}(\theta))_{i=1}^{N}= \hat{F}(\theta)+\frac{1}{\beta}
             (\theta-\theta_{k})=\left(\hat{F}_{i}(\theta)
             +\frac{1}{\beta}
             (\theta_{i}-\theta_{i,k})\right)_{i=1}^{N}$.
              \For{$h = 1$, \ldots, $H_{k}$}
                \For{$i=1$, \ldots, $N$}
                \State  Sample $ \hat{\nabla}_{\theta_{i}}^{T,K_{1}} J_{i}^{\theta^{h}}(t)$ by Monte-Carlo simulations.
                 \State $\theta_{i}^{h+1}=\mathop{\arg\min}\limits_{\theta_{i} \in \Theta_{i}} \left\{\left\langle 2\eta_{h} \hat{F}_{i,k}(\theta^{h}), \theta_{i} \right\rangle+
                 \|\theta_{i}-z_{i}^{h}\|^{2}\right\}$.
                 \State  Sample $ \hat{\nabla}_{\theta_{i}}^{T,K_{1}} J_{i}^{\theta^{h+1}}(t)$ by Monte-Carlo simulations.
                 \State $z_{i}^{h+1}=\mathop{\arg\min}\limits_{\theta_{i} \in \Theta_{i}} \left\{\left\langle 2\eta_{h} \hat{F}_{i,k}(\theta^{h+1}), \theta_{i} \right\rangle+
                 \|\theta_{i}-z_{i}^{h}\|^{2}\right\}$.
                 \EndFor
               \EndFor
              \State $z^{H_{k}+1}=\left(z_{i}^{H_{k}+1}\right)_{i=1}^{N}$.
              \State  Sample $ \left(\hat{\nabla}_{\theta_{i}}^{T,K_{1}} J_{i}^{z^{H_{k}+1}}(t)\right)_{i=1}^{N}$ by Monte-carlo simulations.
          \State $\theta_{k+1}=\mathop{\arg\min}\limits_{\theta \in \Theta} \left\{\left\langle 2\widetilde{\eta} \hat{F}_{k}(z^{H_{k}+1}), \theta \right\rangle+\|\theta-z^{H_{k}+1}\|^{2}\right\}$ .
          \EndFor
            \State Randomly choose $\tau_{K}$ satisfying $P(\tau_{K}=k)=\frac{\gamma_{k}}{\sum_{k=1}^{K}\gamma_{k}}$ , $k=1,\ldots, K $.

          \State Output: $\theta_{\tau_{K}}$.
        \end{algorithmic}
 \end{algorithm}

\subsection{Error Analysis of Pseudo Gradient Estimation}
We will give the difference between the estimated pseudo gradient and the real pseudo gradient at first.

By Assumptions \ref{assumption0}-\ref{assumption2}, for agent $i$, $\nabla_{\theta_{i}}  J_{i}^{\theta}(t)$ can be written as
\begin{equation}\label{policygradient1}
\nabla_{\theta_{i}}  J_{i}^{\theta}(t) =\mathbb{E}_{\mathcal{T} \sim \rho_{\theta}}\left[\sum_{l=t}^{\infty}\left(\sum_{\tau=t}^{l} \nabla_{\theta_{i}} \log \pi_{\theta_{i}}\big(a_{i}(\tau,w)\mid s(\tau,w)\big)\big) \gamma^{l-t} r_{i}\big(s(l,w), a(l,w)\right)\Bigg | \theta\right],
\end{equation}
where $\mathcal{T}$ is the trajectory of agent $i$ for the given $\pi_{\theta}$, that is, $\mathcal{T}=\big(s(t,w), a(t,w), \ldots, s(l,w), \\ a(l,w),   \ldots\big)$, $w$ is the sample path and $\rho_{\theta}$ is the probability distribution density function of the trajectory $\mathcal{T}$ \citep{Baxter}.
If the transition probability density function is  unknown,  it's intractable to compute  the expectation of (\ref{policygradient1}). Thus we employ the stochastic estimator of (\ref{policygradient1}).
The G(PO)MDP gradient estimator of $\nabla_{\theta_{i}} J_{i}^{\theta}(t)$ is
\begin{equation}\label{G(PO)MDP}
 \hat{\nabla}_{\theta_{i}} J_{i}^{\theta}(t)=\sum_{l=t}^{\infty}\left(\sum_{\tau=t}^{l} \nabla_{\theta_{i}} \log \pi_{\theta_{i}}(a_{i}(\tau,w)\mid s(\tau,w))\big) \gamma^{l-t} r_{i}\big(s(l,w), a(l,w)\right).
\end{equation}
Notice that  sampling from a single trajectory $\big(s(t,w), a(t,w), ..., s(l,w), a(l,w), \ldots\big)$ in (\ref{G(PO)MDP}) may cause a  high variance of the G(PO)MDP gradient estimator and sampling from an infinite horizon is not tractable in (\ref{G(PO)MDP}). To this end,  ~\cite{Tianyi} proposed  the stochastic estimator with $K_{1}$ trajectories and a finite time horizon of length $T$ which can be expressed as
\begin{equation}\nonumber
\begin{aligned}
&\hat{\nabla}_{\theta_{i}}^{T,K_{1}} J_{i}^{\theta}(t) =
\frac{1}{K_{1}} \sum_{w=1}^{K_{1}} \sum_{l=t}^{T+t}\left(\sum_{\tau=t}^{l} \nabla_{\theta_{i}} \log \pi_{\theta_{i}}(a_{i}(\tau,w)\mid s(\tau,w))\big)  \gamma^{l-t} r_{i}\big(s(l,w), a(l,w)\right).
\end{aligned}
\end{equation}
Denote $\hat{F}(\theta)=\left(\hat{F_{i}}(\theta)\right)_{i=1}^{N}=\left(-\hat{\nabla}_{\theta_{i}}^{T,K_{1}} J_{i}^{\theta}(t) \right)_{i=1}^{N}$ and
\begin{equation}\nonumber
\nabla_{\theta_{i}}^{T} J_{i}^{\theta}(t)=\mathbb{E}_{\mathcal{T} \sim \rho_{\theta}}\left[\sum_{l=t}^{T+t}\left(\sum_{\tau=t}^{l} \nabla_{\theta_{i}} \log \pi_{\theta_{i}}(a_{i}(\tau)\mid s(\tau))\big) \gamma^{l-t} r_{i}\big(s(l), a(l)\right)\Bigg |\theta \right],
\end{equation}
where $\mathcal{T}$ is the trajectory according to the policy $\pi_{\theta}$. Denote $F(\theta,T)=(F_{i}(\theta,T))_{i=1}^{N}=
\big(-\nabla_{\theta_{i}}^{T}J_{i}^{\theta}(t)
\big)_{i=1}^{N}$, $\forall \ \theta \in \Theta$.

\begin{lemma}\label{gradientestimate}
If Assumptions \ref{assumption0}-\ref{assumption2}  hold, then
\begin{equation}\nonumber
P\bigg\{\|\hat{F}(\theta)-F(\theta)\|^{2} \leqslant M\big(T, K_{1}, \delta\big)\bigg\} \geqslant 1-\frac{\delta}{4K_{1}}, \forall \ \theta \in \Theta, \ \forall \  \delta \in(0,1],
\end{equation}
where $M\big(T, K_{1}, \delta\big)=2N\big(B_{\Theta} U_{R}\big)^{2}\left[\left(\frac{T+1}{1-\gamma}+
\frac{\gamma}{(1-\gamma)^{2}}\right) \gamma^{T+1}\right]^{2}+16\log\big(\frac{8K_{1}}{\delta}\big)
\frac{NB_{\Theta}^{2}U_{R}^{2}\gamma^{2}}{(1-\gamma)^{4}K_{1}}
$.
\end{lemma}

\begin{proof}
By $C_{2}$ inequality, we have
\begin{equation}\label{estimateF}
    \begin{aligned}
    \left\|\hat{F}(\theta)-F(\theta)\right\|^{2}
    &=\left\|\hat{F}(\theta)-F(\theta,T)+F(\theta,T)-F(\theta)\right\|^{2}\\
    & \leqslant 2\left\|\hat{F}(\theta)-F(\theta,T)\right\|^{2}+
    2\left\|F(\theta,T)-F(\theta)\right\|^{2}.\\
\end{aligned}
\end{equation}
For the second term in the above inequality, from Assumptions \ref{assumption0}-\ref{assumption2} and Lemma 6 in \cite{Tianyi}, it follows that
\begin{equation}\label{firsttermestimate}
\begin{aligned}
2\left\|F(\theta,T)-F(\theta)\right\|^{2}
= &2\sum_{i=1}^{N} \left\|-\nabla_{\theta_{i}}^{T} J_{i}^{\theta}(t)+\nabla_{\theta_{i}} J_{i}^{\theta}(t)\right\|^{2}
\leqslant  \sigma_{T} ,
\end{aligned}
\end{equation}
where $\sigma_{T}=2N\big(B_{\Theta} U_{R}\big)^{2}\left[\left(\frac{T+1}{1-\gamma}+\frac{\gamma}{(1-\gamma)^{2}}\right) \gamma^{T+1}\right]^{2}$.
For the first term in (\ref{estimateF}), denote
\begin{equation}\nonumber
\begin{aligned}
\hat{F}(\theta,T,w)= & (\hat{F}_{i}(\theta,T,w))_{i=1}^{N}\\
=& \left(\sum_{l=t}^{T+t}\left(\sum_{\tau=t}^{l} \nabla_{\theta_{i}} \log \pi_{\theta_{i}}(a_{i}(\tau,w)\mid s(\tau,w))\big)  \gamma^{l-t} r_{i}\big(s(l,w), a(l,w)\right)\right)_{i=1}^{N},
\end{aligned}
\end{equation}
and then $\hat{F}(\theta)=\frac{\sum_{w=1}^{K_{1}}\hat{F}(\theta,T,w)}{K_{1}}$. Noting that  $\hat{F}_{i}(\theta,T,w)$ is an unbiased estimator of $F_{i}(\theta,T)$ ~\citep{YLi},  we know that
 $\hat{F}(\theta,T,w)$ is the unbiased estimator of $F(\theta,T)$. By Assumptions
 \ref{assumption0}-\ref{assumption2} and Lemma 5 in \cite{Tianyi}, we have $\left\|\hat{F}_{i}(\theta,T,w)-F_{i}(\theta,T)\right\|^{2} \leqslant \left(\frac{2B_{\Theta}U_{R}\gamma}{(1-\gamma)^{2}}\right)^{2}$, and then
$$\left\|\frac{\hat{F}(\theta,T,w)-F(\theta,T)}{K_{1}}
\right\|^{2} = \frac{ \sum_{i=1}^{N}\left\|\hat{F}_{i}(\theta,T,w)-F_{i}(\theta,T)\right\|^{2}}{K_{1}^{2}} \leqslant  N\left(\frac{2B_{\Theta}U_{R}\gamma}{(1-\gamma)^{2}}
\right)^{2}\frac{1}{K_{1}^{2}}.
$$
Hence, by Lemma \ref{concentration inequality}, it follows that
\begin{equation}\nonumber
\begin{aligned}
P\bigg\{\left\|\hat{F}(\theta)-F(\theta,T)\right\|^{2}
\leqslant 8\log\left(\frac{8K_{1}}{\delta}\right)
\frac{NB_{\Theta}^{2}U_{R}^{2}\gamma^{2}}{(1-\gamma)^{4}K_{1}}\bigg\} \geqslant 1-\frac{\delta}{4K_{1}},
\ \forall \ \delta \in (0,1].
\end{aligned}
\end{equation}
Then combining (\ref{estimateF}) and (\ref{firsttermestimate}) with the above inequality gives
\begin{equation}\nonumber
\begin{aligned} P\bigg\{\left\|\hat{F}(\theta)-F(\theta)\right\|^{2} \leqslant M(T,K_{1},\delta) \bigg\}\geqslant 1-\frac{\delta}{4K_{1}},
\ \forall \ \delta \in (0,1].
\end{aligned}
\end{equation}
\end{proof}

\subsection{Convergence Analysis of Algorithm \ref{alg::pseudoGradient}}
In this section, we establish the convergence of Algorithm \ref{alg::pseudoGradient}.

\begin{lemma}\label{lemmapseudo}
If Assumptions \ref{assumption0}-\ref{assumption4} hold, and we choose $\eta_{h}=\frac{l_{1}}{h^{\frac{2}{3}}}$, $\gamma_{k}=k^{\frac{1}{4}}$ and $H_{k}=k$  in Algorithm \ref{alg::pseudoGradient}, then $\big((\pi_{\theta_{i,k}})_{i=1}^{N}\big)
_{k=1}^{\infty}$ given by Algorithm \ref{alg::pseudoGradient} satisfies
\begin{equation}\nonumber
\begin{aligned}
P\left\{\sup _{i \in \mathcal{N}}\frac{\sum_{k=1}^{K} k^{\frac{1}{4}}\left(\sup \limits_{\theta_{i} \in \widetilde{\Theta_{i}}}J_{i}^{(\theta_{i}, \theta_{-i,k})}(t)-J_{i}^{\theta_{k}}(t)\right)}
{\sum_{k=1}^{K}k^{\frac{1}{4}}}\leqslant L(K,K_{1},T,\delta)
+(\phi(K))^{\frac{1}{4}}\right\} \geqslant 1-\frac{K}{K_{1}}\delta,\\
 \forall \ \delta \in (0,1],\  for \ any \ given \ \text{countable set} \ \widetilde{\Theta_{i}} \subseteq \Theta_{i},
\end{aligned}
\end{equation}
where $L(K,K_{1},T,\delta)=\frac{\sqrt{10}D}{2} \left(4LD+\frac{2D}{\beta}+\frac{\sqrt{N}B_{\Theta}U_{R}}{(1-\gamma)^{2}}\right)
\frac{M_{1}}{K^{\frac{1}{2}}}+\sqrt{2}M_{1}\left(\frac{1}{\frac{1}{\beta}-L}
+2\beta\right)^{\frac{1}{2}}\Big(4LD+\frac{2D}{\beta}+\frac{\sqrt{N}B_{\Theta}U_{R}}{(1-\gamma)^{2}}
\Big)
\Bigg[\frac{\sqrt{M}}{2\left(L^{2}+\frac{1}{\beta^{2}}\right)^{\frac{1}{2}}}\left(4D
\left(2L^{2}+\frac{2}{\beta^{2}}\right)^{\frac{1}{2}}+\left(\frac{2NB_{\Theta}^{2}U_{R}^{2}}{(1-\gamma)^{2}}
+\frac{8D^{2}}{\gamma^{2}}\right)^{\frac{1}{2}}\right)
+4D(1+\sqrt{2})\bigg(\Big(2L^{2}+\frac{2}{\beta^{2}}\Big) \frac{D\sqrt{M}}{\frac{1}{\beta}-L}\Bigg)^{\frac{1}{2}}\Bigg]^{\frac{1}{2}} +
\sqrt{\frac{30}{11}} \left(\frac{1}{\frac{1}{\beta}-L}+2\beta
\right)^{\frac{1}{2}} \left(4LD+\frac{2D}{\beta}+\frac{\sqrt{N}B_{\Theta}U_{R}}{(1-\gamma)^{2}}
\right)
\Bigg(2D \left(2L^{2}+\frac{2}{\beta^{2}}\right)^{\frac{1}{2}}
(2+\sqrt{2})\\
 \frac{l_{1}^{\frac{1}{2}}\big(9M+\frac{D^{2}}{2}\big)^{\frac{1}{2}}}
 {\big(\frac{1}{\beta}-L\big)^{\frac{1}{2}}}\Bigg)^{\frac{1}{2}}
\frac{M_{1}}{K^{\frac{1}{6}}}$, $ (\phi(K))^{\frac{1}{4}}=o\left(\frac{1}{K^{\frac{1}{6}}}\right)$, $M_{1}$ is given by Lemma \ref{gradientd}, $M=M\big(T, K_{1}, \delta\big)=2N\left(B_{\Theta} U_{R}\right)^{2}\left[\left(\frac{T+1}{1-\gamma}+
\frac{\gamma}{(1-\gamma)^{2}}\right) \gamma^{T+1}\right]^{2}+16\log\big(\frac{8K_{1}}{\delta}\big)
\frac{NB_{\Theta}^{2}U_{R}^{2}\gamma^{2}}{(1-\gamma)^{4}K_{1}}
$ and $D=\sqrt{\sum_{i=1}^{N} D_{i}^{2}}$.
\end{lemma}

\begin{proof}
By  Lemma  \ref{propertyofprox} (ii), we have
\begin{equation}\nonumber
\begin{aligned}
\frac{\left\|z_{i}^{h+1}-\theta_{i}\right\|^{2}}{2} \leqslant &
\frac{\left\|z_{i}^{h}-\theta_{i}\right\|^{2}}{2}-\left\langle \eta_{h} \hat{F}_{i,k}(\theta^{h+1}), \theta_{i}^{h+1}-\theta_{i}\right\rangle+\frac{\left\|\eta_{h} \hat{F}_{i,k}(\theta^{h+1})-\eta_{h} \hat{F}_{i,k}(\theta^{h})\right\|^{2}}{2}
\\&-
\frac{\left\| \theta_{i}^{h+1}-z_{i}^{h}\right\|^{2}}{2}, \forall \ i \in \mathcal{N}.
\end{aligned}
\end{equation}
Taking summation for both sides of the above inequality from $i=1$ to $N$ and rearranging the above inequality lead to
\begin{align}\label{threeterm}
\left\langle \eta_{h} F_{k}(\theta^{h+1}), \theta^{h+1}-\theta\right\rangle \leqslant & \frac{\left\|z^{h}-\theta\right\|^{2}}{2}-\frac{\left\|z^{h+1}-\theta \right\|^{2}}{2}+\frac{\left\|\eta_{h} \hat{F}_{k}(\theta^{h+1})-\eta_{h} \hat{F}_{k}(\theta^{h})\right\|^{2}}{2} \notag\\
&+\left\langle -\eta_{h} \hat{F}_{k}(\theta^{h+1})+\eta_{h}F_{k}(\theta^{h+1}), \theta^{h+1}-\theta\right\rangle
-\frac{\left\| \theta^{h+1}-z^{h}\right\|^{2}}{2}.
\end{align}
For the fourth term on the right side of (\ref{threeterm}), by Assumption \ref{assumption1}, we have
\begin{equation}\label{secondterm}
\begin{aligned}
\left\langle -\eta_{h} \hat{F}_{k}(\theta^{h+1})+\eta_{h}F_{k}(\theta^{h+1}), \theta^{h+1}-\theta\right\rangle
&\leqslant \eta_{h}\left\|-\hat{F}_{k}(\theta^{h+1})+F_{k}(\theta^{h+1})\right\| \left\|\theta^{h+1}-\theta\right\|\\
&\leqslant  2\eta_{h}D\left\|-\hat{F}_{k}(\theta^{h+1})+F_{k}(\theta^{h+1})\right\|.
\end{aligned}
\end{equation}
For the third term on the right side of (\ref{threeterm}),  by Cauchy-Schwartz  inequality, Assumptions
\ref{assumption0}-\ref{assumption2} and Lemma \ref{stronglylipschitz}, we have
\begin{align}
&\frac{\left\|\eta_{h} \hat{F}_{k}(\theta^{h+1})-\eta_{h} \hat{F}_{k}(\theta^{h})\right\|^{2}}{2} \notag\\
=& \frac{\eta_{h}^{2}\left\| \hat{F}_{k}(\theta^{h+1})- F_{k}(\theta^{h+1})+ F_{k}(\theta^{h+1})-F_{k}(\theta^{h})+F_{k}(\theta^{h})- \hat{F}_{k}(\theta^{h})\right\|^{2}}{2}\notag\\
\leqslant & \frac{3\eta_{h}^{2}}{2}\Big(\left\| \hat{F}_{k}(\theta^{h+1})- F_{k}(\theta^{h+1})\right\|^{2}+\left\| F_{k}(\theta^{h+1})- F_{k}(\theta^{h})\right\|^{2}+\left\| F_{k}(\theta^{h})- \hat{F}_{k}(\theta^{h})\right\|^{2}\Big)\notag\\
\leqslant & \frac{3\eta_{h}^{2}}{2}\left(\left\| \hat{F}_{k}(\theta^{h+1})- F_{k}(\theta^{h+1})\right\|^{2}
+\left\| F_{k}(\theta^{h})- \hat{F}_{k}(\theta^{h})\right\|^{2}\right) +\frac{3}{2}\eta_{h}^{2}\left(2L^{2}+\frac{2}{\beta^{2}}\right) \left\|\theta^{h+1}-\theta^{h}\right\|^2.\label{thirdterm}
\end{align}
Applying $C_{2}$ inequality into  the term $\left\|\theta^{h+1}-\theta^{h}\right\|^2 $ in the above inequality leads to
\begin{equation}\label{single}
    \begin{aligned}
    \left\|\theta^{h+1}-\theta^{h}\right\|^2 \leqslant 2\left\|\theta^{h+1}-z^{h}\right\|^2+2\left\|\theta^{h}-z^{h}\right\|^2.
    \end{aligned}
\end{equation}
For the term $\left\|\theta^{h}-z^{h}\right\|^2$ in the above inequality, by the non-expansion property of the  proximal mapping in Lemma  \ref{propertyofprox} (iii), Cauchy-Schwartz inequality, Assumptions \ref{assumption0}-\ref{assumption2} and Lemma \ref{stronglylipschitz}, we have
    \begin{align*}
    &\left\|\theta^{h}-z^{h}\right\|^2\\
    =&\sum_{i=1}^{N}\left\|\theta_{i}^{h}-z_{i}^{h}\right\|^2\\
   \leqslant &  \sum_{i=1}^{N}\left\|-\eta_{h-1}\hat{F}_{ik}(\theta^{h-1})+\eta_{h-1} \hat{F}_{ik}(\theta^{h}) \right\|^2\\
    = & \sum_{i=1}^{N}\Bigg\|-\eta_{h-1}\hat{F}_{ik}(\theta^{h-1})+\eta_{h-1}F_{ik}(\theta^{h-1})-\eta_{h-1} F_{ik}(\theta^{h-1})+\eta_{h-1} F_{ik}(\theta^{h}) \\
    &-\eta_{h-1} F_{ik}(\theta^{h}) +\eta_{h-1} \hat{F}_{ik}(\theta^{h}) \Bigg\|^2\\
    \leqslant &  3\eta_{h-1}^2 \left\{  \sum_{i=1}^{N}\left\|-\hat{F}_{ik}(\theta^{h-1})+F_{ik}(\theta^{h-1}) \right\|^2+\sum_{i=1}^{N}\left\|- F_{ik}(\theta^{h-1})+ F_{ik}(\theta^{h}) \right\|^2 \right.\\
    & \left.+\sum_{i=1}^{N}\left\|-  F_{ik}(\theta^{h}) + \hat{F}_{ik}(\theta^{h}) \right\|^2 \right\}
    \\
    =& 3\eta_{h-1}^2 \left\{ \left\|-\hat{F_{k}}(\theta^{h-1})+F_{k}(\theta^{h-1}) \right\|^2+\left\|F_{k}(\theta^{h-1})-F_{k}(\theta^{h}) \right\|^2+\left\|F_{k}(\theta^{h})-\hat{F}_{k}(\theta^{h}) \right\|^2\right\}\\
    \leqslant & 3\eta_{h-1}^2 \left\{ \left\|-\hat{F_{k}}(\theta^{h-1})+F_{k}(\theta^{h-1}) \right\|^2+\left\|F_{k}(\theta^{h})-\hat{F}_{k}(\theta^{h}) \right\|^2\right\}\\
    &+3\eta_{h-1}^2 \left(2L^{2}+\frac{2}{\beta^{2}}\right)\left\|\theta^{h-1}-\theta^{h}\right\|^2 \label{thirdtemson}.
    \end{align*}
Substituting the above inequality into (\ref{single}) gives
\begin{equation}\nonumber
    \begin{aligned}
    \left\|\theta^{h+1}-\theta^{h}\right\|^2 \leqslant&
    6\eta_{h-1}^2 \left\{ \left\|-\hat{F_{k}}(\theta^{h-1})+F_{k}(\theta^{h-1}) \right\|^2+\left\|F_{k}(\theta^{h})-\hat{F}_{k}(\theta^{h}) \right\|^2\right\}\\
    &+6\eta_{h-1}^2 \left(2L^{2}+\frac{2}{\beta^{2}}\right)
    \left\|\theta^{h-1}-\theta^{h}\right\|^2+2\left\|\theta^{h+1}-z^{h}\right\|^2.
    \end{aligned}
\end{equation}
By (\ref{threeterm}), (\ref{secondterm}),
(\ref{thirdterm}) and the above inequality, we have
\begin{align*}
\left\langle \eta_{h} F_{k}(\theta^{h+1}), \theta^{h+1}-\theta\right\rangle \leqslant & \frac{\left\|z^{h}-\theta\right\|^{2}}{2}-\frac{\left\|z^{h+1}-\theta \right\|^{2}}{2}+2\eta_{h}D\left\|-\hat{F}_{k}(\theta^{h+1})+
F_{k}(\theta^{h+1})\right\| \\
&+ \frac{3\eta_{h}^{2}}{2}\left(\left\| \hat{F}_{k}(\theta^{h+1})- F_{k}(\theta^{h+1})\right\|^{2}
+\left\| F_{k}(\theta^{h})- \hat{F}_{k}(\theta^{h})\right\|^{2}\right) \\
&+9\left(2L^{2}+\frac{2}{\beta^{2}}\right) \eta_{h}^{2}\eta_{h-1}^2 \Bigg\{ \left\|-\hat{F_{k}}(\theta^{h-1})+F_{k}(\theta^{h-1}) \right\|^2\\
&+\left\|F_{k}(\theta^{h})-\hat{F}_{k}(\theta^{h}) \right\|^2 \Bigg\}
+9\left(2L^{2}+\frac{2}{\beta^{2}}\right)^{2} \eta_{h}^{2}\eta_{h-1}^2\left\|\theta^{h-1}-\theta^{h}\right\|^2\\
&+\left(3\eta_{h}^{2}\left(2L^{2}+\frac{2}{\beta^{2}}\right)-\frac{1}{2}\right) \left\|\theta^{h+1}-z^{h}\right\|^2.
\end{align*}
From $l_{2} \geqslant 3\left(2L^{2}+\frac{2}{\beta^{2}}\right)$, $l_{1}= \min\left\{\frac{1}{2\sqrt{l_{2}}}, \frac{1}{4l_{2}}  \right\}$ and $\eta_{h}=\frac{l_{1}}{h^{\frac{2}{3}}} $, we know that $9\Big(2L^{2}+\frac{2}{\beta^{2}}\Big) \eta_{h-1}^2 \leqslant \frac{3}{4}$, $9\Big(2L^{2}+\frac{2}{\beta^{2}}\Big)^{2}$
$\eta_{h-1}^2 \leqslant \frac{1}{16}$. This together with the above inequality gives
\begin{align*}
\left\langle \eta_{h} F_{k}(\theta^{h+1}), \theta^{h+1}-\theta\right\rangle \leqslant & \frac{\left\|z^{h}-\theta\right\|^{2}}{2}-\frac{\left\|z^{h+1}-\theta \right\|^{2}}{2}+2\eta_{h}D\left\|-\hat{F}_{k}(\theta^{h+1})+
F_{k}(\theta^{h+1})\right\| \\
&+ \frac{3\eta_{h}^{2}}{2}\left(\left\| \hat{F}_{k}(\theta^{h+1})- F_{k}(\theta^{h+1})\right\|^{2}
+\left\| F_{k}(\theta^{h})- \hat{F}_{k}(\theta^{h})\right\|^{2}\right) \\
&+\frac{3}{4}\eta_{h}^{2} \left( \left\|-\hat{F_{k}}(\theta^{h-1})+F_{k}(\theta^{h-1}) \right\|^2+\left\|F_{k}(\theta^{h})-\hat{F}_{k}(\theta^{h}) \right\|^2\right)\\
&+\frac{1}{16} \eta_{h}^{2}\left\|\theta^{h-1}-\theta^{h}\right\|^2+
\left(3\eta_{h}^{2}\big(2L^{2}+\frac{2}{\beta^{2}}\big)-\frac{1}{2}\right) \left\|\theta^{h+1}-z^{h}\right\|^2.
\end{align*}
According to Assumptions \ref{assumption0}-\ref{assumption2} and Lemma
\ref{stronglymonotone}, $F_{k}(\theta)$ is $\left(\frac{1}{\beta}-L\right)$-strongly monotone and then SVI($F_{k},\Theta)$ has a unique solution ~\citep{KINDERLEHRER1980}, which is denoted by $\overline{\theta_{k}}$. Substituting $\overline{\theta_{k}}$ for $\theta$ in the above inequality implies
\begin{align}
\left\langle \eta_{h} F_{k}(\theta^{h+1}), \theta^{h+1}-\overline{\theta_{k}}\right\rangle \leqslant & \frac{\left\|z^{h}-\overline{\theta_{k}}\right\|^{2}}{2}
-\frac{\left\|z^{h+1}-\overline{\theta_{k}} \right\|^{2}}{2}+2\eta_{h}D\left\|-\hat{F}_{k}(\theta^{h+1})+
F_{k}(\theta^{h+1})\right\| \notag\\
&+ \frac{3\eta_{h}^{2}}{2}\left(\left\| \hat{F}_{k}(\theta^{h+1})- F_{k}(\theta^{h+1})\right\|^{2}
+\left\| F_{k}(\theta^{h})- \hat{F}_{k}(\theta^{h})\right\|^{2}\right) \notag\\
&+ \frac{3}{4}\eta_{h}^{2} \left( \left\|-\hat{F_{k}}(\theta^{h-1})+F_{k}(\theta^{h-1}) \right\|^2+\left\|F_{k}(\theta^{h})-\hat{F}_{k}(\theta^{h}) \right\|^2\right)\notag\\
&+\frac{1}{16} \eta_{h}^{2}\left\|\theta^{h-1}-\theta^{h}\right\|^2+
\left(3\eta_{h}^{2}\left(2L^{2}+\frac{2}{\beta^{2}}\right)-\frac{1}{2}\right) \left\|\theta^{h+1}-z^{h}\right\|^2\label{part53}.
\end{align}
Similar to the proof of (\ref{prox224}) in Lemma \ref{exactlemma}, we have
\begin{equation}\nonumber
\begin{aligned}
\eta_{h}\left(\frac{1}{\beta}-L\right)\frac{\left\|z^{h}-\overline{\theta_{k}}\right\|^{2}}{2}
-\eta_{h}\left(\frac{1}{\beta}-L\right)\left\|z^{h}-\theta^{h+1}\right\|^{2} \leqslant \left\langle \eta_{h} F_{k}(\theta^{h+1}),
\theta^{h+1}-\overline{\theta_{k}}\right\rangle.
\end{aligned}
\end{equation}
Taking summation for both sides of (\ref{part53}) and the above inequality gives
\begin{align*}
&\frac{\left\|z^{h+1}-\overline{\theta_{k}} \right\|^{2}}{2}\\
 \leqslant & \left(1-\eta_{h}\left(\frac{1}{\beta}-L\right)\right)
\frac{\left\|z^{h}-\overline{\theta_{k}}\right\|^{2}}{2}
+2\eta_{h}D\left\|-\hat{F}_{k}(\theta^{h+1})+
F_{k}(\theta^{h+1})\right\| \\
&+ \frac{3\eta_{h}^{2}}{2}\left(\left\| \hat{F}_{k}(\theta^{h+1})- F_{k}(\theta^{h+1})\right\|^{2}
+\left\| F_{k}(\theta^{h})- \hat{F}_{k}(\theta^{h})\right\|^{2}\right) \\
&+ \frac{3}{4}\eta_{h}^{2} \left( \left\|-\hat{F_{k}}(\theta^{h-1})+F_{k}(\theta^{h-1}) \right\|^2+\left\|F_{k}(\theta^{h})-\hat{F}_{k}(\theta^{h}) \right\|^2\right)\\
&+\frac{1}{16} \eta_{h}^{2}\left\|\theta^{h-1}-\theta^{h}\right\|^2+
\Bigg[3\eta_{h}^{2}\left(2L^{2}+\frac{2}{\beta^{2}}\right)
+\eta_{h}\left(\frac{1}{\beta}-L\right)-\frac{1}{2}\Bigg] \left\|\theta^{h+1}-z^{h}\right\|^2.
\end{align*}
From $l_{2} \geqslant \max \left\{3\left(2L^{2}+\frac{2}{\beta^{2}}\right), \left(\frac{1}{\beta}-L\right)\right\}$, $l_{1}= \min\left\{\frac{1}{2\sqrt{l_{2}}}, \frac{1}{4l_{2}}  \right\}$ and $\eta_{h}=\frac{l_{1}}{h^{\frac{2}{3}}} $,
we know that  $3\eta_{h}^{2}\left(2L^{2}+\frac{2}{\beta^{2}}\right)
+\eta_{h}\left(\frac{1}{\beta}-L\right)
\leqslant l_{2}\frac{l_{1}^{2}}{h^{\frac{4}{3}}}+
l_{2}\frac{l_{1}}{h^{\frac{2}{3}}} \leqslant l_{2}l_{1}^{2}+l_{2}l_{1} \leqslant l_{2}\frac{1}{4l_{2}}+l_{2}\frac{1}{4l_{2}} \leqslant \frac{1}{2}$, that is, $3\eta_{h}^{2}\Big(2L^{2}+\frac{2}{\beta^{2}}\Big)
+\eta_{h}\left(\frac{1}{\beta}-L\right)
-\frac{1}{2} \leqslant 0$. This together with the above inequality leads to
\begin{align}
\frac{\left\|z^{h+1}-\overline{\theta_{k}} \right\|^{2}}{2} \leqslant & \left(1-\eta_{h}\left(\frac{1}{\beta}-L\right)\right)
\frac{\left\|z^{h}-\overline{\theta_{k}}\right\|^{2}}{2}
+2\eta_{h}D\left\|-\hat{F}_{k}(\theta^{h+1})+
F_{k}(\theta^{h+1})\right\| \notag\\
&+ \frac{3\eta_{h}^{2}}{2}\left(\left\| \hat{F}_{k}(\theta^{h+1})- F_{k}(\theta^{h+1})\right\|^{2}
+\left\| F_{k}(\theta^{h})- \hat{F}_{k}(\theta^{h})\right\|^{2}\right) \notag\\
&+ \frac{3}{4}\eta_{h}^{2} \left( \left\|-\hat{F_{k}}(\theta^{h-1})+F_{k}(\theta^{h-1}) \right\|^2+\left\|F_{k}(\theta^{h})-\hat{F}_{k}(\theta^{h}) \right\|^2\right)\notag\\
&+\frac{1}{16} \eta_{h}^{2}\left\|\theta^{h-1}-\theta^{h}\right\|^2.\label{betweentheta}
\end{align}
Denote $A_{k}(\theta)=\bigg\{w \in \Omega:\|\hat{F_{k}}\big(\theta\big)-F_{k}(\theta)\|^{2} \leqslant M\big(T, K_{1},\delta\big)\bigg\}$ and
$A_{k}=A_{k}(\theta^{H_{k}-1})${\Large $\cap$} $A_{k}(\theta^{H_{k}})${\Large $\cap$}
$A_{k}(\theta^{H_{k}+1})${\Large $\cap$}$A_{k}(z^{H_{k}+1})$.
By Assumptions \ref{assumption0}-\ref{assumption2} and Lemma \ref{gradientestimate}, we know that $P\big\{A_{k}(\theta)\big\}$ $ \geqslant 1-\frac{\delta}{4K_{1}}$, for any $ \theta \in \Theta$.  Hence, we have $P\big\{A_{k}\big\} \geqslant 1-\frac{\delta}{K_{1}}$. Then, by (\ref{betweentheta}), we obtain
\begin{equation}\nonumber
\begin{aligned}
\frac{\left\|z^{H_{k}+1}-\overline{\theta_{k}} \right\|^{2}}{2} \leqslant \left(1-\eta_{H_{k}}\left(\frac{1}{\beta}-L\right)\right)
\frac{\left\|z^{H_{k}}-\overline{\theta_{k}}\right\|^{2}}{2}
+2\eta_{H_{k}}D\sqrt{M}+\frac{9\eta_{H_{k}}^{2}}{2}M+\frac{D^{2}}{4} \eta_{H_{k}}^{2}, \\ \ \forall \ w \in A_{k}.
\end{aligned}
\end{equation}
Combining $\eta_{H_{k}}\left(\frac{1}{\beta}-L\right) \geqslant 0$ and  the above inequality with  Lemma
\ref{lastlemma} gives
\begin{equation}\label{part591}
\begin{aligned}
\left\|z^{H_{k}+1}-\overline{\theta_{k}} \right\|\leqslant \left(\frac{4D\sqrt{M}}{\frac{1}{\beta}-L}\right)^{\frac{1}{2}}+
 \left(\frac{l_{1}\left(9M+\frac{D^{2}}{2}\right)}
 {\left(\frac{1}{\beta}-L\right)}\right)^{\frac{1}{2}} \frac{1}{H_{k}^{\frac{1}{3}}}+(\phi(H_{k}))^{\frac{1}{2}}, \forall \ w \in A_{k},
\end{aligned}
\end{equation}
where $\phi(H_{k})=o\left(\frac{1}{H_{k}^{\frac{2}{3}}}\right)$.
Denote $\widetilde{\theta}_{k+1}=\mathop{\arg\min}\limits_{\theta \in \Theta} \left\{\left\langle 2\widetilde{\eta} F_{k}(z^{H_{k}+1}), \theta \right\rangle+\|\theta-z^{H_{k}+1}\|^{2}\right\}$, where $\widetilde{\eta}=\frac{1}{2\sqrt{L^{2}+
\frac{1}{\beta^{2}}}}$.
According to the non-expansion property of the proximal mapping in Lemma  \ref{propertyofprox} (iii), we have
\begin{equation}\label{wildetheta}
\left\|\theta_{k+1}-\widetilde{\theta}_{k+1} \right\|^{2}\leqslant \widetilde{\eta}^{2} \left\|F_{k}(z^{H_{k}+1})-\hat{F}_{k}(z^{H_{k}+1}) \right\|^{2}.
\end{equation}
By Assumptions \ref{assumption0}-\ref{assumption2}, Lemma
\ref{stronglymonotone} and Lemma  \ref{contorolstrong}, we have
\begin{align}
&\langle F_{k}(\theta_{k+1}), \theta_{k+1}-\theta\rangle \notag\\
=&\langle F_{k}(\theta_{k+1})-F_{k}(\widetilde{\theta}_{k+1})+F_{k}
(\widetilde{\theta}_{k+1}), \theta_{k+1}-\widetilde{\theta}_{k+1}
+\widetilde{\theta}_{k+1}-\theta\rangle\notag\\
=&\langle F_{k}(\theta_{k+1})-F_{k}(\widetilde{\theta}_{k+1}), \theta_{k+1}-\widetilde{\theta}_{k+1}\rangle+\langle F_{k}(\widetilde{\theta}_{k+1}), \theta_{k+1}-\widetilde{\theta}_{k+1}\rangle\notag\\
&+\langle F_{k}(\theta_{k+1})-F_{k}(\widetilde{\theta}_{k+1}),
\widetilde{\theta}_{k+1}-\theta\rangle+\langle F_{k}(\widetilde{\theta}_{k+1}),\widetilde{\theta}_{k+1}
-\theta\rangle\notag\\
\leqslant & 2\left(2L^{2}+\frac{2}{\beta^{2}}\right)^{\frac{1}{2}}D\left\|\theta_{k+1}-\widetilde{\theta}_{k+1}\right\|+
2\left(2L^{2}+\frac{2}{\beta^{2}}\right)^{\frac{1}{2}}D\left\|\theta_{k+1}-\widetilde{\theta}_{k+1}\right\|
\notag\\
&+\left\|F_{k}(\widetilde{\theta}_{k+1})\right\|\left\|\theta_{k+1}-\widetilde{\theta}_{k+1}
\right\|+2D \left(2L^{2}+\frac{2}{\beta^{2}}\right)^{\frac{1}{2}}(2+\sqrt{2})
\left\|z^{H_{k}+1}-\overline{\theta_{k}} \right\|.\label{part60}
\end{align}
Similar to the proof of (\ref{boundedoff}) in Lemma \ref{exactlemma}, by  Assumptions \ref{assumption0}-\ref{assumption2}, we know that $\left\|F(\widetilde{\theta}_{k+1})\right\| \\ \leqslant
\frac{\sqrt{N}B_{\Theta}U_{R}}{(1-\gamma)^{2}}$, then
$\left\|F_{k}(\widetilde{\theta}_{k+1})\right\|
\leqslant \bigg(2\left\|F(\widetilde{\theta}_{k+1})\right\|^{2}
+
\frac{2}{\beta^{2}}
\left\|\widetilde{\theta}_{k+1}-\theta_{k}\right\|^{2}\bigg)
^{\frac{1}{2}}
\leqslant \left(\frac{2NB_{\Theta}^{2}U_{R}^{2}}{(1-\gamma)^{4}}+\frac{8D^{2}}{\beta^{2}}\right)^{\frac{1}{2}}$. This together with (\ref{part591}), (\ref{wildetheta}) and   (\ref{part60}) gives
\begin{align}
\langle F_{k}(\theta_{k+1}), \theta_{k+1}-\theta\rangle
\leqslant & 4\left(2L^{2}+\frac{2}{\beta^{2}}\right)^{\frac{1}{2}}D\left\|\theta_{k+1}-\widetilde{\theta}_{k+1}\right\|
+\left(\frac{2NB_{\Theta}^{2}U_{R}^{2}}{(1-\gamma)^{4}}
+\frac{8D^{2}}{\beta^{2}}\right)^{\frac{1}{2}}\notag\\
&\times \left\|\theta_{k+1}-\widetilde{\theta}_{k+1}
\right\|+2D \left(2L^{2}+\frac{2}{\beta^{2}}\right)^{\frac{1}{2}}(2+\sqrt{2})
\left\|z^{H_{k}+1}-\overline{\theta_{k}} \right\|\notag\\
\leqslant &C_{3}\left\|F_{k}(z^{H_{k}+1})-\hat{F}_{k}(z^{H_{k}+1}) \right\|
+C_{4}
 +C_{5}\frac{1}{H_{k}^{\frac{1}{3}}}+C_{6}(\phi(H_{k}))^{\frac{1}{2}},\label{tlastitrea}
\end{align}
where $C_{3}=\widetilde{\eta}\left(4\left(2L^{2}+\frac{2}{\beta^{2}}\right)^{\frac{1}{2}}D+\left(\frac{2NB_{\Theta}^{2}U_{R}^{2}}{(1-\gamma)^{2}}
+\frac{8D^{2}}{\gamma^{2}}\right)^{\frac{1}{2}}\right)$, $C_{4}=2D \left(2L^{2}+\frac{2}{\beta^{2}}\right)^{\frac{1}{2}}(2+\sqrt{2})\\
\left(\frac{4D\sqrt{M}}{\frac{1}{\beta}-L}\right)^{\frac{1}{2}}$, $C_{5}=2D \sqrt{2L^{2}+\frac{2}{\beta^{2}}}
(2+\sqrt{2})
 \left(\frac{l_{1}\big(9M+\frac{D^{2}}{2}\big)}
 {\big(\frac{1}{\beta}-L\big)}\right)^{\frac{1}{2}}$ and $C_{6}=2D \left(2L^{2}+\frac{2}{\beta^{2}}\right)^{\frac{1}{2}}(2+\sqrt{2})$.
From $A_{k} \subseteq A_{k}(z^{H_{k}+1})$ and (\ref{tlastitrea}), it follows that
\begin{align*}
\left\langle F_{k}(\theta_{k+1}), \theta_{k+1}-\theta \right\rangle \leqslant C_{3}\sqrt{M}
+C_{4}+C_{5}\frac{1}{H_{k}^{\frac{1}{3}}}
+C_{6}(\phi(H_{k}))^{\frac{1}{2}}, \forall \ w \in A_{k}.
\end{align*}
Recall that $P\big\{A_{k}\big\} \geqslant 1-\frac{\delta}{K_{1}}$. This together with the above inequality gives
\begin{align}\nonumber
P\left\{\left\langle F_{k}(\theta_{k+1}), \theta_{k+1}-\theta \right\rangle \leqslant C_{3}\sqrt{M}
+C_{4}+C_{5}\frac{1}{H_{k}^{\frac{1}{3}}}
+C_{6}(\phi(H_{k}))^{\frac{1}{2}} \right\} \geqslant P\left\{A_{k}\right\} \geqslant 1-\frac{\delta}{K_{1}}.
\end{align}
Noting that  $P\big\{A_{k}\big\} \geqslant 1-\frac{\delta}{K_{1}}$, we have
\begin{align}\label{intersectionAk}
P\left\{\bigcap_{k=1}^{K}A_{k} \right\} =1-P\left\{\bigcup_{k=1}^{K} A_{k}^{c} \right\} \geqslant 1-\frac{K}{K_{1}}\delta.
\end{align}
Noting that $\gamma_{k}=k^{\frac{1}{4}}$ and $H_{k}=k$, we denote $\widetilde{A}=\bigg\{w \in \Omega:\frac{\sum_{k=1}^{K} k^{\frac{1}{4}}\left\|\theta_{k}-\overline{\theta_{k}}
\right\|^{2}}{\sum_{k=1}^{K}k^{\frac{1}{4}}} \leqslant 2\frac{K^{\frac{1}{4}}D^{2}}{\sum_{k=1}^{K} k^{\frac{1}{4}}}+\frac{2\left(\frac{1}{\frac{1}
 {\gamma}-L}+2\beta\right)\sum_{k=1}^{K}k^{\frac{1}{4}} P_{k}}{\sum_{k=1}^{K} k^{\frac{1}{4}}}\bigg\}$, where $P_{k}=C_{3}\sqrt{M}
+C_{4}+C_{5}\frac{1}{H_{k}^{\frac{1}{3}}}
+C_{6}(\phi(k))^{\frac{1}{2}}$.
By Assumption \ref{assumption4}, similar to the proof of (\ref{thetawidetheta})-(\ref{nearfinal1}) in Lemma \ref{exactlemma}, for any $w \in$
{\Large $\cap$}$_{k=1}^{K}$$A_{k}$, we obtain that  $w \in \widetilde{A}$. Hence, {\Large $\cap$}$_{k=1}^{K}$$A_{k} \subseteq$$\widetilde{A}$. This together with (\ref{intersectionAk}) gives
\begin{align}\label{intersectionwidetildeAk}
P\left\{\widetilde{A} \right\} \geqslant P\left\{\bigcap_{k=1}^{K}A_{k} \right\} \geqslant 1-\frac{K}{K_{1}}\delta.
\end{align}
Noting that $\sum_{k=1}^{K} k^{\frac{1}{4}}\geqslant \frac{4}{5}K^{\frac{5}{4}}$ and $\sum_{k=1}^{K} k^{\frac{-1}{12}}\leqslant \frac{12}{11}K^{\frac{11}{12}}$, we have
\begin{align}
\frac{\sum_{k=1}^{K} k^{\frac{1}{4}}\left\|\theta_{k}-\overline{\theta_{k}}
\right\|^{2}}{\sum_{k=1}^{K}k^{\frac{1}{4}}}
 \leqslant &2\frac{K^{\frac{1}{4}}D^{2}}{\sum_{k=1}^{K} k^{\frac{1}{4}}}+\frac{2\left(\frac{1}{\frac{1}
 {\beta}-L}+2\beta\right)\sum_{k=1}^{K}k^{\frac{1}{4}} P_{k}}{\sum_{k=1}^{K} k^{\frac{1}{4}}}\notag\\
 \leqslant & \frac{5D^{2}}{2}\frac{1}{K}+2\left(\frac{1}{\frac{1}
 {\beta}-L}+2\beta\right)\left(C_{3}\sqrt{M}
+C_{4}\right) \notag\\
&+\frac{30}{11}\left(\frac{1}
 {\frac{1}{\beta}-L}+2\beta\right)C_{5}\frac{1}{K^{\frac{1}{3}}}
 +(\phi(K))^{\frac{1}{2}},  \forall \ w \in \widetilde{A}\notag.
\end{align}
Then, by Lyapunov inequality and the above inequality, we have
\begin{align}
\frac{\sum_{k=1}^{K} k^{\frac{1}{4}}\left\|\theta_{k}-\overline{\theta_{k}}
\right\|}{\sum_{k=1}^{K}k^{\frac{1}{4}}}
\leqslant &
\sqrt{\frac{\sum_{k=1}^{K} k^{\frac{1}{4}}\left\|\theta_{k}-\overline{\theta_{k}}
\right\|^{2}}{\sum_{k=1}^{K}k^{\frac{1}{4}}}}\notag\\
\leqslant &
\frac{\sqrt{10}}{2}D\frac{1}{K^{\frac{1}{2}}}+\sqrt{2}
\left(\frac{1}{\frac{1}
 {\beta}-L}+2\beta\right)^{\frac{1}{2}}\left(C_{3}\sqrt{M}
+C_{4}\right)^{\frac{1}{2}}\notag\\
&+\sqrt{\frac{30}{11}}\left(\frac{1}
 {\frac{1}{\beta}-L}+2\beta\right)^{\frac{1}{2}}C_{5}^{\frac{1}{2}}
 \frac{1}{K^{\frac{1}{6}}}
+(\phi(K))^{\frac{1}{4}},  \forall \ w \in \widetilde{A}\label{part63}.
\end{align}
By Assumptions \ref{assumption0}-\ref{assumption2}, similar to the proof of $(\ref{prox231})$ in Lemma \ref{exactlemma}, we have
\begin{equation}\nonumber
\begin{aligned}
\langle F(\theta_{k}), \theta_{k}-\theta\rangle\leqslant \left(4LD+\frac{2D}{\beta}+\frac{\sqrt{N}B_{\Theta}U_{R}}{(1-\gamma)^{2}}
\right)\left\|\theta_{k}-\overline{\theta_{k}}\right\|.
\end{aligned}
\end{equation}
Then combining (\ref{part63}) with the above inequality gives
\begin{align}
\frac{\sum_{k=1}^{K} k^{\frac{1}{4}}\left( \sup \limits_{\theta \in \Theta} \left\langle F(\theta_{k}), \theta_{k}-\theta \right\rangle \right)}{\sum_{k=1}^{K}k^{\frac{1}{4}}}
\leqslant&\left(4LD+\frac{2D}{\beta}+\frac{\sqrt{N}B_{\Theta}
U_{R}}
{(1-\gamma)^{2}}
\right)
\frac{\sum_{k=1}^{K} k^{\frac{1}{4}}\left\|\theta_{k}-\overline{\theta_{k}}
\right\|}{\sum_{k=1}^{K}k^{\frac{1}{4}}}\notag\\
\leqslant& L_{1}\frac{1}{K^{\frac{1}{2}}}+L_{2}+L_{3}
 \frac{1}{K^{\frac{1}{6}}}
+(\phi(K))^{\frac{1}{4}}, \forall \ w \in \widetilde{A}\label{unmeasurable},
\end{align}
where $L_{1}=\frac{\sqrt{10}}{2}D\left(4LD+\frac{2D}{\beta}+\frac{\sqrt{N}B_{\Theta}
U_{R}}
{(1-\gamma)^{2}}
\right)$, $L_{2}=\sqrt{2}
\left(\frac{1}{\frac{1}
 {\beta}-L}+2\beta\right)^{\frac{1}{2}}\Big(4LD+\frac{2D}{\beta}\\+\frac{\sqrt{N}B_{\Theta}
U_{R}}
{(1-\gamma)^{2}}
\Big)$$\big(C_{3}\sqrt{M}
+C_{4}\big)^{\frac{1}{2}}$ and $L_{3}=\sqrt{\frac{30}{11}}\left(4LD+\frac{2D}{\beta}+\frac{\sqrt{N}B_{\Theta}
U_{R}}
{(1-\gamma)^{2}}
\right)\left(\frac{1}
 {\frac{1}{\beta}-L}+2\beta\right)^{\frac{1}{2}}C_{5}^{\frac{1}{2}}$.
 By Assumptions \ref{assumption0}-\ref{assumption3} and Lemma \ref{gradientd}, we have
\begin{equation}\nonumber
\begin{aligned}
&\sup _{i \in \mathcal{N}}\frac{\sum_{k=1}^{K} k^{\frac{1}{4}}\left(\sup \limits_{\theta_{i} \in \Theta_{i}}J_{i}^{(\theta_{i}, \theta_{-i,k})}(t)-J_{i}^{(\theta_{i,k}, \theta_{-i,k})}(t)\right)}{\sum_{k=1}^{K}k^{\frac{1}{4}}}\\
\leqslant &
\sup _{i \in \mathcal{N}}\frac{\sum_{k=1}^{K} k^{\frac{1}{4}}M_{1}\left( \sup \limits_{\theta_{i} \in \Theta_{i}} \langle F_{i}(\theta_{k}), \theta_{i,k}-\theta_{i}\rangle\right)}{\sum_{k=1}^{K}
k^{\frac{1}{4}}}.
\end{aligned}
\end{equation}
Noting that $\sup \limits_{\theta_{i} \in \Theta_{i}}\left\langle F_{i}(\theta_{k}), \theta_{i,k}-\theta_{i}\right\rangle \leqslant  \sup \limits_{\theta \in \Theta} \left\langle F(\theta_{k}), \theta_{k}-\theta \right\rangle$, for any countable set $\widetilde{\Theta_{i}} \subseteq \Theta_{i}$, it follows that
$\sup \limits_{\theta_{i} \in \widetilde{\Theta_{i}}}J_{i}^{(\theta_{i}, \theta_{-i,k})}(t)-J_{i}^{(\theta_{i,k}, \theta_{-i,k})}(t) \leqslant \sup \limits_{\theta_{i} \in \Theta_{i}}J_{i}^{(\theta_{i}, \theta_{-i,k})}(t)-J_{i}^{(\theta_{i,k}, \theta_{-i,k})}(t)$. This together with the above inequality leads to
\begin{align*}
&\sup _{i \in \mathcal{N}}\frac{\sum_{k=1}^{K} k^{\frac{1}{4}}\left(\sup \limits_{\theta_{i} \in \widetilde{\Theta_{i}}}J_{i}^{(\theta_{i}, \theta_{-i,k})}(t)-J_{i}^{(\theta_{i,k}, \theta_{-i,k})}(t)\right)}{\sum_{k=1}^{K}k^{\frac{1}{4}}}\\
\leqslant &
\sup _{i \in \mathcal{N}}\frac{\sum_{k=1}^{K} k^{\frac{1}{4}}\left(\sup \limits_{\theta_{i} \in \Theta_{i}}J_{i}^{(\theta_{i}, \theta_{-i,k})}(t)-J_{i}^{(\theta_{i,k}, \theta_{-i,k})}(t)\right)}{\sum_{k=1}^{K}k^{\frac{1}{4}}}\\
\leqslant &
\frac{\sum_{k=1}^{K} k^{\frac{1}{4}}M_{1}\left(  \sup \limits _{\theta \in \Theta} \left\langle F(\theta_{k}), \theta_{k}-\theta \right\rangle\right)}{\sum_{k=1}^{K}
k^{\frac{1}{4}}}.
\end{align*}
Hence, by (\ref{unmeasurable}) and the above inequality, we have
\begin{align*}
\sup _{i \in \mathcal{N}}\frac{\sum_{k=1}^{K} k^{\frac{1}{4}}\left(\sup \limits_{\theta_{i} \in \widetilde{\Theta_{i}}}J_{i}^{(\theta_{i}, \theta_{-i,k})}(t)-J_{i}^{\theta_{k}}(t)\right)}
{\sum_{k=1}^{K}k^{\frac{1}{4}}}\leqslant L(K,K_{1},T,\delta)
+(\phi(K))^{\frac{1}{4}}, \forall \ w \in \widetilde{A}.
\end{align*}
Then, from (\ref{intersectionwidetildeAk}) and the above inequality, we have
\begin{align*}
P\Bigg\{\sup _{i \in \mathcal{N}}\frac{\sum_{k=1}^{K} k^{\frac{1}{4}}\left(\sup \limits_{\theta_{i} \in \widetilde{\Theta_{i}}}J_{i}^{(\theta_{i}, \theta_{-i,k})}(t)-J_{i}^{\theta_{k}}(t)\right)}
{\sum_{k=1}^{K}k^{\frac{1}{4}}}\leqslant L(K,K_{1},T,\delta)
+(\phi(K))^{\frac{1}{4}}\Bigg\} \geqslant 1-\frac{K}{K_{1}}\delta,
\end{align*}
that is, Lemma \ref{lemmapseudo} holds.
\end{proof}

\begin{theorem}\label{pseudogradient}
If Assumptions \ref{assumption0}-\ref{assumption4} hold, and we choose $\eta_{h}=\frac{l_{1}}{h^{\frac{2}{3}}}$, $\gamma_{k}=k^{\frac{1}{4}}$, $H_{k}=k$ and $K_{1}=K$ in Algorithm \ref{alg::pseudoGradient}, then $\big((\pi_{\theta_{i,k}})_{i=1}^{N}\big)
_{k=1}^{\infty}$  given by Algorithm \ref{alg::pseudoGradient} is a $k^{\frac{1}{4}}$-weighted asymptotic Nash equilibrium of the game $\Gamma$ in probability.
\end{theorem}
\begin{proof}
If Assumptions \ref{assumption0}-\ref{assumption4} hold, then by Lemma \ref{lemmapseudo}, for any countable set $\widetilde{\Theta_{i}} \subseteq \Theta_{i}$, we have
\begin{equation}\label{Landphi}
\begin{aligned}
P\Bigg\{\sup _{i \in \mathcal{N}}\frac{\sum_{k=1}^{K} k^{\frac{1}{4}}\left(\sup \limits_{\theta_{i} \in \widetilde{\Theta_{i}}}J_{i}^{(\theta_{i}, \theta_{-i,k})}(t)-J_{i}^{\theta_{k}}(t)\right)}
{\sum_{k=1}^{K}k^{\frac{1}{4}}}\leqslant R_{1}(T,K)+R_{2}(K)\Bigg\} \geqslant 1-\delta, \notag \\ \ \forall \ \delta \in (0,1],
\end{aligned}
\end{equation}
where $R_{1}(T,K)=\sqrt{2}M_{1}\left(\frac{1}{\frac{1}
 {\beta}-L}+2\beta\right)^{\frac{1}{2}}\left(4LD+\frac{\sqrt{N}B_{\Theta}U_{R}}{(1-\gamma)^{2}}
+\frac{2D}{\beta}\right)
\Bigg[\frac{\sqrt{M}}{2\sqrt{L^{2}+\frac{1}{\beta^{2}}}}\Bigg(4D
\Big(2L^{2}\\+\frac{2}{\beta^{2}}\Big)^{\frac{1}{2}}+\left(\frac{2NB_{\Theta}^{2}U_{R}^{2}}{(1-\gamma)^{2}}
+\frac{8D^{2}}{\gamma^{2}}\right)^{\frac{1}{2}}\Bigg)
+4D\left(1+\sqrt{2}\right)\left(\left(2L^{2}+\frac{2}{\beta^{2}}\right) \frac{D\sqrt{M}}{\frac{1}{\beta}-L}\right)^{\frac{1}{2}}\Bigg]^{\frac{1}{2}}$, $R_{2}(K)=\frac{\sqrt{10}D}{2}\\ \Big(4LD+\frac{2D}{\beta}+\frac{\sqrt{N}B_{\Theta}U_{R}}
{(1-\gamma)^{2}}
\Big) \frac{M_{1}}{K^{\frac{1}{2}}}+
\sqrt{\frac{30}{11}}  \left(\frac{1}{\frac{1}{\beta}-L}+2\beta
\right)^{\frac{1}{2}} \left(4LD+\frac{2D}{\beta}+\frac{\sqrt{N}B_{\Theta}U_{R}}{(1-\gamma)^{2}}
\right)
\Bigg(2D \Big(2L^{2}\\+\frac{2}{\beta^{2}}\Big)^{\frac{1}{2}}
(2+\sqrt{2})
 \left(\frac{l_{1}\big(9M+\frac{D^{2}}{2}\big)}
 {\big(\frac{1}{\beta}-L\big)}\right)^{\frac{1}{2}}\Bigg)^{\frac{1}{2}}
\frac{M_{1}}{K^{\frac{1}{6}}}
(\phi(K))^{\frac{1}{4}}$.
From $M=M\big(T, K, \delta\big)=2N\big(B_{\Theta} U_{R}\big)^{2}\\ \Big[\left(\frac{T+1}{1-\gamma}+
\frac{\gamma}{(1-\gamma)^{2}}\right)  \gamma^{T+1}\Big]^{2}
+16\log\left(\frac{8K}{\delta}\right)
\frac{NB_{\Theta}^{2}U_{R}^{2}\gamma^{2}}{(1-\gamma)^{4}K}
$ and $\gamma \in(0,1)$, we have $\lim \limits_{T \rightarrow \infty,K \rightarrow \infty}R_{1}(T,K)\\=0$. Therefore,  for any $\epsilon >0$, there exist $\widetilde{T} > 0$ and $K_{2} > 0$ such that $R_{1}(T,K) \leqslant \frac{\epsilon}{2}$ if $T \geqslant \widetilde{T}$ and $K \geqslant K_{2}$. Recalling that $(\phi(K))^{\frac{1}{4}}=o\left(\frac{1}{K^{\frac{1}{6}}}\right)$, we have $\lim \limits_{K \rightarrow \infty}R_{2}(K)=0$.  Therefore,  there exists $K_{3} > 0$ such that $R_{2}(K) \leqslant \frac{\epsilon}{2}$ if $K \geqslant K_{3}$.
To sum up, for any $\delta \in (0,1]$ and $\epsilon >0$, there exist $\widetilde{T}$ and $\widetilde{K}=\max\left\{K_{2}, K_{3}\right\}$
such that $R_{1}(T,K)+R_{2}(K)
\leqslant \epsilon$ if
$T \geqslant \widetilde{T}$ and $K \geqslant \widetilde{K}$. This together with (\ref{Landphi}) gives
\begin{equation}\nonumber
P\left\{\sup _{i \in \mathcal{N}}\frac{\sum_{k=1}^{K}k^{\frac{1}{4}}\left(\sup \limits _{\theta_{i} \in \widetilde{\Theta_{i}}}J_{i}^{(\theta_{i}, \theta_{-i,k})}(t)-J_{i}^{(\theta_{i,k}, \theta_{-i,k})}(t)\right)}{\sum_{k=1}^{K}k^{\frac{1}{4}}} \leqslant \epsilon\right\}\geqslant 1-\delta.
\end{equation}
Hence, from Definition \ref{nashprobability}, we know
$\big((\pi_{\theta_{i,k}})_{i=1}^{N}\big)
_{k=1}^{\infty}$ is a $k^{\frac{1}{4}}$-weighted asymptotic Nash equilibrium of the game $\Gamma$ in probability.
\end{proof}

\begin{remark}
For the stochastic game with the finite state and action space, if we consider direct parameterization, that is, $\pi_{\theta_{i}}(a_{i}\mid s)=\theta_{ s,a_{i}}$, where $\theta_{s,a_{i}}\geqslant 0$ and $\sum_{a_{i} \in \mathcal{A}_{i}}\theta_{s,a_{i}}=1$, then the direct parameterization doesn't meet with the conditions in Assumption \ref{assumption2}. However, there exists a similar conclusion to Lemma \ref{FLIPSCHITZ} in Lemma 7 in \cite{Runyu}. To be more exact, constant $L$ changes into $\frac{2U_{R}\sum_{i=1}^{n}|\mathcal{A}_{i}|}{1-\gamma^{3}}$,
where $|\mathcal{A}_{i}|$ is the number of actions of agent $i$. If we consider $\alpha$-greedy direct parameterization, the parameterization doesn't meet with the conditions in Assumption \ref{assumption2} neither.
we can still come to the same conclusion as  Lemma \ref{FLIPSCHITZ} by means of the proof of Lemma 7 in  \cite{Runyu}.

For the case with finite state and action space, direct parameterization satisfies gradient dominant  theorem from Lemma 3 in \cite{Runyu}. Since $\alpha$-greedy direct parameterization satisfies Assumption \ref{assumption1}, then gradient domination theorem holds for $\alpha$-greedy direct parameterization by Lemma \ref{gradientd}. If Assumption \ref{assumption3} holds, the Nash equilibrium problem can be equivalent to SVI$(F,\Theta)$ under the two classes of parameterization. Furthermore, we have illustrated that the two classes of parameterization satisfies Lemma \ref{FLIPSCHITZ}. So, Algorithm \ref{alg::exactGradient} can be used in the context of the two kinds of parameterization. Thus, Theorem \ref{exactgradient} holds for the two kinds of parameterization. Similarly, Algorithm \ref{alg::pseudoGradient} can be used under $\alpha$-greedy direct parameterization.
\end{remark}

\section{Numerical Example}
Consider a two-person stochastic game, where the set of agents, the state space and the action space are given by  $\mathcal{N}=\{1,2\}$, $\mathcal{S}=\{1,2\}$, $\mathcal{A}_{1}=\mathcal{A}_{2}=\{1,2\}$.  We assume that  the immediate rewards of agents are independent of the state and are given in  Table \ref{Reward Table}.  The discount factor $\gamma$ is taken as $0.9$.  We consider $\alpha$-greedy direct parameterization, that is, $\pi_{\theta_{i}}(a_{i}=1\mid s=1)=(1-\alpha)\theta_{i}(1)+\alpha/2$, $\pi_{\theta_{i}}(a_{i}=2\mid s=1)=(1-\alpha)\theta_{i}(1)+\alpha/2$, $\pi_{\theta_{i}}(a_{i}=1\mid s=2)=(1-\alpha)\theta_{i}(2)+\alpha/2$, $\pi_{\theta_{i}}(a_{i}=2\mid s=2)=(1-\alpha)\theta_{i}(2)+\alpha/2$, $ i=1,2$, where $\theta=[\theta_{1};\theta_{2}]=[\theta_{1}(1), \theta_{1}(2);\theta_{2}(1), \theta_{2}(2)] \in \Theta$ and $\Theta=\Theta_{1}\times \Theta_{2} \subseteq [0,1]^2\times [0,1]^2 $ is parameters set.  We choose $\alpha=0.01$.
 The initial probability density function $\rho_{t}(s)$ of the environment state at time $t$ is $\frac{1}{2}\delta(s-1)+\frac{1}{2}\delta(s-2)$.
The Markov kernel  of the induced Markov chain
$\rho_{t,t+1}^{\theta}(s\mid s=1)=\frac{3}{5}\delta(s-1)+\frac{2}{5}\delta(s-2)$ and $\rho_{t,t+1}^{\theta}(s\mid s=2)=\frac{7}{10}\delta(s-1)+\frac{3}{10}\delta(s-2)$, $\forall \ \theta \in \Theta$. By (\ref{valuefunction}) and Proposition \ref{policygradient}, we have the pseudo gradient $F(\theta)=[F_{1}(\theta),F_{2}(\theta)]$, $M_{1}=1$ and  Lipschitz constant $L=5.63$,  where $F_{1}(\theta)=\Bigg[-\frac{(1.4\gamma -\gamma+1)(\alpha-1)((1-\alpha)\theta_{2}(1)
+2+\frac{\alpha}{2})}{2(\gamma+0.6\gamma -0.7\gamma +0.7\gamma^{2}-0.6\gamma^{2}-1)},\frac{(2\gamma 0.6-\gamma-1)(\alpha-1)((1-\alpha)\theta_{2}(2)
+2+\frac{\alpha}{2})}{2(\gamma+0.6\gamma -0.7\gamma +0.7\gamma^{2}-0.6\gamma^{2}-1)}\Bigg]$, $F_{2}(\theta)=\Bigg[\frac{(1.4\gamma -\gamma+1)(\alpha-1)((\alpha-1)\theta_{1}(1)
-2-\frac{\alpha}{2})}{2(\gamma+0.6\gamma -0.7\gamma +0.7\gamma^{2}-0.6\gamma^{2}-1)},-\frac{(1.4\gamma -\gamma-1)(\alpha-1)(-(1-\alpha)\theta_{2}(2)
-2-\frac{\alpha}{2})}{2(\gamma+0.6\gamma -0.7\gamma +0.7\gamma^{2}-0.6\gamma^{2}-1)}\Bigg]$. Then, it follows that  Assumption
\ref{assumption1}, Assumption \ref{assumption3} and Assumption \ref{assumption4} hold  and $(0 ,0,0,0)$ is the solution of MVI$(F,\Theta)$.
 Denote $\epsilon_{1,K}=\sup \limits_{i \in \mathcal{N}} \frac{\sum_{k=1}^{K} k^{\frac{1}{2}}\left(\sup \limits_{\theta_{i} \in \Theta_{i}} J_{i}^{\left(\theta_{i}, \theta_{-i, k}\right)}(t)-J_{i}^{\left(\theta_{i, k}, \theta_{-i, k}\right)}(t)\right)}{\sum_{k=1}^{K} k^{\frac{1}{2}}}$ and $\epsilon_{2,K}=\sup \limits_{i \in \mathcal{N}} \frac{\sum_{k=1}^{K} k^{\frac{1}{4}}\left(\sup \limits_{\theta_{i} \in \Theta_{i}} J_{i}^{\left(\theta_{i}, \theta_{-i, k}\right)}(t)-J_{i}^{\left(\theta_{i, k}, \theta_{-i, k}\right)}(t)\right)}{\sum_{k=1}^{K} k^{\frac{1}{4}}}$. Then, we perform  simulations for Algorithm \ref{alg::exactGradient} and Algorithm \ref{alg::pseudoGradient}, respectively.

 \begin{table}[H]
 \centering
 \caption{Reward Table.}
 \begin{tabular}{|c|c|c|}
    \hline
          &$a_{2}=1$ & $a_{2}=2$\\
          \hline
   $a_{1}=1$& $(3,3)$ & $(4,0)$\\
   \hline
   $a_{1}=2$& $(0,4)$ & $(2,2)$ \\
   \hline
  \end{tabular}
  \label{Reward Table}
 \end{table}

 \begin{figure}[H]
  \begin{minipage}[b]{.5\linewidth}
  \centering
  \includegraphics[height=3.5cm,width=7cm]{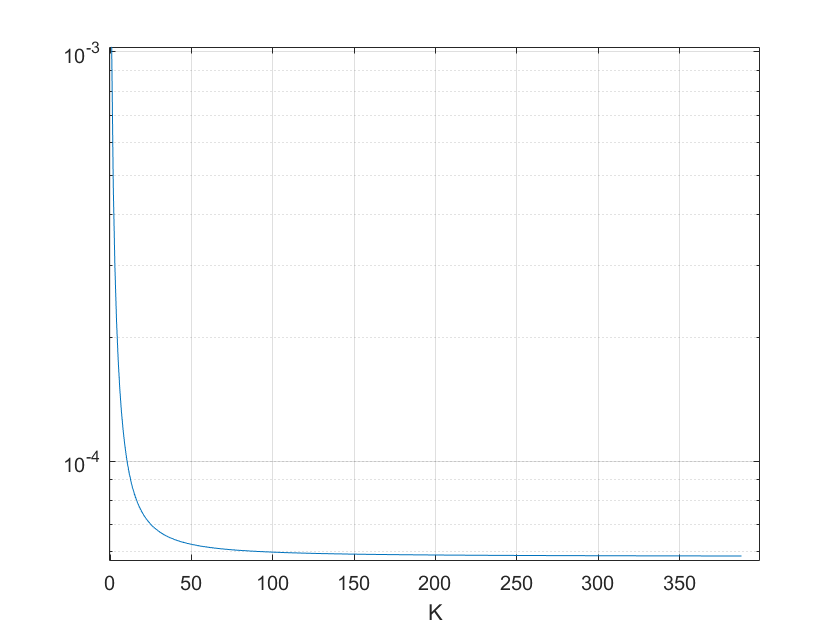}
\caption{Curve of $\epsilon_{1,K}$ with respect to $K$ \\ for
 the case with the exact pseudo gradients.}
\label{exactgradientgraph}
\end{minipage}
  \begin{minipage}[b]{.5\linewidth}
  \centering
  \includegraphics[height=3.5cm,width=7cm]{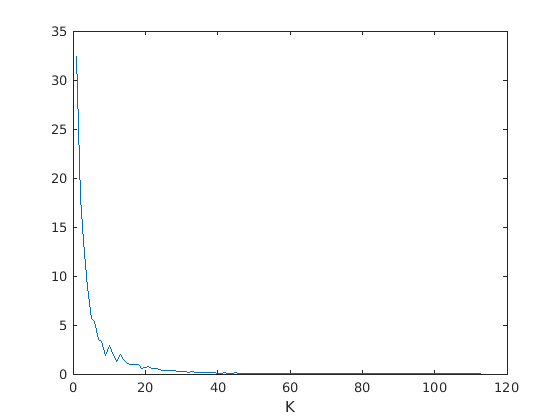}
\caption{Curve of $\epsilon_{2,K}$ with respect to $K$ \\ for  the case with the unknown \\
pseudo gradients.}
\label{unknownpseudograph}
\end{minipage}
 \end{figure}
For the case with the exact pseudo gradients (Algorithm \ref{alg::exactGradient}), Figure \ref{exactgradientgraph}  shows the curve of $\epsilon_{1,K}$ with respect to $K$. It can be seen  that $\epsilon_{1,K}$ vanishes to $0$ as $K$ increases, which implies $\big((\pi_{\theta_{i,k}})_{i=1}^{N}\big)
_{k=1}^{\infty}$ is a $k^{\frac{1}{2}}$-weighted asymptotic Nash equilibrium of the game, that is, Theorem \ref{exactgradient} follows.
  For the case with the unknown pseudo gradients (Algorithm \ref{alg::pseudoGradient}), the sampling time horizon  of the G(PO)MDP estimator  and  the number of trajectories are taken as $T=20$ and $K_{1}=K+1$. Figure \ref{unknownpseudograph}  shows  $\epsilon_{2,K}$ vanishes to $0$ as $K$ increases.
  It can be seen from Figure \ref{unknownpseudograph} that $\big((\pi_{\theta_{i,k}})_{i=1}^{N}\big)
_{k=1}^{\infty}$ is a $k^{\frac{1}{4}}$-weighted asymptotic Nash equilibrium of the game  with probability one, which is a better result comparing from Theorem \ref{pseudogradient} where $\big((\pi_{\theta_{i,k}})_{i=1}^{N}\big)
_{k=1}^{\infty}$ is a $k^{\frac{1}{4}}$-weighted asymptotic Nash equilibrium of the game in probability. Comparing Figure \ref{exactgradientgraph} with Figure \ref{unknownpseudograph}, it can be seen that Algorithm \ref{alg::pseudoGradient} have a  slower convergence rate than Algorithm \ref{alg::exactGradient}.

\section{Conclusion}
The general-sum stochastic game  with an unknown transition probability density function is investigated in this paper. Each agent only observes the environment state and its own reward and is
unknown about the transition  probability density function of the environment state and  the others' actions and rewards. We define the concepts of weighted asymptotic Nash equilibrium for a given sequence with probability $1$ and in probability  under policy parameterization and prove the equivalence between Nash equilibrium  and variational inequality problems. We have proposed two-loop algorithms for the solution of the variational inequality problem for the cases with exact and unknown pseudo gradients, respectively. In the outer loop, we sequentially update the constructed strongly monotone variational inequality and we employ a single-call extra-gradient algorithm for solving the constructed strongly monotone variational inequality in the inner loop. It is shown that the algorithm is convergent to the $k^{\frac{1}{2}}$-weighted asymptotic  Nash equilibrium in the context of exact pseudo gradients. Further, in the context of unknown pseudo gradients, a decentralized algorithm is proposed by leveraging the G(PO)MDP gradient estimator of the pseudo gradient. Also, we provide the convergence  guarantee  to the $k^{\frac{1}{4}}$-weighted asymptotic  Nash equilibrium in probability.

It is worth noting that our assumptions about parameterization do not hold for softmax parameterization. So, how to design an algorithm for learning Nash equilibria under softmax parameterization is a future direction. It is also worth considering  how to make better-verified assumptions than Assumption \ref{assumption4} and choose appropriate parameterization so that pseudo gradients have better properties. In addition, we would consider how to design an algorithm to estimate the pseudo gradient with a lower computation complexity in future.

\appendix
\section*{Appendix A Supplement Definitions and Lemmas}

\begin{definition}\label{definitionofmonotone}
~\citep{FFacchinei} Mapping $G(x): K \rightarrow \mathbb{R}^{n}$ is
\begin{itemize}
    \item[(i)]  monotone: if
\begin{equation*}
    \left\langle G(x)-G(x^{\prime}), x-x^{\prime}\right\rangle \geqslant 0, \quad \forall\ x,\ x^{\prime} \in K;
\end{equation*}
  \item[(ii)]  strongly monotone: if
\begin{equation*}
    \left\langle G(x)-G(x^{\prime}), x-x^{\prime}\right\rangle > 0, \quad \forall\ x,\ x^{\prime} \in K;
\end{equation*}
 \item[(iii)] $\mu$-strongly monotone: if there exists a constant $\mu >0$, such that
\begin{equation*}
\left\langle G(x)-G(x^{\prime}), x-x^{\prime}\right\rangle \geqslant \mu\left\|x-x^{\prime}\right\|^{2}, \quad \forall\ x,\  x^{\prime} \in K.
\end{equation*}
\end{itemize}
\end{definition}

\begin{definition}\label{definitionofgap}
~\citep{FFacchinei}
For SVI$(G,K)$, its prime gap function is
$
G_{\text {gap }}(x) \equiv \sup\limits _{y \in K} G^{T}(x)(x-y),\  x \in  K
$ and its dual gap function is
$G_{\text {dual }}(x) \equiv \sup \limits_{y \in K} G^{T}(y)(x-y),\  x \in K$.
\end{definition}

\begin{lemma}\label{stronglylipschitz}
If Assumptions $\ref{assumption0}$-\ref{assumption2} hold, then for any $\theta \in \Theta$,  $F_{k}(\theta)$ is $\sqrt{2L^{2}+\frac{2}{\beta^{2}}}$-Lipschitz continuous with respect to $\theta$, where $L$ is given by Lemma \ref{FLIPSCHITZ}.
\end{lemma}
\begin{proof}
If Assumptions $\ref{assumption0}$-\ref{assumption2} hold, then by Lemma \ref{FLIPSCHITZ}, it follows that
$$
\begin{aligned}
\left\|F_{k}(\theta)-F_{k}(\theta^{\prime})\right\|^{2}&
=\left\|F(\theta)-F(\theta^{\prime})
+\frac{1}{\gamma}(\theta-\theta^{\prime})\right\|^{2}\\
&\leqslant 2\left\|F(\theta)-F(\theta^{\prime})\right\|^{2}+2\left\|\frac{1}{\gamma}(\theta-\theta^{\prime})\right\|^{2}\\
&\leqslant 2L^{2}\left\|\theta-\theta^{\prime}\right\|^{2}+\frac{2}{\gamma^{2}}\left\|\theta-\theta^{\prime}\right\|^{2}\\
&\leqslant \left(2L^{2}+\frac{2}{\beta^{2}}\right)\left\|\theta-
\theta^{\prime}\right\|^{2},  \forall \ \theta, \ \theta^{\prime} \in \Theta.
\end{aligned}
$$
\end{proof}

\begin{lemma}\label{stronglymonotone}
If Assumptions $\ref{assumption0}$-\ref{assumption2} hold, then $F_{k}(\theta)$ is $\left(\frac{1}{\beta}-L\right)$-strongly monotone with respect to
$\theta \in \Theta$, where $L$ is given by Lemma \ref{FLIPSCHITZ}.
\end{lemma}
\begin{proof}
If Assumptions $\ref{assumption0}$-\ref{assumption2}  hold, by Lemma \ref{FLIPSCHITZ}, it follows that
\begin{align*}
\left\langle  F_{k}(\theta)-
F_{k}(\theta^{\prime}), \theta-\theta^{\prime}\right\rangle&=
\left\langle  F(\theta)+\frac{1}{\beta}(\theta-\theta_{k})-
F(\theta^{\prime})-\frac{1}{\beta}\left(\theta^{\prime}-\theta_{k} \right), \theta-\theta^{\prime}\right\rangle\\
&\geqslant \left\langle F(\theta)-F(\theta^{\prime}), \theta-\theta^{\prime}\right\rangle+\frac{1}{\beta}
\left\|\theta-\theta^{\prime}\right\|^{2}\\
&\geqslant -L\left\|\theta-\theta^{\prime}\right\|^{2}+\frac{1}{\beta}\left\|\theta-\theta^{\prime}\right\|^{2}\\
&\geqslant \left(\frac{1}{\beta}-L\right)
\left\|\theta-\theta^{\prime}\right\|^{2},  \forall \ \theta, \ \theta^{\prime} \in \Theta.
\end{align*}
By Definition \ref{definitionofmonotone} (iii), the lemma is true.
\end{proof}

\begin{lemma}\label{weaklyandstrongly}
If Assumptions $\ref{assumption0}$-\ref{assumption2} hold and denote $\hat{\theta}$ as the solution of
SVI$(F_{k},\Theta)$, then
\begin{equation}\nonumber
\sup _{\theta \in \Theta}\left\langle  F(\hat{\theta}), \hat{\theta}-\theta\right\rangle  \leqslant \frac{2D}{\beta} \left\|\theta_{k}-\hat{\theta}\right\|.
\end{equation}
In particular, if $\theta=(\theta_{i},\hat{\theta}_{-i})$, then
\begin{equation}\nonumber
\sup _{\theta_{i} \in \Theta_{i}}\left\langle  F_{i}(\hat{\theta}), \hat{\theta_{i}}-\theta_{i}\right\rangle  \leqslant \frac{2D}{\beta} \left\|\theta_{k}-\hat{\theta}\right\|.
\end{equation}
\end{lemma}
\begin{proof}
Noting that  $\hat{\theta}$ is the solution of
SVI$(F_{k},\Theta)$, we have
$
\sup \limits_{\theta \in \Theta}\big\langle  F_{k}(\hat{\theta}), \hat{\theta}-\theta\big\rangle=\sup \limits _{\theta \in \Theta} \big\langle  F_{k}(\hat{\theta})+\frac{1}{\beta}(\hat{\theta}-\theta_{k}), \hat{\theta}-\theta\big\rangle \leqslant 0
$, which together with Assumption \ref{assumption1} gives
$
\sup \limits_{\theta \in \Theta}\big\langle  F(\hat{\theta}), \hat{\theta}-\theta\big\rangle  \leqslant \frac{1}{\beta}
\big\langle  \hat{\theta}-\theta_{k}, \hat{\theta}-\theta\big\rangle \leqslant \frac{2D}{\beta} \left\|\theta_{k}-\hat{\theta}\right\|.
$
\end{proof}

\begin{lemma}\label{performancedifference}\citep{Runyu}
If Assumption $\ref{assumption0}$ hold, then for any $\theta=\left(\theta_{i}, \theta_{-i}\right)$ and  $\hat{\theta}^{i}=\left(\theta_{i}^{\prime}, \theta_{-i}\right)
 \in \Theta$, we have
$J_{i}^{(\theta_{i}^{\prime}, \theta_{-i})}(t) -J_{i}^{(\theta_{i}, \theta_{-i})}(t)
=\frac{1}{1-\gamma} \int_{\mathcal{S}\times \mathcal{A}_{i}} d_{\rho_{t}}^{\hat{\theta}^{i}}(s^{\prime}) \pi_{\theta_{i}^{\prime}}\left(a_{i} \mid s^{\prime}\right) \\
\times \Big(\int_{\mathcal{A}_{-i}} \pi_{\theta_{-i}}\left(a_{-i} \mid s^{\prime}\right) A_{i}^{\theta}\left(s^{\prime}, a,t+1\right) \mathrm{d} a_{-i}\Big)\mathrm{d} a_{i} \mathrm{d} s^{\prime}$, $\forall \ i \in \mathcal{N}$.
\end{lemma}

\begin{lemma}\citep{Y}\label{propertyofprox}
Define proximal mapping $P_{\Theta_{i}}(x-y): \mathbb{R}^{d_{i}} \times \mathbb{R}^{d_{i}} \rightarrow \Theta_{i}$ as
$$
P_{\Theta_{i}}(x-y)=\underset{z \in \Theta_{i}}{\operatorname{argmin}}\left\{2\langle y, z\rangle+\left\|x-z\right\|^{2}\right\}, \quad \forall \ x ,\\\ y \in \mathbb{R}^{d_{i}}.
$$
\begin{itemize}

\item[(i)] For any $x$, $y$, $z \in \mathbb{R}^{d_{i}}$,
    $$
\frac{\left\|P_{\Theta_{i}}(x-y)- z\right\|^{2}}{2}
 \leqslant \frac{\left\|x- z\right\|^{2}}{2}+\left\langle y, z-P_{\Theta_{i}}(x-y)\right\rangle-\frac{\left\|x-P_{\Theta_{i}}(x-y))\right\|^{2}}{2};
    $$

\item[(ii)] For any $x_{1}$, $z \in \mathbb{R}^{d_{i}}$ and $y_{1}$, $y_{2} \in \mathbb{R}^{d_{i}}$,
 \begin{equation}\nonumber
\begin{aligned}
  & \frac{\left\|P_{\Theta_{i}}(x_{1}-y_{2})- z\right\|^{2}}{2}\\
\leqslant &  \frac{\left\|x_{1}- z\right\|^{2}}{2}-\left\langle y_{2}, P_{\Theta_{i}}(x_{1}-y_{1})-z)\right\rangle+
\left\langle y_{2}-y_{1}, P_{\Theta_{i}}(x_{1}-y_{1})-P_{\Theta_{i}}(x_{1}-y_{2})\right\rangle\\
&-\frac{\left\|P_{\Theta_{i}}(x_{1}-y_{1})-P_{\Theta_{i}}(x_{1}-y_{2}))\right\|^{2}}{2}-\frac{\left\|P_{\Theta_{i}}(x_{1}-y_{1})-x_{1}\right\|^{2}}{2}\\
\leqslant &\frac{\left\|x_{1}- z\right\|^{2}}{2}-\left\langle y_{2}, P_{\Theta_{i}}(x_{1}-y_{1})-z)\right\rangle+\frac{\left\|y_{2}-y_{1}\right\|^{2}}{2}-\frac{\left\|P_{\Theta_{i}}(x_{1}-y_{1})- x_{1}\right\|^{2}}{2};
\end{aligned}
\end{equation}
\item[(iii)] Non-expansion  property:
 \begin{equation}\nonumber
 \left\|P_{\Theta_{i}}(x_{1}-y_{1})-P_{\Theta_{i}}(x_{1}-y_{2})\right\| \leqslant \left\|y_{1}-y_{2}\right\|.
\end{equation}
\end{itemize}
\end{lemma}

\begin{lemma}\label{contorolstrong}
\citep{Lin} If $G$ is a $\sqrt{2L^{2}+\frac{2}{\beta^{2}}}$-Lipschitz continous mapping and strong-\\ly monotone and let $w_{*}$ be the solution of SVI$(G, K)$, then for any $\widehat{z} \in K$, by constructing  $\overline{z}=\operatorname{P}_{K}(\widehat{z}-\widetilde{\eta} G(\widehat{z}))$, where $\widetilde{\eta}=\frac{1}{\sqrt{2} \sqrt{2L^{2}+\frac{2}{\beta^{2}}}}$, we have $\sup \limits_{z \in K} G(\overline{z})^{\top}(\overline{z}-z) \leqslant 2D \sqrt{2L^{2}+\frac{2}{\beta^{2}}}(2+\sqrt{2})
\left\|\widehat{z}-w_{*}\right\|$, where $\sup \limits_{z, z^{\prime} \in K}\left\|z-z^{\prime}\right\| \leqslant 2D$ and $L$ and $\beta$ are constants.
\end{lemma}

\begin{lemma}\label{concentration inequality}
\citep{Iosif}  (Concentration inequality) If $X_{1}, X_{2}, ..., X_{N} \in \mathbb{R}^{d}$
denote a vector-valued martingale difference sequence satisfying $\big\|X_{n}\big\| \leqslant V$ and $\mathbb{E}\big[X_{n} \big| X_{1}, ..., X_{n-1}\big]\\=\mathbf{0}$, $\forall \ n \in \{1,\ldots,N\}$, then for any $\delta \in(0,1]$, we have
$$
P\left\{\left\|\sum_{n=1}^{N} X_{n}\right\|^{2}>2 \log (2 / \delta) V^{2} N\right\} \leqslant \delta.
$$
\end{lemma}
\begin{lemma} \label{lastlemma}
\citep{Chung} If $\{u_{k}, k\geqslant 1\}$ is a real sequence satisfying
$$
u_{k+1} \leqslant \left(1-\frac{c}{k^{s}}\right) u_{k}+\frac{d}{k^{t}}, \quad 0<s<1, \ s<t, \ c>0, \ d>0,
$$
then
$$
u_{k} \leqslant \frac{d}{c} \frac{1}{k^{t-s}}+\phi(k),
$$
where $\phi(k)=o\left(\frac{1}{k^{t-s}}\right)$.
\end{lemma}

\section*{Appendix B Proofs of Theorem \ref{connectionofne}, Theorem \ref{existencenash} and Lemma \ref{FLIPSCHITZ}}
\textbf{Proof of Theorem \ref{connectionofne}:}
For any countable set $\widetilde{\Theta_{i}} \subseteq \Theta_{i}$, by the property of conditional expectation, we have
\begin{align}
&\mathbb{E}\left[\sup _{\theta_{i} \in \widetilde{\Theta_{i}}} J_{i}^{(\theta_{i}, \theta_{-i,\tau_K})}(t)-J_{i}^{(\theta_{i,\tau_K}, \theta_{-i,\tau_K})}(t)\biggm|\theta_{k},k=1,...,K\right] \notag\\
=&\mathbb{E}\left\{\mathbb{E}\left[\sup _{\theta_{i} \in \widetilde{\Theta_{i}}}J_{i}^{(\theta_{i}, \theta_{-i,\tau_K})}(t)-J_{i}^{(\theta_{i,\tau_K}, \theta_{-i,\tau_K})}(t)\biggm|\tau_K,\theta_{k},k=1,...,K\right]\biggm|\theta_{k},k=1,...,K\right\}\notag\\
=&\sum_{k=1}^{K}\mathbb{E}\left[\sup _{\theta_{i} \in \widetilde{\Theta_{i}}}J_{i}^{(\theta_{i}, \theta_{-i,\tau_K})}(t)-J_{i}^{(\theta_{i,\tau_K}, \theta_{-i,\tau_K})}(t)\biggm|\tau_K=k, \theta_{k},k=1,...,K\right] \notag\\
&P\bigg\{\tau_K=k\biggm|\theta_{k},k=1,...,K\bigg\}\notag\\
=&\sum_{k=1}^{K}\mathbb{E}\left[\sup _{\theta_{i} \in \widetilde{\Theta_{i}}}J_{i}^{(\theta_{i}, \theta_{-i,k})}(t)-J_{i}^{(\theta_{i,k}, \theta_{-i,k})}(t)\biggm|\tau_K=k, \theta_{k},k=1,...,K\right]\notag\\
&P\bigg\{\tau_K=k\biggm|\theta_{k},
k=1,2,...,K\bigg\}.
\label{conditionalexpectation}
\end{align}
Noting that $\tau_{K}$ is independent of $\theta_{k}$, it follows that $P\bigg\{\tau_K=k\biggm|\theta_{k},k=1,...,K\bigg\}=P\bigg\{\tau_K=k\bigg\}$.
It's easy to see that
\begin{align}
&\mathbb{E}\left[\sup _{\theta_{i} \in \widetilde{\Theta_{i}}}J_{i}^{(\theta_{i}, \theta_{-i,k})}(t)-J_{i}^{(\theta_{i,k}, \theta_{-i,k})}(t)\biggm|\tau_K=k, \theta_{k},k=1,...,K\right]\notag\\
=&\mathbb{E}\left[\sup _{\theta_{i} \in \widetilde{\Theta_{i}}}J_{i}^{(\theta_{i}, \theta_{-i,k})}(t)-J_{i}^{(\theta_{i,k}, \theta_{-i,k})}(t)\biggm|\theta_{k},k=1,...,K\right].\label{independence}
\end{align}
From the definition of total reward function, we have $$\mathbb{E}\left[\sup _{\theta_{i} \in \widetilde{\Theta_{i}}}J_{i}^{(\theta_{i}, \theta_{-i,k})}(t)-J_{i}^{(\theta_{i,k}, \theta_{-i,k})}(t)\biggm|\theta_{k},k=1,...,K\right]
=\sup _{\theta_{i} \in \widetilde{\Theta_{i}}}J_{i}^{(\theta_{i}, \theta_{-i,k})}(t)-J_{i}^{(\theta_{i,k}, \theta_{-i,k})}(t)
.$$
This together with (\ref{independence}) gives
\begin{equation}\nonumber
\begin{aligned}
&\mathbb{E}\left[\sup _{\theta_{i} \in \widetilde{\Theta_{i}}}J_{i}^{(\theta_{i}, \theta_{-i,k})}(t)-J_{i}^{(\theta_{i,k}, \theta_{-i,k})}(t)\biggm|\tau_K=k, \theta_{k},k=1,...,K\right]\\
=&\sup _{\theta_{i} \in \widetilde{\Theta_{i}}}J_{i}^{(\theta_{i}, \theta_{-i,k})}(t)-J_{i}^{(\theta_{i,k}, \theta_{-i,k})}(t).
\end{aligned}
\end{equation}
Combining the above inequality with (\ref{conditionalexpectation}) gives
\begin{align}
&\mathbb{E}\left[\sup_{\theta_{i}\in \widetilde{\Theta_{i}}}J_{i}^{(\theta_{i}, \theta_{-i,\tau_K})}(t)-J_{i}^{(\theta_{i,\tau_K}, \theta_{-i,\tau_K})}(t)\biggm|\theta_{k},k=1,...,K\right] \notag\\
=&\sum_{k=1}^{K}\left(\sup_{\theta_{i} \in \widetilde{\Theta_{i}}}J_{i}^{(\theta_{i}, \theta_{-i,k})}(t)-J_{i}^{(\theta_{i,k}, \theta_{-i,k})}(t)\right)P\bigg\{\tau_K=k\bigg\} \notag\\
=&\frac{\sum_{k=1}^{K}\gamma_{k}\left(\sup \limits_{\theta_{i} \in \widetilde{\Theta_{i}}}J_{i}^{(\theta_{i}, \theta_{-i,k})}(t)-J_{i}^{(\theta_{i,k}, \theta_{-i,k})}(t)\right)}{\sum_{k=1}^{K}\gamma_{k}}.\notag
\end{align}
Noting that $\left(\left(\pi_{\theta_{i,k}}\right)_{i=1}^{N}\right)_{k=1}^{\infty}$ is a $\gamma_{k}$-weighted $\epsilon_{k}$-Nash equilibrium of the game $\Gamma$, it follows that
\begin{equation}\nonumber
\begin{aligned}
&\sup_{i \in \mathcal{N}}
\mathbb{E}\left[\sup_{\theta_{i}\in \widetilde{\Theta_{i}}}J_{i}^{(\theta_{i}, \theta_{-i,\tau_K})}(t)-J_{i}^{(\theta_{i,\tau_K}, \theta_{-i,\tau_K})}(t)\biggm|\theta_{k},k=1,...,K\right]\\
=& \sup_{i \in \mathcal{N}}\frac{\sum_{k=1}^{K}\gamma_{k}\left(\sup \limits _{\theta_{i} \in \widetilde{\Theta_{i}}}J_{i}^{(\theta_{i}, \theta_{-i,k})}(t)-J_{i}^{(\theta_{i,k}, \theta_{-i,k})}(t)\right)}{\sum_{k=1}^{K}\gamma_{k}}
\leqslant \epsilon_{K},
\end{aligned}
\end{equation}
then the theorem holds.
$\hfill\blacksquare$
\\ \hspace*{\fill} \\

\textbf{Proof of Theorem \ref{existencenash}:}
By Assumption \ref{assumption1} and  Lemma 3.1 in \cite{Philip}, we know that SVI$(F,\Theta)$ has a solution. By Assumption \ref{assumption0} and Assumption \ref{assumption3}, Lemma \ref{equivalenceofnefne} implies the equivalence between the Nash equilibrium problem and SVI$(F,\Theta)$. Therefore, the game has a Nash equilibrium.
$\hfill\blacksquare$
\\ \hspace*{\fill} \\

\textbf{Proof of Lemma \ref{FLIPSCHITZ}:}
For agent $i$, by Assumptions
\ref{assumption0}-\ref{assumption1}, Lemma $3.2$ in \cite{Zhang} and Proposition \ref{policygradient}, we have
\begin{equation}\nonumber
\begin{aligned}
\nabla_{\theta_{i}} V_{i}^{\theta}(s,t)=\sum_{l=t}^{\infty} \sum_{\tau=0}^{\infty} \gamma^{l+\tau-t} \int & r_{i}\left(s(l+\tau), a(l+\tau)\right)  \nabla_{\theta_{i}} \log \pi_{\theta_{i}} \left(a_{i}(l) \mid s(l)\right)\\
& \rho_{\theta, t: l+\tau}  \mathrm{d} s_{t+1} ... \mathrm{d} s_{l+\tau} \mathrm{d} a_{t+1}  ... \mathrm{d} a_{l+\tau},
\end{aligned}
\end{equation}
where $
\rho_{\theta, t:l+\tau}=\left[\prod_{u=t}^{l+\tau-1} \rho\left(s(u+1) \mid s(u), a(u)\right)\right] \left[\prod_{u=t}^{l+\tau} \pi_{\theta}\left(a(u) \mid s(u)\right)\right]
$
is the probability density function of the trajectory $\big(s(t), a(t),..., s(l+\tau), a(l+\tau)
\big)$ and the integral is  over all trajectories. Then, it follows that
\begin{align}\label{differencenabla}
&\left\|\nabla_{{\theta}_{i}} V_{i}^{\theta^{1}}(s,t)-\nabla_{{\theta}_{i}} V_{i}^{\theta^{2}}(s,t)\right\| \notag\\
=&\Bigg\| \sum_{l=t}^{\infty} \sum_{\tau=0}^{\infty} \gamma^{l+\tau-t} \Bigg(\int \left\{\nabla_{\theta_{i}} \log \pi_{\theta^{1}_{i}}\left(a_{i}(l) \mid s(l)\right)-\nabla_{\theta_{i}} \log \pi_{\theta^{2}_{i}}\left(a_{i}(l)  \mid s(l)\right)\right\}\notag\\
&\times r_{i}\left(s(l+\tau), a(l+\tau)\right)\rho_{\theta^{1}, t: l+\tau} +\int r_{i}\left(s(l+\tau), a(l+\tau)\right)  \nabla_{\theta_{i}} \log \pi_{\theta^{2}_{i}}\left(a_{i}(l) \mid s(l)\right)\notag\\
& \times \left(\rho_{\theta^{1}, t: l+\tau}-\rho_{\theta^{2}, t: l+\tau}\right) \Bigg) \mathrm{d} s_{t+1} ... \mathrm{d} s_{l+\tau} \mathrm{d} a_{t+1}  ... \mathrm{d} a_{l+\tau} \Bigg\| \notag\\
\leqslant &\sum_{l=t}^{\infty} \sum_{\tau=t}^{\infty} \gamma^{l+\tau-t} \Big(\int \left\|\nabla_{\theta_{i}} \log \pi_{\theta^{1}_{i}}\left(a_{i}(l)  \mid s(l)\right)-\nabla_{\theta_{i}} \log \pi_{\theta^{2}_{i}}\left(a_{i}(l) \mid s(l)\right)\right\| \notag\\
&\times \left|r_{i}\left(s(l+\tau), a(l+\tau)\right)\right| \rho_{\theta^{1}, t: l+\tau} \mathrm{d} s_{t+1} ... \mathrm{d} s_{l+\tau} \mathrm{d} a_{t+1}  ... \mathrm{d} a_{l+\tau}
+\int\left|r_{i}\left(s(l+\tau), a(l+\tau)\right)\right| \notag\\
&\times \left\|\nabla_{\theta_{i}} \log \pi_{\theta^{2}_{i}}\left(a_{i}(l) \mid s(l)\right)\right\|
\left|\rho_{\theta^{1}, t: l+\tau}-\rho_{\theta^{2}, t: l+\tau}\right|\Big) \mathrm{d} s_{t+1} ... \mathrm{d} s_{l+\tau} \mathrm{d} a_{t+1}  ... \mathrm{d} a_{l+\tau}.
\end{align}
Denote
\begin{align}
I_{1}=&\int\left|r_{i}\left(s(l+\tau), a(l+\tau)\right)\right| \left\|\nabla_{\theta_{i}} \log \pi_{\theta^{1}_{i}}\left(a_{i}(l)  \mid s(l)\right)-\nabla_{\theta_{i}} \log \pi_{\theta^{2}_{i}}\left(a_{i}(l) \mid s(l)\right)\right\| \notag\\
&\times \rho_{\theta^{1}, t: l+\tau} \mathrm{d} s_{t+1} ... \mathrm{d} s_{l+\tau} \mathrm{d} a_{t+1}  ... \mathrm{d} a_{l+\tau},\label{I11}\\
I_{2}=&\int\left|r_{i}\left(s(l+\tau), a(l+\tau)\right)\right| \left\|\nabla_{\theta_{i}} \log \pi_{\theta^{2}_{i}}\left(a_{i}(l) \mid s(l)\right)\right\|\notag\\
& \times
\left|\rho_{\theta^{1}, t: l+\tau}-\rho_{\theta^{2}, t: l+\tau}\right| \mathrm{d} s_{t+1} ... \mathrm{d} s_{l+\tau} \mathrm{d} a_{t+1}  ... \mathrm{d} a_{l+\tau}.\label{I22}
\end{align}
For the term $I_{1}$, by  Assumption \ref{assumption1}, we have
\begin{equation}\label{I1}
    I_{1} \leqslant U_{R}  L_{\Theta}  \left\|\theta^{1}_{i}-\theta^{2}_{i}\right\|.
\end{equation}
For the term $I_{2}$, denote $\mathcal{U}_{l+\tau}=\{u: u=t, ..., l+\tau\}$. From the definitions of  $\rho_{\theta^{1}, t: l+\tau}$ and $\rho_{\theta^{2}, t:l+\tau}$, we have
\begin{align}\label{density}
\rho_{\theta^{1}, t: l+\tau}-\rho_{\theta^{2}, t: l+\tau}=&\left[\prod_{u=t}^{l+\tau-1} \rho\left(s(u+1) \mid s(u), a(u)\right)\right] \notag\\
& \times\Bigg[\prod_{u \in \mathcal{U}_{l+\tau}} \pi_{\theta^{1}}\left(a(u) \mid s(u)\right)-
\prod_{u \in \mathcal{U}_{l+\tau}} \pi_{\theta^{2}}\left(a(u) \mid s(u)\right)\Bigg].
\end{align}
From Assumption \ref{assumption2} and Taylor expansion of $\prod_{u \in \mathcal{U}_{l+\tau}} \pi_{\theta}\left(a(u) \mid s(u)\right)$ near $\theta=\theta^{1}$, there exists some $\widetilde{\theta}=\lambda \theta^{1}+(1-\lambda) \theta^{2} $ and $\lambda \in[0,1]$, such that
\begin{align}\nonumber
&\left|\prod_{u \in \mathcal{U}_{l+\tau}} \pi_{\theta^{1}}\left(a(u) \mid s(u)\right)-\prod_{u \in \mathcal{U}_{l+\tau}} \pi_{\theta^{2}}\left(a(u) \mid s(u)\right)\right| \notag\\
=&\left|\left(\theta^{1}-\theta^{2}\right)^{\top}\left[\sum_{m \in \mathcal{U}_{l+\tau}} \nabla \pi_{\widetilde{\theta}}\left(a(m) \mid s(m)\right) \prod_{u \in \mathcal{U}_{l+\tau}, u \neq m} \pi_{\widetilde{\theta}}\left(a(u) \mid s(u)\right)\right]\right| \notag\\
= &\left|\left(\theta^{1}-\theta^{2}\right)^{\top}\left[\sum_{m \in \mathcal{U}_{l+\tau}} \nabla \log \pi_{\widetilde{\theta}}\left(a(m) \mid s(m)\right) \prod_{u \in \mathcal{U}_{l+\tau}} \pi_{\widetilde{\theta}}\left(a(u) \mid s(u)\right)\right]\right| \notag\\
\leqslant &\left\|\theta^{1}-\theta^{2}\right\| \sum_{m \in \mathcal{U}_{l+\tau}}\left\|\nabla \log \pi_{\widetilde{\theta}}\left(a(m) \mid s(m)\right)\right\| \prod_{u \in \mathcal{U}_{l+\tau}} \pi_{\tilde{\theta}}\left(a(u) \mid s(u)\right) \notag\\
\leqslant &\left\|\theta^{1}-\theta^{2}\right\| (l+\tau-t+1)  B_{\Theta} \sqrt{N}  \prod_{u \in \mathcal{U}_{l+\tau}} \pi_{\widetilde{\theta}}\left(a(u) \mid s(u)\right)\notag ,
\end{align}
 where $\nabla \log \pi_{\widetilde{\theta}}\left(a(u) \mid s(u)\right)=(\nabla_{\theta_{i}} \log \pi_{\theta_{i}}(a_{i}(u) \mid s(u))_{i=1}^{N}$.
 From (\ref{I22}), (\ref{density}), the above inequality and Assumption \ref{assumption2}, it follows that
\begin{align*}
I_{2}  \leqslant &\left\|\theta^{1}-\theta^{2}\right\|  U_{R} \sqrt{N}  B_{\Theta}^{2} \int\left[\prod_{u=t}^{l+\tau-1} \rho\left(s(u+1) \mid s(u), a(u)\right)\right] (l+\tau-t+1)\notag\\
 &\times \prod_{u \in \mathcal{U}_{l+\tau}} \pi_{\widetilde{\theta}}\left(a(u) \mid s(u)\right)  \mathrm{d} s_{t+1} ... \mathrm{d} s_{l+\tau} \mathrm{d} a_{t+1}  ... \mathrm{d} a_{l+\tau} \\
=&\left\|\theta^{1}-\theta^{2}\right\|  U_{R} \sqrt{N} B_{\Theta}^{2} (l+\tau-t+1).
\end{align*}
By (\ref{differencenabla}), (\ref{I1}) and the above inequality, we have
\begin{align*}
\left\|\nabla_{{\theta}_{i}} V_{i}^{\theta^{1}}(s,t)-\nabla_{{\theta}_{i}} V_{i}^{\theta^{2}}(s,t)\right\| &\leqslant \sum_{l=t}^{\infty} \sum_{\tau=0}^{\infty}\gamma^{l+\tau-t}\left(I_{1}+I_{2}\right)\\
 \leqslant & \sum_{l=t}^{\infty} \sum_{\tau=0}^{\infty} \gamma^{l+\tau-t}  \Bigg(U_{R} L_{\Theta}  \left\|\theta^{1}_{i}-\theta^{2}_{i}\right\|\\
&+(l+\tau-t+1) U_{R}  \sqrt{N}  B_{\Theta}^{2} \left\|\theta^{1}-\theta^{2}\right\|\Bigg)\\
\leqslant & \frac{1}{(1-\gamma)^{2}}  U_{R}  L_{\Theta} \left\|\theta^{1}_{i}-\theta^{2}_{i}\right\|+\frac{1+\gamma}{(1-\gamma)^{3}} U_{R}  \sqrt{N} B_{\Theta}^{2} \left\|\theta^{1}-\theta^{2}\right\| \ .
\end{align*}
This together with Assumption \ref{assumption1} and Proposition \ref{policygradient} gives
\begin{align*}
\left\|\nabla_{{\theta}_{i}} J_{i}^{\theta^{1}}(t)-\nabla_{{\theta}_{i}} J_{i}^{\theta^{2}}(t)\right\|
=&\left\|\int_{s \in \mathcal{S}}\left(\nabla_{{\theta}_{i}} V_{i}^{\theta^{1}}(s,t)-\nabla_{{\theta}_{i}} V_{i}^{\theta^{2}}(s,t)\right)\rho_{t}(s) \mathrm{d}s\right\|\\
\leqslant & \int_{s \in \mathcal{S}}
\left\|\nabla_{{\theta}_{i}} V_{i}^{\theta^{1}}(s,t)-\nabla_{{\theta}_{i}} V_{i}^{\theta^{2}}(s,t)\right\| \rho_{t}(s) \mathrm{d}s\\
\leqslant &\frac{1}{(1-\gamma)^{2}}  U_{R}  L_{\Theta} \left\|\theta^{1}_{i}-\theta^{2}_{i}\right\|+\frac{1+\gamma}{(1-\gamma)^{3}} U_{R}  \sqrt{N} B_{\Theta}^{2} \left\|\theta^{1}-\theta^{2}\right\| \ .
\end{align*}
Hence, we have
\begin{align*}
&\sum_{i=1}^{N} \left\|\nabla_{{\theta}_{i}} J_{i}^{\theta^{1}}(t)-\nabla_{{\theta}_{i}} J_{i}^{\theta^{2}}(t)\right\|^{2}\\
\leqslant & 2\sum_{i=1}^{N}\left(\frac{ U_{R} \cdot L_{\Theta}}{(1-\gamma)^{2}} \right)^{2} \left\|\theta^{1}_{i}-\theta^{2}_{i}\right\|^{2} +2 \left(\frac{1+\gamma}{(1-\gamma)^{3}}  U_{R} \sqrt{N}  B_{\Theta}^{2}\right)^{2} \left\|\theta^{1}-\theta^{2}\right\|^{2}\\
\leqslant &2\left(\frac{ U_{R}  L_{\Theta}}{(1-\gamma)^{2}} \right)^{2} \left\|\theta^{1}-\theta^{2}\right\|^{2} +2 \left(\frac{1+\gamma}{(1-\gamma)^{3}}  U_{R} \sqrt{N} B_{\Theta}^{2}\right)^{2} \left\|\theta^{1}-\theta^{2}\right\|^{2}\\
=&\left[2\left(\frac{ U_{R}  L_{\Theta}}{(1-\gamma)^{2}} \right)^{2}+2 \left(\frac{1+\gamma}{(1-\gamma)^{3}}  U_{R}  \sqrt{N} B_{\Theta}^{2}\right)^{2}\right]
\left\|\theta^{1}-\theta^{2}\right\|^{2},
\end{align*}
that is,
$$\left\|F(\theta^{1})-F(\theta^{2})\right\| \leqslant  L\left\|\theta^{1}-\theta^{2}\right\|.$$
$\hfill\blacksquare$

\vskip 0.2in
\bibliography{policygradientcomplete}

\end{document}